\documentclass[twoside,11pt]{article}

\usepackage[abbrvbib,preprint]{jmlr2e}

\usepackage{comment}

 \usepackage{float}
 \usepackage{algorithm}
\usepackage{algpseudocode}
\usepackage{subcaption}
\usepackage{tikz}
\usepackage{xcolor}
\usepackage{hyperref}
\usepackage{multirow}
\usepackage{mathtools}
\usepackage{booktabs}
\usepackage{xr}
\usepackage[normalem]{ulem}
\usepackage{natbib}
\newtheorem{assumption}{Assumption}

\newfloat{algorithmframework}{htbp}{loa}
\floatname{algorithmframework}{Algorithm Framework}

\usepackage{lastpage}
\jmlrheading{}{2025}{1-\pageref{LastPage}}{}{}{}{Jianwei Shi, Sameh Abdulah, Ying Sun, and Marc G. Genton}

\ShortHeadings{Scalable Asynchronous Federated Modeling for Spatial Data}{Jianwei Shi, Sameh Abdulah, Ying Sun, and Marc G. Genton}
\firstpageno{1}

\begin{document}
\title{Scalable Asynchronous Federated Modeling for Spatial Data}
%\author{Jianwei Shi$^1$, Sameh Abdulah$^2$, Ying Sun$^1$, and Marc G. Genton$^1$\\%\thanks{
    %The authors gratefully acknowledge \textit{please remember to list all relevant funding sources in the unblinded version}}\hspace{.2cm}\\
 % $^1$Statistics Program,\\
%  $^2$Applied Mathematics and Computational Sciences Program,\\
%King Abdullah University of Science and Technology,\\ Thuwal, 23955, Saudi Arabia.}

\author{
\hspace{-3mm}  \name Jianwei Shi$^1$ \email jianwei.shi@kaust.edu.sa 
  \AND
  \name Sameh Abdulah$^2$ \email sameh.abdulah@kaust.edu.sa 
  \AND
  \name Ying Sun$^1$ \email ying.sun@kaust.edu.sa 
  \AND
  \name Marc G. Genton$^1$ \email marc.genton@kaust.edu.sa \\
  \addr $^1$Statistics Program \\
  \addr $^2$Applied Mathematics and Computational Sciences Program\\
  King Abdullah University of Science and Technology, 
  Thuwal, 23955, Saudi Arabia
}

\editor{}

  \maketitle

\begin{abstract}
Spatial data are central to applications such as environmental monitoring and urban planning, but are often distributed across devices where privacy and communication constraints limit direct sharing. Federated modeling offers a practical solution that preserves data privacy while enabling global modeling across distributed data sources. For instance, environmental sensor networks are privacy- and bandwidth-constrained, motivating federated spatial modeling that shares only privacy-preserving summaries to produce timely, high-resolution pollution maps without centralizing raw data. However, existing federated modeling approaches either ignore spatial dependence or rely on synchronous updates that suffer from stragglers in heterogeneous environments. This work proposes an asynchronous federated modeling framework for spatial data based on low-rank Gaussian process approximations. The method employs block-wise optimization and introduces strategies for gradient correction, adaptive aggregation, and stabilized updates. We establish linear convergence with explicit dependence on staleness, a result of %\sout{ independent interest to the optimization community} 
{standalone theoretical significance}. Moreover, numerical experiments demonstrate that the asynchronous algorithm achieves synchronous performance under balanced resource allocation and significantly outperforms it in heterogeneous settings, showcasing superior robustness and scalability.
\end{abstract}


  \section{Introduction}

  Spatial data consist of information associated with specific locations and play a crucial role in various applications, including environmental monitoring \citep{haining1993spatial}, climate modeling \citep{daly2006guidelines}, public health  \citep{waller2004applied}, urban planning \citep{duhr2012role}, and transportation analysis \citep{miller1999potential}. These data are often distributed across multiple locations and collected from diverse devices, including satellites, sensors, and mobile devices. As a result, privacy and communication constraints usually restrict the direct sharing of raw data \citep{mckenna2018role}. Federated learning \citep{konevcny2016federated,mcmahan2017communication, graser2022role}, first introduced in the Machine Learning (ML) community, provides a promising solution by enabling collaborative analysis while preserving data privacy and minimizing the need for data exchange between different workers. Typically, a central server coordinates the learning process, while local workers (devices) perform computations on their data and communicate only some summary statistics, rather than sharing the raw data itself, %\sout{cannot be used to recover the original data} 
  {thereby preserving privacy.}

{Gaussian processes (GPs) are widely used in modeling spatial data due to their flexibility in capturing spatial dependence \citep{stein1999interpolation}. We use the term federated modeling to describe the collaborative estimation of a GP model across multiple local spatial data sources without sharing raw observations.}
  A key challenge in federated %\sout{learning} 
  {modeling} for spatial data %\sout{(federated modeling)} 
  arises from the strong dependence among spatial observations. %\sout{, which are modeled as realizations of a Gaussian process} 
Unlike standard federated learning setups, which assume independent local samples among workers, spatial datasets exhibit cross-worker correlation that must be explicitly accounted for. Consequently, the log-likelihood function cannot be decomposed into a simple sum of local functions computed from each worker’s data, as is typically required for federated optimization.

  To address this challenge, the literature provides two main classes of approaches. The first class \citep{yin2020fedloc,achituve2021personalized,yu2022federated,guo2022federated,yue2024federated, kontoudis2024scalable} ignores spatial dependence, treating the global log-likelihood as the sum of local log-likelihoods from different workers. %\sout{This assumption is typically violated in the presence of spatial correlation and is prone to biased estimates with overly confident uncertainty.}  
  To better handle the spatial dependence, the second class leverages low-rank approximations  %\sout{, including knots-based methods} 
%\sout{, eigenfunction expansions based on Mercer's theorem }
to approximate either the prior dependence \citep{shi2025decentralized,katzfuss2017parallel} or posterior dependence \citep{hoang2016distributed,chung2024federated,gao2024federated}. %\sout{ and tile low-rank (TLR) factorizations that exploit off-diagonal low-rank structure }
In this class, spatial dependence among workers, whether in the prior or posterior, is captured through a shared low-rank random vector, which induces conditional independence among workers and enables the log-likelihood to be decoupled once conditioned on this vector. Since approximating the prior is more principled and interpretable than approximating the posterior {(See Section \ref{sec:literature_federate_spatial} for more details)}. %\sout{as the former can correspond to a valid Gaussian process. For these reasons,} 
This article focuses on low-rank approximations to the prior, leaving posterior approximations for future investigation.

{In this work, we emphasize that low-rank models are particularly suitable for federated modeling for spatial data. On one hand, the conditional log-likelihood on the low-rank random vector can be decoupled, yielding a new objective that can be expressed as a sum of local contributions, thereby facilitating federated optimization (see Section~\ref{sec:low-rank} and also \cite{shi2025decentralized}). On the other hand, the low-rank approximation preserves the dependence structure among workers to some extent and is thus closer to the original model.}
  %\sout{In this work} 
  {Specifically}, we theoretically establish that the low-rank approximation yields an asymptotically smaller dimension-normalized Kullback–Leibler (KL) divergence from the original distribution compared to the naive independence model (see Proposition \ref{prop:kl_IDVSLOW} in Section \ref{sec:low-rank}). Consistent with this theory, the numerical results in Section \ref{sec:simulation:comparison_lowrank_independent} demonstrate that the low-rank model {indeed achieves a smaller KL divergence} and provides more accurate and robust parameter estimates than the independence model.  

% let's use federated modeling from here instead of federated learning
  However, to the best of our knowledge, existing federated modeling methods for spatial data based on low-rank models rely on synchronous updates for parameter estimation \citep{shi2025decentralized,gao2024federated}. In each iteration of the optimization, the server must wait for all local workers to complete their computations and communication before proceeding. Still, some machines or devices have far less compute, slower networks, or larger local datasets, creating imbalances that produce stragglers and delay the overall modeling process. To mitigate this inefficiency, asynchronous federated learning enables the server to process and aggregate partial updates from clients as soon as they become available, without waiting for all participants to complete their updates, thereby reducing idle wait times and increasing system throughput \citep{xu2023asynchronous,xu2024fedfa}. %\sout{The performance of asynchronous federated learning can be significantly impacted by the staleness of local computations and the heterogeneity of local functions} \citep{fraboni2023general}.  \sout{To improve the efficiency of asynchronous federated modeling, various strategies have been proposed }\citep{xie2019asynchronous,wang2024tackling}. \sout{, including adaptively adjusting the learning rate based on the staleness of local updates}\citep{xie2019asynchronous}\sout{, or leveraging historical local information to reduce the bias in server updates caused by heterogeneity} \citep{wang2024tackling}. \textcolor{blue}{(Deleted to avoid repetition with Section~2.)}

 Although significant progress has been made in asynchronous federated learning { (see, e.g., \cite{fraboni2023general,xie2019asynchronous,wang2024tackling}; see also Section~\ref{sec:asynchronous_federated_learning} for details)}, these methods are not directly applicable to spatial low-rank models. The fundamental limitation lies in the fact that existing asynchronous approaches are primarily designed for loss functions expressed as a summation of local objectives, with parameters updated simultaneously. By contrast, the loss function associated with spatial low-rank models exhibits a more intricate structure, in which the parameters can be partitioned into multiple blocks \citep{shi2025decentralized}. For certain blocks, closed-form solutions exist, provided the remaining blocks are fixed. Consequently, block-wise optimization is more suitable for such models than simultaneous updates, ensuring computational efficiency.

  This work introduces a novel asynchronous federated learning method for spatial low-rank models. The proposed method adopts the block-wise optimization scheme of its synchronous counterpart while addressing two challenges posed by asynchronous updates. The first challenge arises on the worker side: a worker may receive multiple candidate vectors before it is ready for the next computation, creating ambiguity in the choice of parameter vector and the identification of the corresponding block’s local quantity. To resolve this, each parameter vector is augmented with an iteration index and a block label that specifies the associated local quantity. Each worker maintains a buffer to store the received augmented vectors, with the relative order of updates preserved for each block label, ensuring the most recent parameter vector is used in computation.
The second challenge arises on the server side: naive averaging or summation of local quantities, as employed in the synchronous algorithm, can lead to instability and slower convergence due to the asynchronous arrival and varying staleness of updates. To mitigate this issue, we develop three strategies: (i) local gradient correction, which incorporates Hessian information to compensate for inconsistencies in local gradients; (ii) adaptive aggregation, which dynamically adjusts contribution weights according to the staleness of local quantities; and (iii) a moving average of historical vectors, which improves stability by smoothing fluctuations in updates.

  We establish the theoretical guarantee for the convergence of the proposed asynchronous algorithm. In particular, {under some reasonable conditions,} we show that the algorithm achieves linear convergence to the optimal solution when the step size is chosen appropriately, with an explicit characterization of its dependence on staleness. When the effect of average staleness can be neglected, the step size should scale inversely with the maximum staleness across workers to ensure linear convergence, and the resulting convergence rate improves proportionally with the inverse of this maximum staleness. To the best of our knowledge, both the algorithm and its theoretical foundation have not been previously explored in the optimization literature, and thus our results may be of %\sout{ independent interest} 
  {standalone theoretical significance}. Further details are provided in Section~\ref{sec:theory}. 
  
  The effectiveness of the proposed method is further demonstrated through numerical experiments. The results indicate that the asynchronous algorithm achieves comparable, and in some cases slightly inferior, performance to the synchronous version when computational resources are well balanced. However, in scenarios with substantial computational imbalance, arising from heterogeneous processing power or unequal local sample sizes, the asynchronous algorithm outperforms its synchronous counterpart. This advantage is observed across diverse parameter settings, sample sizes, number of machines, and data partitioning schemes, underscoring the robustness and practical utility of the proposed approach.

  The contributions of this article are fourfold. First, we provide both theoretical and numerical comparisons between the independence model and the low-rank model, highlighting the advantages of the latter in terms of approximation quality and parameter estimation. Second, we propose a novel asynchronous algorithm designed explicitly for spatial low-rank models. Third, we establish theoretical convergence guarantees for the proposed asynchronous algorithm, which may also be of standalone theoretical significance. Fourth, we conduct extensive numerical experiments to evaluate the performance of the algorithm, demonstrating its robustness and practical benefits across diverse computational scenarios.

{ Section~\ref{sec:background} reviews spatial federated learning, contrasts independence-based and low-rank objectives, and surveys asynchronous methods with focus on staleness and objective heterogeneity. Section~\ref{sec:low-rank} reviews federated low-rank spatial modeling and derives the summation-form objective that enables distributed optimization. Section~\ref{sec:asynchronous} introduces our asynchronous federated algorithm, tailored to the block structure of the model, which details worker/server buffering, staleness-aware aggregation, local gradient correction, and moving-average stabilization. Section~\ref{sec:theory} establishes linear convergence with explicit dependence on staleness and clarifies the step-size scaling. Section~\ref{sec:simulation} presents numerical experiments comparing independence versus low-rank models and asynchronous versus synchronous updates under heterogeneous compute and data settings. The paper concludes with a summary and future directions, while the Appendix provides derivations, proofs, and more experimental results.}

 \section{Background}
   \label{sec:background}
  Federated learning (FL) has gained significant attention as a privacy-preserving approach for distributed model training. While conventional FL methods assume independent distributed data across workers, spatial data exhibit strong spatial correlations, posing unique challenges for standard federated optimization techniques. In this section, we discuss relevant works on federated learning for spatial data and asynchronous federated learning methods.

      \subsection{Federated Learning for Spatial Data}\label{sec:literature_federate_spatial}
  Handling spatial data in federated learning frameworks requires specialized techniques to account for spatial correlations. The most straightforward approach is to ignore spatial dependence and assume independence across clients, leading to a log-likelihood function that is the sum of local log-likelihood functions. Works adopting this method typically differ in their choice of local functions and optimization techniques for parameter estimation. \cite{yin2020fedloc} directly use the log-likelihood function based on local data as the local function. However, when the local sample size is large, the computational cost becomes significant. To mitigate this, \cite{achituve2021personalized}, \cite{guo2022federated}, and \cite{yu2022federated} employ low-rank approximation-based local functions, whereas \cite{yue2024federated} keep the local function unchanged but reduce computational overhead by applying stochastic gradient methods on small mini-batches of local data. Furthermore, \cite{kontoudis2024scalable} extend this approach to decentralized communication networks.

  To better capture spatial correlations among local workers, low-rank approximations have been widely explored, striking a balance between computational efficiency and the preservation of essential spatial structures. These approximations can be broadly categorized into two types. The first type approximates the prior distribution by employing a shared low-rank structure to directly model spatial dependencies \citep{katzfuss2017parallel,shi2025decentralized}. The second type approximates the posterior distribution by using low-rank representations to capture conditional spatial dependencies given the observed data; these representations are then optimized via variational inference \citep{gal2014distributed,hoang2016distributed}. Recent works \citep{chung2024federated,gao2024federated} extend this second approach to the federated learning of multi-output Gaussian processes. The main advantage of the first type is that it corresponds to a valid Gaussian process if the low-rank model is appropriately chosen, making the model more interpretable and theoretically grounded. In contrast, the second type has the advantage of providing a variational lower bound on the original log-likelihood function, which can help mitigate overfitting, especially when the low-rank structure itself is learned and introduces a large number of additional parameters \citep{bui2017unifying}. In this article, we focus on the first type of approximation and leave the exploration of the second type for future work.

  \subsection{Asynchronous Federated Learning}\label{sec:asynchronous_federated_learning}
  Most existing federated learning methods rely on synchronous updates, where the central server waits for all workers to complete their local computations before aggregating the results. This synchronization requirement can introduce significant inefficiencies, particularly in heterogeneous environments where devices vary in terms of computational power, network conditions, or data sizes. To address these limitations, recent studies have explored asynchronous federated learning strategies. A comprehensive review of these approaches can be found in the survey paper by \cite{xu2023asynchronous}.

  The performance of asynchronous federated learning can be significantly impacted by the staleness of local computations and the heterogeneity of local functions \citep{fraboni2023general}. Staleness refers to the use of outdated global parameters by workers during local updates, which can lead to slower convergence. Heterogeneity arises when local functions differ across workers. In this case, server updates based on a subset of workers may be biased, deviating from the results that would be obtained if all workers contributed. To mitigate the issues from staleness and heterogeneity, various strategies have been proposed \citep{xie2019asynchronous,nguyen2022federated,xu2024fedfa,wang2024tackling,zakerinia2024communication,zang2024efficient}, 
  For example,  \cite{xie2019asynchronous} propose staleness-aware algorithms that adaptively adjust the learning rate based on the staleness of local updates. \cite{wang2024tackling} leverages historical local information to reduce the bias in server updates caused by heterogeneity. \cite{zang2024efficient} address both staleness and heterogeneity simultaneously through an efficient update scheme.

  \section{Federated Low-Rank Spatial Modeling}
  \label{sec:low-rank}

  This section provides an overview of spatial low-rank models. It discusses how they can be applied within the constraints of federated modeling, including strategies for distributed computation, data privacy preservation, and communication-efficient parameter estimation. Herein, we consider a standard spatial regression model in which the observed data at location $\mathbf{s} \in \mathcal{S}$ are decomposed into a linear contribution from covariates, a spatially correlated latent process, and independent measurement noise, as follows:

  \begin{equation}\label{eq:spatial_data}
    z(\textbf{s})=\boldsymbol{x}^{\top}(\textbf{s})\boldsymbol{\gamma}+ w(\textbf{s})+\epsilon(\textbf{s}), \quad \textbf{s} \in \mathcal{S},
  \end{equation}
where $z(\textbf{s})$ is the observed data at location $\textbf{s}$, 
$\mathcal{S}$ is the spatial domain, 
$\boldsymbol{x}\in\mathbb{R}^p$ is the covariate vector, 
$\boldsymbol{\gamma}\in\mathbb{R}^p$ is the coefficient vector of the covariates, 
$w(\textbf{s})$ is the latent Gaussian process with mean zero and covariance function 
$c_{\boldsymbol{\theta}}(\textbf{s},\textbf{s}')$ parameterized by the vector $\boldsymbol{\theta}$, 
and $\epsilon(\textbf{s})$ is a Gaussian noise term with mean zero and variance $\delta^{-1}$ where $\delta$ is a positive parameter.

%The covariance function \(c_{\boldsymbol{\theta}}(\textbf{s},\textbf{s}')\) characterizes the spatial correlation between observations at locations \(\textbf{s}\) and \(\textbf{s}'\).

%In this aritcle, we use the Matérn covariance function due to its flexibility and suitability for modeling spatial processes with varying degrees of smoothness \citep{stein1999interpolation}. It is defined as:
%\begin{equation}
%    c_{\boldsymbol{\theta}}(\textbf{s},\textbf{s}') =\sigma^2 \frac{2^{1-\nu}}{\Gamma(\nu)} \left( \sqrt{2\nu} \frac{\|\textbf{s}-\textbf{s}'\|}{\beta} \right)^{\nu} K_{\nu} \left( \sqrt{2\nu} \frac{\|\textbf{s}-\textbf{s}'\|}{\beta} \right),
% \end{equation}
% where \(\sigma^2\) is the variance parameter, \(\nu\) is the smoothness parameter, \(\beta\) is the range parameter, and \(K_{\nu}\) is the modified Bessel function of the second kind.

  Under a federated data modeling setting and constraints, the data are distributed across $J$ workers, 
each owning a local dataset 
$\mathcal{D}_j = \{(\textbf{s}_{i}, z_{i}, \boldsymbol{x}_{i})\}_{i\in \mathcal{I}_j}$, 
where $z_i=z(\textbf{s}_i),\boldsymbol{x}_i=\boldsymbol{x}^{\top}(\textbf{s}_i)$ and $\mathcal{I}_j$ is the index set of data points at worker $j$ with size $n_j=|\mathcal{I}_j|$. 
The full dataset is the union of all local datasets, 
$\mathcal{D} = \bigcup_{j=1}^J \mathcal{D}_j$. 
For convenience, we define the following notations. 
At each worker $j$, define $\boldsymbol{z}_j=(z_i,i\in \mathcal{I}_j)$ as the vector of observations, 
$\boldsymbol{X}_j=(\boldsymbol{x}_{i},i\in \mathcal{I}_j)^\top\in \mathbb{R}^{n_j\times p}$ as the covariate matrix, 
$S_j=\{\textbf{s}_i,i\in \mathcal{I}_j\}$ as the set of locations. 
We also define 
$\boldsymbol{z}=(\boldsymbol{z}_{1}^\top, \ldots, \boldsymbol{z}_{J}^\top)^\top$, 
$\boldsymbol{X}=(\boldsymbol{X}_{1}^\top, \ldots, \boldsymbol{X}_{J}^\top)^\top$, 
$S=S_1\bigcup \cdots \bigcup S_J$, 
$N=\sum_{j=1}^J n_j$ as the full dataset version of $\boldsymbol{z}_j, \boldsymbol{X}_j, S_j, n_j$. 
For two sets of locations $S'$ and $S''$, define 
$\boldsymbol{C}_{\boldsymbol{\theta}}(S',S'')=\{c_{\boldsymbol{\theta}}(\textbf{s}',\textbf{s}'')\}_{\textbf{s}' \in S',\textbf{s}'' \in S''}$ 
as the (cross-)covariance matrix of the latent Gaussian process among locations in $S'$ and $S''$. 
It reduces to a vector when $S'$ or $S''$ contains only one location. 
For a square matrix $\boldsymbol{C}$, $\text{diag}(\boldsymbol{C})$ is the diagonal matrix with the diagonal elements of $\boldsymbol{C}$.

  The full log-likelihood for data $S$ is:
  \begin{equation}
    \ell(\boldsymbol{\gamma}, \boldsymbol{\theta}, \delta) := -\frac{N}{2}\log(2\pi) - \frac{1}{2}\log|\boldsymbol{C}(\boldsymbol{\theta})+\delta^{-1}\boldsymbol{I}| - \frac{1}{2}(\boldsymbol{z}-\boldsymbol{X}\boldsymbol{\gamma})^{\top}\left(\boldsymbol{C}(\boldsymbol{\theta})+\delta^{-1}\boldsymbol{I}\right)^{-1}(\boldsymbol{z}-\boldsymbol{X}\boldsymbol{\gamma}),
  \end{equation}
  where $\boldsymbol{C}(\boldsymbol{\theta}):=\boldsymbol{C}_{\boldsymbol{\theta}}(S, S)$. Due to the spatial correlations among observations across different workers, the full log-likelihood cannot be decomposed into a sum of local functions, which is a key requirement for federated spatial modeling implementations.

  To address this challenge, one approach assumes independence between cross-worker spatial correlations, resulting in a log-likelihood function that decomposes into a sum of local log-likelihoods \citep{yin2020fedloc,achituve2021personalized,guo2022federated,yu2022federated,yue2024federated,kontoudis2024scalable}:

  \begin{equation}\label{eq:ind_log_likelihood}
    \ell^\text{ind}(\boldsymbol{\gamma}, \boldsymbol{\theta}, \delta): =\sum_{j=1}^J\ell_j(\boldsymbol{\gamma}, \boldsymbol{\theta}, \delta),
  \end{equation}
  with $\ell_j(\boldsymbol{\gamma}, \boldsymbol{\theta}, \delta): =
  -\frac{n_j}{2}\log(2\pi) - \frac{1}{2}\log|\boldsymbol{C}_j(\boldsymbol{\theta})+\delta^{-1}\boldsymbol{I}| - \frac{1}{2}(\boldsymbol{z}_j-\boldsymbol{X}_j\boldsymbol{\gamma})^{\top}\left(\boldsymbol{C}_j(\boldsymbol{\theta})+\delta^{-1}\boldsymbol{I}\right)^{-1}(\boldsymbol{z}_j-\boldsymbol{X}_j\boldsymbol{\gamma})$ where
  $C_j(\boldsymbol{\theta}):=C_{\boldsymbol{\theta}}(S_j, S_j)$. However, neglecting the spatial correlations between observations across different workers results in information loss and may compromise the performance of the federated modeling approach.

  The other approach leverages low-rank models to capture spatial correlations across workers \citep{shi2025decentralized,katzfuss2017parallel,chung2024federated,gao2024federated}.  Given a set of shared knots $S^*\subset \mathcal{S}$ among workers with cardinality $m \ll N$, we consider the {knots-based} low-rank model of the form:
  \begin{equation}
    z(\textbf{s}_{i}) = \boldsymbol{x}_{i}^\top \boldsymbol{\gamma} + \boldsymbol{b}^\top(\textbf{s}_{i}; \boldsymbol{\theta}) \boldsymbol{\eta} + \widetilde{w}(\textbf{s}_{i}) + \epsilon(\textbf{s}_{i}), \quad i\in \mathcal{I}_j,\; j=1,\ldots,J,
  \end{equation}
  where $\boldsymbol{b}(\textbf{s}_i; \boldsymbol{\theta})=\boldsymbol{C}_{\boldsymbol{\theta}}(\{\textbf{s}_i\}, S^*)\boldsymbol{C}_{\boldsymbol{\theta}}(S^*, S^*)^{-1}\in \mathbb{R}^{m}$, $\boldsymbol{\eta}\in \mathbb{R}^{m}$ is the low-dimensional random vector with a multivariate normal distribution $\boldsymbol{\eta} \sim \mathcal{N}(\boldsymbol{0}_m, \boldsymbol{K}(\boldsymbol{\theta}))$ where $\boldsymbol{K}(\boldsymbol{\theta})=\boldsymbol{C}_{\boldsymbol{\theta}}(S^*, S^*)\in \mathbb{R}^{m\times m}$, and $\widetilde{w}(\textbf{s}_{i})$ is the residual term that is assumed to be independent from $\boldsymbol{\eta}$ and is also independent across different workers. According to the trade-off of local computation complexity and approximation accuracy, there are several choices for the distribution of the residual term; these alternatives will be discussed in detail later in {Remark \ref{re:residual}} %\sout{the article}. 
  With this model, the cross-covariance matrix of observations among different workers $j_1,j_2$ is $\boldsymbol{C}_{\boldsymbol{\theta}}(S_{j_1}, S^*)\boldsymbol{C}_{\boldsymbol{\theta}}(S^*, S^*)^{-1} \boldsymbol{C}_{\boldsymbol{\theta}}(S^*, S_{j_2})$, and when the knots are appropriately chosen, the cross-covariance matrix is close to the true cross-covariance matrix \citep{banerjee2008gaussian}. 

\begin{remark}[{Knot selection}]
    {Knot selection has been studied primarily in centralized settings where all data are available on a single machine. Data-dependent approaches include selecting a subset of observed locations randomly, as in the Nyström method \citep{kumar2012sampling}, optimizing knots by minimizing the KL divergence between the variational approximation and the exact posterior distribution \citep{titsias2009variational}, or placing knots according to support points \citep{song2025large}. Alternatively, domain-based approaches, such as placing knots on a regular grid or employing space-filling designs \citep{banerjee2008gaussian}, do not require access to data locations. In federated settings, where locations may not be shared, the former class requires adaptation, whereas the latter can be applied directly. For simplicity, in our numerical experiments, we adopt the latter approach and place knots on a predefined grid over the domain.}
\end{remark}

  \begin{remark}[{Low-Rank Models}]
    Several types of low-rank models have been developed for spatial data, including knots-based methods~\citep{banerjee2008gaussian}, eigenfunction expansions~\citep{greengard2022efficient}, and random Fourier features~\citep{hensman2018variational}. Each approach corresponds to a different choice of the basis function \(\boldsymbol{b}(\textbf{s}_{j}; \boldsymbol{\theta})\) and the associated covariance matrix \(\boldsymbol{K}\). Since the asynchronous algorithm to be introduced %\sout{is similar for} 
    {can be adapted to} different low-rank models, and the low-rank model based on knots-based methods is more flexible and easier to implement. This work focuses on the low-rank model based on knot-based methods.
  \end{remark}

  Although the log-likelihood function of the low-rank model (denoted as \(\ell^{\text{low-rank}}\)) still cannot be expressed as a sum of local functions, {we show in Appendix \ref{sec:appendix:additional_details} that} it is closely related to a new objective function %\sout{, with two more parameters $\boldsymbol{\mu}\in \mathbb{R}^{m}$ and $\boldsymbol{\Sigma}\in \mathbb{R}^{m\times m }$,} 
  that can be written as a summation of local functions, and a common term:
  \begin{equation}\label{eq:low-rank-obj}
    \ell^\text{low-rank}(\boldsymbol{\gamma}, \delta, \boldsymbol{\theta}) =-
    \min_{\boldsymbol{\mu}, \boldsymbol{\Sigma}} \left\{
      f(\boldsymbol{\mu}, \boldsymbol{\Sigma}, \boldsymbol{\gamma}, \delta, \boldsymbol{\theta})
      \coloneqq \sum_j f_j(\boldsymbol{\mu}, \boldsymbol{\Sigma}, \boldsymbol{\gamma}, \delta, \boldsymbol{\theta})
      + h(\boldsymbol{\mu}, \boldsymbol{\Sigma}, \boldsymbol{\theta})
      \right\},
    \end{equation}
    where  {$\boldsymbol{\mu}\in \mathbb{R}^{m}$ and $\boldsymbol{\Sigma}\in \mathbb{R}^{m\times m }$}, by denoting  $\boldsymbol{R}_j(\delta, \boldsymbol{\theta}) =(\operatorname{cov}(\widetilde{w}(\textbf{s}_{i_1}), \widetilde{w}(\textbf{s}_{i_2})))_{i_1,i_2\in \mathcal{I}_j}+\delta^{-1} \boldsymbol{I}_{n_j}$ and $\boldsymbol{B}_j(\boldsymbol{\theta}) =(\boldsymbol{b}(\textbf{s}_1; \boldsymbol{\theta}), \ldots, \boldsymbol{b}(\textbf{s}_{n_j}; \boldsymbol{\theta}))^\top\in \mathbb{R}^{n_j\times m}$,
    \begin{align*}
      f_j & := \frac{1}{2} \log \det \left[ \boldsymbol{R}_j(\delta, \boldsymbol{\theta}) \right]
      + \frac{1}{2} \operatorname{tr} \left\{ \boldsymbol{B}_j^\top(\boldsymbol{\theta}) \boldsymbol{R}_j^{-1}(\delta, \boldsymbol{\theta}) \boldsymbol{B}_j(\boldsymbol{\theta}) (\boldsymbol{\Sigma} + \boldsymbol{\mu} \boldsymbol{\mu}^\top) \right\} \\
       & \quad - (\boldsymbol{z}_j - \boldsymbol{X}_j \boldsymbol{\gamma})^\top \boldsymbol{R}_j^{-1}(\delta, \boldsymbol{\theta}) \boldsymbol{B}_j(\boldsymbol{\theta}) \boldsymbol{\mu}
      + \frac{1}{2} (\boldsymbol{z}_j - \boldsymbol{X}_j \boldsymbol{\gamma})^\top \boldsymbol{R}_j^{-1}(\delta, \boldsymbol{\theta}) (\boldsymbol{z}_j - \boldsymbol{X}_j \boldsymbol{\gamma}), \\
      h & := \frac{1}{2} \left[ \boldsymbol{\mu}^\top \boldsymbol{K}^{-1}(\boldsymbol{\theta}) \boldsymbol{\mu}
      + \operatorname{tr} \left( \boldsymbol{K}^{-1}(\boldsymbol{\theta}) \boldsymbol{\Sigma} \right)
      - \log \left( \frac{\det(\boldsymbol{\Sigma})}{\det(\boldsymbol{K}(\boldsymbol{\theta}))} \right)
      - m \right].
    \end{align*}
    {This result is a generalization of \cite{shi2025decentralized}, where the residual terms  \(\widetilde{w}(\textbf{s}_{i})\) is zero,  to a more general covariance structure for the residual.}

  %  \sout{
  %  The derivation of the new objective function extends the approach in}~\cite{shi2025decentralized}\sout{, which assumes the residual term \(\widetilde{w}(\textbf{s}_{i})\) to be zero, to a more general covariance structure for the residual. The complete derivation is provided in the supplementary material}

    \begin{remark}[{Covariance structure for the residual}]\label{re:residual}
      The distribution of the residual term \(\widetilde{w}(\textbf{s}_{i})\) can be specified in three common ways:

      \textbf{1. Predictive Process Approach}~\citep{banerjee2008gaussian}: This approach assumes no residual term, i.e., \(\widetilde{w}(\textbf{s}_{i}) = 0\), resulting in covariance matrix \(\boldsymbol{R}_j(\delta, \boldsymbol{\theta}) = \delta^{-1} \boldsymbol{I}_{n_j}\).

      \textbf{2. Modified Predictive Process Approach}~\citep{finley2009improving}: To account for the discrepancy between the full Gaussian process and its low-rank approximation, the residual \(\widetilde{w}(\textbf{s}_{i})\) is modeled as independent Gaussian noise with variance equal to the difference between the marginal variance of the full process and that of the low-rank approximation: \vspace{-10pt}
      \begin{equation}\label{eq:modified_predictive_process}
        \widetilde{w}(\textbf{s}_{i}) \stackrel{iid}{\sim} \mathcal{N}\left(0, c_{\boldsymbol{\theta}}(\textbf{s}_{i},\textbf{s}_{i}) - \boldsymbol{C}_{\boldsymbol{\theta}}(\{\textbf{s}_i\}, S^*) \boldsymbol{C}_{\boldsymbol{\theta}}(S^*, S^*)^{-1} \boldsymbol{C}_{\boldsymbol{\theta}}(S^*, \{\textbf{s}_i\})\right),\vspace{-10pt}
      \end{equation}
      leading to:
      \vspace{-10pt}
      \[\vspace{-10pt}
      \boldsymbol{R}_j(\delta, \boldsymbol{\theta}) = \operatorname{diag}\left(\boldsymbol{C}_{\boldsymbol{\theta}}(S_j, S_j) - \boldsymbol{C}_{\boldsymbol{\theta}}(S_j, S^*) \boldsymbol{C}_{\boldsymbol{\theta}}(S^*, S^*)^{-1} \boldsymbol{C}_{\boldsymbol{\theta}}(S^*, S_j)\right) + \delta^{-1} \boldsymbol{I}_{n_j}.
      \]
      Compared to the predictive process approach, this modification recovers the marginal variance of the latent Gaussian process at each location, yielding improved accuracy.

      \textbf{3. Full Local Covariance Approach}: This method preserves the full residual covariance structure within each worker. For worker \(j\), let \(\boldsymbol{\widetilde{w}}_j = (\widetilde{w}(\textbf{s}_{i}))_{i \in \mathcal{I}_j}^\top\) denote the residual vector. It is assumed to follow a multivariate normal distribution with a covariance matrix given by the difference between the full local covariance and the low-rank approximation: \vspace{-10pt}
      \begin{equation}\label{eq:full_local_covariance}
        \boldsymbol{\widetilde{w}}_j \sim \mathcal{N}\left(\boldsymbol{0}_{n_j}, \boldsymbol{C}_{\boldsymbol{\theta}}(S_j, S_j) - \boldsymbol{C}_{\boldsymbol{\theta}}(S_j, S^*) \boldsymbol{C}_{\boldsymbol{\theta}}(S^*, S^*)^{-1} \boldsymbol{C}_{\boldsymbol{\theta}}(S^*, S_j)\right),
        \vspace{-10pt}
      \end{equation}
      resulting in:
      \vspace{-10pt}
      \[
      \boldsymbol{R}_j(\delta, \boldsymbol{\theta}) = \boldsymbol{C}_{\boldsymbol{\theta}}(S_j, S_j) - \boldsymbol{C}_{\boldsymbol{\theta}}(S_j, S^*) \boldsymbol{C}_{\boldsymbol{\theta}}(S^*, S^*)^{-1} \boldsymbol{C}_{\boldsymbol{\theta}}(S^*, S_j) + \delta^{-1} \boldsymbol{I}_{n_j}.
      \vspace{-10pt}
      \]
      While this approach offers the most accurate approximation of the residual process, it incurs a higher computational cost, scaling cubically with the local sample size due to the need to evaluate and invert dense covariance sub-matrices. {To alleviate the computational cost on each machine, the low-rank approximation can be further applied locally, resulting in a hierarchical low-rank approximation.}
    \end{remark}

{Since the algorithm proposed and theory developed later can be adapted to other cases, in the following, we focus on the low-rank model with full local covariance as in \eqref{eq:full_local_covariance}, where the cross-covariances between workers are approximated using a low-rank model, while the local covariances are fully preserved within each worker.}
     
    \begin{proposition}\label{prop:kl_IDVSLOW}
      Let \( P \), \( P_1 \), and \( P_2 \) denote the distributions of the latent variables under the original Gaussian process \eqref{eq:spatial_data}, the low-rank model with full local covariance \eqref{eq:full_local_covariance}, and the independent model \eqref{eq:ind_log_likelihood}, respectively. Then, the dimension-normalized KL divergence satisfies
      \[
      \mathrm{KL}_N(P \,\|\, P_1) \leq \mathrm{KL}_N(P \,\|\, P_2) + \frac{m}{N},
      \]
      where \( \mathrm{KL}_N(P \,\|\, Q) := \frac{1}{N} \mathrm{KL}(P \,\|\, Q) \) is the KL divergence scaled by the dimension \( N \), and \( Q \in \{P_1, P_2\} \).
    \end{proposition}

{The proof of Proposition \ref{prop:kl_IDVSLOW} is provided in Section \ref{sec:proof_prop_kl_IDVSLOW}}
    \begin{remark}[{Normalization of KL Divergence with Dimension}]
      The KL divergence between two multivariate Gaussian distributions generally increases with the dimension. Specifically, if \( P = \mathcal{N}(\boldsymbol{0}, \boldsymbol{C}_p) \) and \( Q = \mathcal{N}(\boldsymbol{0}, \boldsymbol{C}_q) \), then
      \[
      \mathrm{KL}(P \,\|\, Q) = \frac{1}{2} \left( \mathrm{tr}(\boldsymbol{C}_q^{-1} \boldsymbol{C}_p) - N + \log \frac{\det \boldsymbol{C}_q}{\det \boldsymbol{C}_p} \right),
      \]
      which grows with the ambient dimension \( N \). To allow meaningful comparisons as the dimension increases, the KL divergence is normalized by \( N \).
    \end{remark}

    \begin{remark}[{KL Divergence and Likelihood Connection}]
      The reason for using \( \mathrm{KL}_N(P \| Q) \) instead of \( \mathrm{KL}_N(Q \| P) \) lies in its close connection to the likelihood. Specifically, the normalized KL divergence \( \mathrm{KL}_N(P \| Q) \) can be written as
      \[
      \mathrm{KL}_N(P \,\|\, Q) = -\frac{1}{N} \, \mathbb{E}_{\boldsymbol{w} \sim P} \left[ \log q(\boldsymbol{w}) \right] + \text{const.},
      \]
      where \( q \) is the density of \( Q \). This expression shows that minimizing \( \mathrm{KL}_N(P \| Q) \) is equivalent to maximizing the expected log-likelihood under model \( Q \), up to an additive constant.
    \end{remark}

    \begin{remark}\label{rem: low-rank_independent}
      When the number of inducing variables \( m \) satisfies \( m = o(N) \), a common assumption in low-rank models, we have
      \[
      \mathrm{KL}_N(P \,\|\, P_1) \leq \mathrm{KL}_N(P \,\|\, P_2),
      \]
      as \( N \to \infty \). Thus, the low-rank model with local covariance achieves a lower normalized KL divergence than the independent model in the large-sample limit.
    \end{remark}

    Compared to the naive independence model, the low-rank model offers a more accurate approximation of the true cross-covariance structure, potentially leading to improved performance in both inference and prediction. Proposition \ref{prop:kl_IDVSLOW}, whose proof is provided in Section \ref{sec:appendix:proofs}, formalizes this advantage by showing that the low-rank model asymptotically yields a smaller dimension-normalized KL divergence from the original distribution than the naive independence model (see Remark \ref{rem: low-rank_independent}). To the best of our knowledge, this proposition has not been established in the literature, which may be of {standalone theoretical significance}. The advantage of the low-rank model is also verified by the simulation results in Section \ref{sec:simulation:comparison_lowrank_independent}.  Motivated by this theoretical benefit and the improved performance in simulation, we develop the asynchronous algorithm based on the low-rank models.

% The spatial low-rank model leverages a shared low-rank structure to approximate the Gaussian process \(w(\textbf{s})\). Specifically, \(w(\textbf{s})\) is represented as:
% \[
% w(\textbf{s}) \approx \boldsymbol{\phi}(\textbf{s})^\top \boldsymbol{\eta},
% \]
% where \(\boldsymbol{\phi}(\textbf{s})\) is a vector of basis functions and \(\boldsymbol{\eta}\) is a low-dimensional random vector with a multivariate normal distribution, \(\boldsymbol{\eta} \sim \mathcal{N}(\boldsymbol{\mu}, \boldsymbol{\Sigma})\). This approximation reduces the computational complexity of modeling spatial correlations.

% The parameters of the model include \(\boldsymbol{\gamma}\), \(\sigma^2\), \(\boldsymbol{\mu}\), \(\boldsymbol{\Sigma}\), and the parameters of the basis functions \(\boldsymbol{\phi}(\textbf{s})\). These parameters are estimated by maximizing the log-likelihood function:
% \[
% \ell(\boldsymbol{\gamma}, \sigma^2, \boldsymbol{\mu}, \boldsymbol{\Sigma}, \boldsymbol{\phi}) = -\frac{1}{2} \sum_{j=1}^J \left[ \log |\boldsymbol{\Sigma}| + \text{tr}(\boldsymbol{\Sigma}^{-1} S_j) + \sum_{i=1}^{n_j} \frac{(z_{j,i} - \boldsymbol{x}_{j,i}^\top \boldsymbol{\gamma})^2}{\sigma^2} \right],
% \]
% where \(S_j\) is the covariance matrix of the low-rank representation for worker \(j\).

% The federated learning framework ensures that the optimization of this log-likelihood function is performed in a distributed manner, with each worker computing local updates based on its dataset \(\mathcal{D}_j\), and the central server aggregating these updates to refine the global parameters.

    \subsection{Synchronous Federated Modeling Algorithm}
    Since the new objective function in Equation \eqref{eq:low-rank-obj} now takes a summation form, federated modeling can be employed to optimize the objective function. Furthermore, apart from the parameters $\delta$ and $\boldsymbol{\theta}$, the optimization of other
    parameters can be explicitly solved when the remaining parameters are held fixed. Thus, the objective function can be optimized in a block coordinate descent manner. We note that while \cite{shi2025decentralized} developed a synchronous federated modeling algorithm for the special case of zero residual terms, we demonstrate that it can also be generalized to handle the case of non-zero residuals.

    Specifically, given the parameters $\boldsymbol{\mu}^t,\boldsymbol{\Sigma}^t,\boldsymbol{\gamma}^t, \delta^t,\boldsymbol{\theta}^t$ at iteration $t$, the updates of the parameters $\boldsymbol{\mu},\boldsymbol{\Sigma},\boldsymbol{\gamma}$ are given by:
    \begin{equation*}
      \begin{aligned}
        \boldsymbol{\mu}^{t+1}, \boldsymbol{\Sigma}^{t+1} & = \arg\min_{\boldsymbol{\mu}, \boldsymbol{\Sigma}} f(\boldsymbol{\mu}, \boldsymbol{\Sigma}, \boldsymbol{\gamma}^t, \delta^t, \boldsymbol{\theta}^t), \\
        \boldsymbol{\gamma}^{t+1} & = \arg\min_{\boldsymbol{\gamma}} f(\boldsymbol{\mu}^{t+1}, \boldsymbol{\Sigma}^{t+1}, \boldsymbol{\gamma}, \delta^t, \boldsymbol{\theta}^t),
      \end{aligned}
    \end{equation*}
    Resulting in the following updates:

    \begin{equation}\label{eq:mu_sigma_update}
      \begin{aligned}
        \boldsymbol{\Sigma}^{t+1} & =
        \left[\frac{1}{J}\sum_{j=1}^J \boldsymbol{B}_j^\top(\boldsymbol{\theta}^t) \boldsymbol{R}_j^{-1}(\delta^t,\boldsymbol{\theta}^t)\boldsymbol{B}_j(\boldsymbol{\theta}^t) + \boldsymbol{K}^{-1}(\boldsymbol{\theta}^t)\right]^{-1}, \\
        \boldsymbol{\mu}^{t+1} & = \boldsymbol{\Sigma}^{t+1}\left(\frac{1}{J}\sum_{j=1}^J \boldsymbol{B}_j^\top(\boldsymbol{\theta}^t)\boldsymbol{R}_j^{-1}(\delta^t,\boldsymbol{\theta}^t)(\boldsymbol{z}_j-\boldsymbol{X}_j\boldsymbol{\gamma}^t)\right),
      \end{aligned}
    \end{equation}
    \begin{equation}\label{eq:gamma_update}
      \boldsymbol{\gamma}^{t+1} = \left(\frac{1}{J}\sum_{j=1}^J \boldsymbol{X}_j^\top\boldsymbol{R}_j^{-1}(\delta^t,\boldsymbol{\theta}^t)\boldsymbol{X}_j\right)^{-1}\left(\frac{1}{J}\sum_{j=1}^J \boldsymbol{X}_j^\top\boldsymbol{R}_j^{-1}(\delta^t,\boldsymbol{\theta}^t)(\boldsymbol{z}_j-\boldsymbol{B}_j(\boldsymbol{\theta}^t)\boldsymbol{\mu}^{t+1})\right),
    \end{equation}
    where each local worker $j$ computes and transmits the following quantities to the central server: $\boldsymbol{B}_j^\top(\boldsymbol{\theta}^t) \boldsymbol{R}_j^{-1}(\delta^t,\boldsymbol{\theta}^t)\boldsymbol{B}_j(\boldsymbol{\theta}^t)$, $\boldsymbol{B}_j^\top(\boldsymbol{\theta}^t) \boldsymbol{R}_j^{-1}(\delta^t,\boldsymbol{\theta}^t)(\boldsymbol{z}_j-\boldsymbol{X}_j\boldsymbol{\gamma}^t)$, $\boldsymbol{X}_j^\top\boldsymbol{R}_j^{-1}(\delta^t,\boldsymbol{\theta}^t)\boldsymbol{X}_j$, $\boldsymbol{X}_j^\top\boldsymbol{R}_j^{-1}(\delta^t,\boldsymbol{\theta}^t)(\boldsymbol{z}_j-\boldsymbol{B}_j(\boldsymbol{\theta}^t)\boldsymbol{\mu}^t)$.
    The central server aggregates these quantities to update the parameters, which are then broadcast back to all workers for the next iteration.

    For parameters $\delta$ and $\boldsymbol{\theta}$, which lack closed-form solutions, we employ the Newton-Raphson method for parameter updates. This choice is motivated by its rapid convergence and computational efficiency for low-dimensional parameter spaces. The update rule is:
    \begin{equation}\label{eq:delta_theta_update}
      \begin{aligned}
        \widetilde{\boldsymbol{\theta}}^{t+1} = \widetilde{\boldsymbol{\theta}}^t - \alpha_t & \left[\frac{1}{J}\sum_{j=1}^J \frac{\partial^2 f_j(\boldsymbol{\mu}^{t+1},\boldsymbol{\Sigma}^{t+1},\boldsymbol{\gamma}^{t+1},\widetilde{\boldsymbol{\theta}}^t)}{\partial^2\widetilde{\boldsymbol{\theta}}}+\frac{1}{J}\frac{\partial^2 h(\boldsymbol{\mu}^{t+1},\boldsymbol{\Sigma}^{t+1},\boldsymbol{\theta})}{\partial^2\widetilde{\boldsymbol{\theta}^t}}\right]^{-1}\\
         & \times\left(\frac{1}{J}\sum_{j=1}^J \frac{\partial f_j(\boldsymbol{\mu}^{t+1},\boldsymbol{\Sigma}^{t+1},\boldsymbol{\gamma}^{t+1},\widetilde{\boldsymbol{\theta}}^t)}{\partial\widetilde{\boldsymbol{\theta}}}+\frac{1}{J}\frac{\partial h(\boldsymbol{\mu}^{t+1},\boldsymbol{\Sigma}^{t+1},\boldsymbol{\theta}^t)}{\partial\widetilde{\boldsymbol{\theta}}}\right)
      \end{aligned}
    \end{equation}
    where $\widetilde{\boldsymbol{\theta}} = (\delta, \boldsymbol{\theta}^\top)^\top$ combines both parameters and $\alpha_t\in (0,1]$ represents the step size. If the Hessian in Equation \eqref{eq:delta_theta_update} is not positive definite or contains eigenvalues that are too small, we replace negative eigenvalues by their absolute values and set overly small eigenvalues to a moderate positive threshold. A single Newton-Raphson iteration is sufficient for convergence in this context, as verified by the simulation. Similarly, each local worker computes the local gradient and Hessian of the objective function with respect to $\delta$ and $\boldsymbol{\theta}$ and sends them to the central server for aggregation. The central server updates the parameters and broadcasts them back to all workers for the next iteration.

  \section{Asynchronous Federated Modeling Algorithm}\label{sec:asynchronous}

    The assumption that all individual workers have equal computing power and identical communication delays with the central server is not entirely valid. Therefore, it is preferable for federated modeling to accommodate delays in computation or in transferring updates from the workers. This section provides an efficient asynchronous federated algorithm designed for low-rank models. The proposed algorithm aims to enhance efficiency in environments with heterogeneous computational resources or communication delays, while maintaining convergence rates comparable to the synchronous approach when resources are uniform.

    Although existing asynchronous methods \citep{fraboni2023general} based on stochastic gradient descent (SGD) can be directly applied to the objective function $f(\boldsymbol{\mu}, \boldsymbol{\Sigma}, \boldsymbol{\gamma}, \delta, \boldsymbol{\theta})$ in Equation \eqref{eq:low-rank-obj}, they do not exploit the specific structure of the function, which can result in a much slower convergence rate. In contrast, the synchronous method presented earlier leverages this structure through a block-wise update strategy. Building on this insight, the corresponding asynchronous algorithm should also adopt the same structure-aware design, rather than relying on generic gradient-based techniques.

    Similar to the process of the synchronous algorithm, the process of the asynchronous algorithm involves workers computing local quantities based on the parameter vectors they have and sending them to the server. The server aggregates the local quantities and updates the parameters, then sends the updated parameters back to the local workers. The difference is that the server should not wait for all workers to complete their computations and communication before proceeding. Instead, the server processes and aggregates as soon as partial local quantities from some workers become available, which reduces idle wait times and increases the system throughput.

    However, two challenges arise when designing an efficient asynchronous algorithm for the objective function $f(\boldsymbol{\mu}, \boldsymbol{\Sigma}, \boldsymbol{\gamma}, \delta, \boldsymbol{\theta})$ in Equation \eqref{eq:low-rank-obj}. The first challenge concerns the worker-side computation and arises from the block-wise update of the parameter vector, where each iteration is divided into three sub-steps: update of $\boldsymbol{\mu}, \boldsymbol{\Sigma}$ in Equation \eqref{eq:mu_sigma_update}; update of $\boldsymbol{\gamma}$ in Equation \eqref{eq:gamma_update}; update of $\delta$ and $\boldsymbol{\theta}$ in Equation \eqref{eq:delta_theta_update}. In the synchronous setting, each worker at each sub-step receives a single parameter vector and uses it to compute the corresponding local quantity. By contrast, in the asynchronous setting, a worker may receive multiple candidate parameter vectors before it is ready for the following computation, creating ambiguity in both choosing the appropriate parameter vector and determining which local quantity to compute.

    To eliminate ambiguity and efficiently manage parameter vectors for subsequent computations, we propose augmenting each parameter vector with an iteration index and a sub-step label that specifies which local quantity to compute, while each worker maintains a buffer to store the received augmented parameter vector. Specifically, let  $\mathbb{P}$ denote the buffer of a worker. Each element  $\texttt{p}\in \mathbb{P}$ is an augmented parameter vector represented as
    \begin{equation}\label{eq:augmented_parameter_vector}
      (\boldsymbol{\phi},\ t,\ \texttt{step}) \text{ with } \boldsymbol{\phi} = (\boldsymbol{\mu}, \boldsymbol{\Sigma}, \boldsymbol{\gamma}, \delta, \boldsymbol{\theta}),
    \end{equation}
   where $t = 0, 1, 2, \ldots$ is the iteration index and \texttt{step} the sub-step label in the set:
    \begin{equation}\label{eq:sub_step_label}
      \left\{\texttt{step}_{\boldsymbol{\mu}, \boldsymbol{\Sigma}},\ \texttt{step}_{\boldsymbol{\gamma}},\ \texttt{step}_{\delta, \boldsymbol{\theta}} \right\},
    \end{equation}
    For convenience, we denote by $\texttt{p.param}, \texttt{p.iter}, \texttt{p.step}$ the parameter tuple $\boldsymbol{\phi}$, the iteration index, and the sub-step label of the augmented parameter vector $\texttt{p}$.

    The buffer $\mathbb{P}$ follows a \emph{sub-step-aware FIFO policy}: when a new augmented parameter vector $\texttt{p}^{\text{new}}$ is received at worker~\( j \), it is appended at the end of $\mathbb{P}$ if no existing entry in $\mathbb{P}$ has the same sub-step label $\texttt{p}^{\text{new}}\texttt{.step}$; otherwise, $\texttt{p}^{\text{new}}$ replaces the existing entry with the same sub-step label $\texttt{p}^{\text{new}}\texttt{.step}$. This ensures that at most one entry per sub-step is maintained, while preserving the relative order of distinct sub-steps. When worker~\( j \) is ready to compute, an augmented parameter vector $\texttt{p}$ is popped from the front of the buffer and used for the local computation, which is abstracted as a function \( L_j(\texttt{p}) \). This function maps the received augmented parameter vector $\texttt{p}$ to the local quantity needed by the server. The mapping depends not only on the parameter vector, but also on the sub-step label $\texttt{p.step}$. Specifically,
    \begin{equation}\label{eq:local_quantity_function}
      L_j(\texttt{p}) :=
      \begin{cases}
        \left(
        \boldsymbol{B}_j^\top(\boldsymbol{\theta}) \boldsymbol{R}_j^{-1}(\delta, \boldsymbol{\theta}) \boldsymbol{B}_j(\boldsymbol{\theta}),\
        \boldsymbol{B}_j^\top(\boldsymbol{\theta}) \boldsymbol{R}_j^{-1}(\delta, \boldsymbol{\theta})(\boldsymbol{z}_j - \boldsymbol{X}_j \boldsymbol{\gamma})
        \right),
        & \text{if } \texttt{p.step}= \texttt{step}_{\boldsymbol{\mu}, \boldsymbol{\Sigma}}, \\[1.5ex]

        \left(
        \boldsymbol{X}_j^\top\boldsymbol{R}_j^{-1}(\delta,\boldsymbol{\theta})\boldsymbol{X}_j,\
        \boldsymbol{X}_j^\top\boldsymbol{R}_j^{-1}(\delta,\boldsymbol{\theta})(\boldsymbol{z}_j-\boldsymbol{B}_j(\boldsymbol{\theta})\boldsymbol{\mu})
        \right),
        & \text{if } \texttt{p.step} = \texttt{step}_{\boldsymbol{\gamma}}, \\[1.5ex]

        \left(
        \displaystyle\frac{\partial^2 f_j(\boldsymbol{\mu}, \boldsymbol{\Sigma}, \boldsymbol{\gamma}, \widetilde{\boldsymbol{\theta}})}{\partial^2 \widetilde{\boldsymbol{\theta}}},\
        \displaystyle\frac{\partial f_j(\boldsymbol{\mu}, \boldsymbol{\Sigma}, \boldsymbol{\gamma}, \widetilde{\boldsymbol{\theta}})}{\partial \widetilde{\boldsymbol{\theta}}}
        \right),
        & \text{if } \texttt{p.step} = \texttt{step}_{\delta, \boldsymbol{\theta}}.
      \end{cases}
    \end{equation}

    Once the local quantity is computed, the worker sends it to the server, augmented with the same iteration index and sub-step label as those used in the parameter vector from which the quantity was derived, along with the worker identifier. We denote the augmented local quantity by $\texttt{q}$:
    \begin{equation}\label{eq:augmented_local_quantity}
      \texttt{q} = (\boldsymbol{\xi}, t, \texttt{step}, j),
    \end{equation}
    where $\boldsymbol{\xi}$ is the local quantity tuple computed by the worker at the right-hand side of Equation \eqref{eq:local_quantity_function}, $t$ is the iteration index, $\texttt{step}$ is the sub-step label, and $j=1, \ldots, J$ is the worker identifier. In the following, we use the notation $\texttt{q.quantity}, \texttt{q.iter}, \texttt{q.step}, \texttt{q.worker}$ to denote the local quantity tuple, iteration index, sub-step label, and worker identifier of the augmented local quantity $\texttt{q}$.

    The second challenge is that server-side naive aggregation, such as simple averaging or summation of local quantities, as done in the synchronous algorithm, can lead to unstable and slower convergence, as demonstrated in our numerical experiments (see the first experiment in Section \ref{sec:simulation:comparison_async_sync_lowrank}). After receiving local quantities from a subset of workers, the server computes global aggregates using both the newly received local quantities and those previously received from other workers. Since these local quantities may be based on outdated or inconsistent parameter versions, their contributions can be misaligned, ultimately degrading convergence performance.

    While this issue is inherent to asynchronous systems, it can be effectively mitigated. We propose the following strategies: 1) local gradient correction, which incorporates Hessian information to compensate for inconsistencies of local gradients in Equation \eqref{eq:delta_theta_update}; 2) adaptive aggregation, which dynamically adjusts the contribution weights based on the staleness of local quantities; 3) a moving average of historical parameter vectors to enhance the stability of the algorithm. These strategies are elaborated in the following paragraphs.

        \paragraph{Local Gradient Correction}
    To mitigate inconsistencies among local gradients, we adopt a local gradient correction strategy that incorporates second-order information. Specifically, at the server side, during the sub-step $\texttt{step}_{\delta, \boldsymbol{\theta}}$, let $\texttt{p}_j$ denote the augmented parameter vector on which worker $ j$'s local gradient is based, where $\texttt{p}_j.\texttt{step} = \texttt{step}_{\delta, \boldsymbol{\theta}}$.
    Among the set \(\{\texttt{p}_j: j = 1, \ldots, J\}\), let \(\texttt{p}_{rec}\) denote the element corresponding to the most recent iteration index. We use the notation \((\boldsymbol{\mu}_j, \boldsymbol{\Sigma}_j, \boldsymbol{\gamma}_j, \widetilde{\boldsymbol{\theta}}_j)\) for the parameters associated with \(\texttt{p}_j\), and \((\boldsymbol{\mu}_{rec}, \boldsymbol{\Sigma}_{rec}, \boldsymbol{\gamma}_{rec},  \widetilde{\boldsymbol{\theta}}_{rec})\) for those associated with \(\texttt{p}_{rec}\). The original local gradient corresponding to the $j$-th worker is $\boldsymbol{g}_j :=
    \frac{\partial f_j(\boldsymbol{\mu}_j, \boldsymbol{\Sigma}_j, \boldsymbol{\gamma}_j, \widetilde{\boldsymbol{\theta}}_j)}{\partial \widetilde{\boldsymbol{\theta}}_j} $.
    The corrected local gradient at the $j$-th worker is then given by
    \begin{equation}\label{eq:local_gradient_correction}
      \begin{aligned}
        \boldsymbol{g}_j^\text{corrected} := & \boldsymbol{g}_j
        +
        \frac{\partial^2 f_j(\boldsymbol{\mu}_j, \boldsymbol{\Sigma}_j, \boldsymbol{\gamma}_j, \widetilde{\boldsymbol{\theta}}_j)}{\partial \widetilde{\boldsymbol{\theta}}_j^2} (\widetilde{\boldsymbol{\theta}}_{rec}-\widetilde{\boldsymbol{\theta}}_j)  \\
         & +\frac{\partial^2 f_j(\boldsymbol{\mu}_j, \boldsymbol{\Sigma}_j, \boldsymbol{\gamma}_j, \widetilde{\boldsymbol{\theta}}_j)}{\partial \widetilde{\boldsymbol{\theta}} \, \partial \boldsymbol{\mu}} (\boldsymbol{\mu}_{rec} - \boldsymbol{\mu}_j)
        + \frac{\partial^2 f_j(\boldsymbol{\mu}_j, \boldsymbol{\Sigma}_j, \boldsymbol{\gamma}_j, \widetilde{\boldsymbol{\theta}}_j)}{\partial \widetilde{\boldsymbol{\theta}} \, \partial \boldsymbol{\Sigma}} (\boldsymbol{\Sigma}_{rec} - \boldsymbol{\Sigma}_j).
      \end{aligned}
    \end{equation}

    This correction term helps align the local gradients, improving convergence, especially when parameter divergence across workers is significant. Here, the term involving a partial derivative concerning $\boldsymbol{\gamma}$ is excluded from the correction term, as a reasonable initial estimate of $\boldsymbol{\gamma}$, close to the optimizer,  can be obtained even under a misspecified dependence structure (e.g., assuming independence; see \citealp{cui2004m,wu2007m}), thus its variation is relatively small across iterations.

    We note that the additional partial derivatives involved in Equation \eqref{eq:local_gradient_correction} are given by:
    \begin{equation}\label{eq:local_gradient_correction_partial_derivative}
      G_j(\texttt{p}): = \left(\frac{\partial^2 f_j(\boldsymbol{\mu}, \boldsymbol{\Sigma}, \boldsymbol{\gamma}, \widetilde{\boldsymbol{\theta}})}{\partial \widetilde{\boldsymbol{\theta}} \, \partial \boldsymbol{\mu}}, \frac{\partial^2 f_j(\boldsymbol{\mu}, \boldsymbol{\Sigma}, \boldsymbol{\gamma}, \widetilde{\boldsymbol{\theta}})}{\partial \widetilde{\boldsymbol{\theta}} \, \partial \boldsymbol{\Sigma}}\right) \text{ if } \texttt{p.step}=\texttt{step}_{\delta, \boldsymbol{\theta}}.
    \end{equation}

    Equation \eqref{eq:local_gradient_correction_partial_derivative} at $\texttt{p}_j$ can be easily computed from $\boldsymbol{g}_j$ using the chain rule. As a result, the local gradient correction can be implemented efficiently, incurring minimal additional computational overhead. However, for other types of local quantities in Equation \eqref{eq:local_quantity_function}, while the same correction principle remains applicable, computing the corresponding partial derivatives is considerably more complex and computationally expensive. In such cases, the adaptive aggregation strategy described below offers a practical alternative for alleviating the effects of staleness.
%%%%%%%%%%%

        \paragraph{Adaptive Aggregation}
    To mitigate the adverse effects of stale local quantities, we propose an adaptive aggregation scheme that assigns staleness-aware weights to local contributions. For convenience, we introduce the following notation: at server iteration $t$ during sub-step \texttt{step}, let $\texttt{q}_j^t$ denote the augmented local quantity received from worker $j$, with $\texttt{q}_j^t.\texttt{step} = \texttt{step}$ and $\texttt{q}_j^t.\texttt{worker} = j$. Let $t_j^t := \texttt{q}_j^t.\texttt{iter}$ be the iteration index encoded within it. The staleness of the local quantity from worker $j$ is then defined as $\tau_j^t := t - t_j^t$.

    Instead of a simple average, the global quantity is computed as a weighted average:
    \begin{equation}\label{eq:adaptive_aggregation}
      \sum_{j=1}^J w_j^t \texttt{q}_j^t.\texttt{quantity},
    \end{equation}
    where the weights \(w_j^t\) are designed to satisfy two key principles: 1) larger staleness \(\tau_j^t\) should correspond to smaller weight \(w_j^t\), thereby reducing the influence of potentially outdated or misaligned local updates; and 2) as the iteration index \(t\) increases, the weights should become more uniform. The latter reflects the intuition that as optimization progresses, the parameter vectors across workers tend to converge, and the impact of staleness diminishes. In this regime, uniform weighting helps avoid bias introduced by unequal contributions.

    Following these principles, we propose the following weighting scheme:
    \begin{equation}\label{eq:adaptive_aggregation_weight}
      w_j^t \propto
      \begin{cases}
        \left(\tau_j^t + 1 + \max\{\sqrt{t}, \|{\boldsymbol{g}}^{t-1}\|\}\right)^{-a}, & \text{if } t_j^t \leq T_c \text{ for all } j, \\
        1, & \text{otherwise},
      \end{cases}
    \end{equation}
    where \({\boldsymbol{g}}^{t-1}\) denotes the aggregated gradient at iteration \(t-1\) and step \(\texttt{step}_{\delta, \boldsymbol{\theta}}\), namely, \({\boldsymbol{g}}^{t-1} = \sum_{j=1}^J w_j^{t-1} \boldsymbol{g}_j^{corrected, t-1}\) with \(\boldsymbol{g}_j^{corrected, t-1}\) defined as in Equation~\eqref{eq:local_gradient_correction}. The term \(s_j^t\) ensures that stale gradients are downweighted. The additional term \(\max\{\sqrt{t}, \|{\boldsymbol{g}}^{t-1}\|\}\) serves to gradually flatten the weighting over time, promoting balance as optimization proceeds. Here, $a>0$ controls the decay rate of the weights. When all local quantities are based on sufficiently recent parameters, i.e., \(t_j^t > T_c\) for all \(j\), we assign uniform weights to ensure fairness across machines.

        \paragraph{Moving Average}
    To stabilize the optimization trajectory, we incorporate a moving average of historical parameters. Let \(\boldsymbol{\phi}_t \) denote the parameter tuple defined as in Equation \eqref{eq:augmented_parameter_vector} at iteration \(t\). The moving average is updated by
    \begin{equation}\label{eq:moving_average}
      \boldsymbol{\phi}_t \gets \sum_{i=0}^{M} \frac{\omega^i}{Z} \, \boldsymbol{\phi}_{t - i}, \quad \text{with } M \leq t \text{ and } Z = \sum_{i=0}^{M} \omega^i,
    \end{equation}
    where \(\gets\) denotes assignment, and \(\omega \in (0,1)\) controls the exponential decay rate of the weights. The arithmetic is applied element-wise to each component of the parameter tuple. This weighting scheme assigns greater importance to recent parameter values, thereby enhancing stability while maintaining adaptability to new information.

    In the above, we described the computation, receiving, and sending processes on the worker side, as well as the aggregation and computation steps on the server side. Several implementation details remain. First, to manage communication between the server and workers efficiently and to avoid congestion caused by concurrent computations, parallel threads should be employed to handle communication tasks independently. Second, an initialization step is required to ensure that, when computing global aggregates, the server can use the initial local quantities from workers that have not yet sent updated values. Third, the server maintains only the most recently received local quantity from each worker. Specifically, let $\mathbb{Q}$ denote the server-side buffer that stores local quantities. When a new augmented local quantity $\texttt{q}^{\text{new}}$ is received from worker $j$, it replaces the existing entry $\texttt{q}$ in $\mathbb{Q}$ for which $\texttt{q.step} = \texttt{q}^{\text{new}}\texttt{.step}$ and $\texttt{q.iter} = \texttt{q}^{\text{new}}\texttt{.iter}$.
    Fourth, we introduce a configuration variable, $\texttt{Agg\_threshold}$, to control the level of asynchrony. For each step, only when the number of local quantities received from workers is greater than or equal to $\texttt{Agg\_threshold}$, the server computes the global aggregate and updates the parameter vector. Specifically, the server maintains a counter $\texttt{Num.step}$ initialized at 0 for each step $\texttt{step} \in \{\texttt{step}_{\boldsymbol{\mu}, \boldsymbol{\Sigma}}, \texttt{step}_{\boldsymbol{\gamma}}, \texttt{step}_{\delta, \boldsymbol{\theta}}\}$. Upon receiving a new augmented local quantity $\texttt{q}^{\text{new}}$, it replaces the matching entry in $\mathbb{Q}$ and increments $\texttt{Num.step}$. Aggregation is triggered only when $\texttt{Num.step} \geq \texttt{Agg\_threshold}$. Thus, when $\texttt{Agg\_threshold} = J$, where $J$ is the number of workers, the aggregation becomes synchronous, requiring new quantities from all workers before proceeding. In contrast, setting $\texttt{Agg\_threshold} < J$ allows asynchronous aggregation, enabling the server to proceed once it has new quantities from only a subset of workers. This mechanism provides a flexible trade-off between synchronization overhead and update staleness.

    \begin{remark}[Hyperparameter tuning]\label{remark:hyperparameter_tuning}
      Several hyperparameters influence the convergence behavior of the algorithm, including the step size $\alpha_t$ for the Newton update in Equation \eqref{eq:delta_theta_update}, the weighting parameters $a$ and $T_c$ in Equation \eqref{eq:adaptive_aggregation_weight}, and the moving average parameters $\omega$ and $t_M$ in Equation \eqref{eq:moving_average}. Specifically, a large $\alpha_t$ may cause non-convergence, while a small one slows down convergence. Larger values of $a$ or $T_c$ give more weight to recent quantities, potentially accelerating convergence but introducing bias and risking divergence. In contrast, smaller $a$ or $T_c$ place more emphasis on stale quantities, which can slow down progress. Similarly, a large $\omega$ combined with a small $t_M$ makes the algorithm overly sensitive to recent updates, resulting in fluctuations, whereas a small $\omega$ and large $t_M$ cause the algorithm to rely heavily on historical parameters, leading to sluggish convergence.  In practice, we recommend running 10–30 iterations to identify a suitable set of hyperparameters. Empirically, setting $\alpha_t = 0.5 \times \texttt{Agg\_threshold} / J$ and choosing $a \in \{1, 2\}$, $T_c \in \{0, 3\}$, $\omega \in \{0.5, 0.7\}$, and $t_M \in \{2, 8, 50\}$ is typically sufficient to ensure fast and stable convergence in the asynchronous setting.
    \end{remark}

    In summary, the complete algorithm consists of three components: initialization, worker-side computation, and server-side aggregation and updates, as outlined in Algorithm~\ref{algo:init}, Algorithm~\ref{alg:worker}, and Algorithm~\ref{algo:server}, respectively.
    
\begin{algorithmframework}
    {\footnotesize
      \algrenewcommand\algorithmicindent{0.8em}
      \sloppy
\linespread{1}
      \noindent
      \begin{minipage}[t]{0.48\textwidth}
        \begin{algorithm}[H]
          \captionsetup{font=footnotesize}
          \caption{Initialization}
          \label{algo:init}
          \begin{algorithmic}[1]
            \footnotesize
            \State Initialize \(\boldsymbol{\phi}\gets(\boldsymbol{\mu}, \boldsymbol{\Sigma}, \boldsymbol{\gamma}, \boldsymbol{\delta}, \boldsymbol{\theta})\), \(\texttt{p}\gets(\boldsymbol{\phi}, 0, \texttt{step}_{\boldsymbol{\mu}, \boldsymbol{\Sigma}})\)  at the server and local workers.
            \State Initialize empty buffers $\mathbb{Q}$ in the server to store initial local quantities.
            \For{\(\texttt{step} \gets \texttt{step}_{\boldsymbol{\mu}, \boldsymbol{\Sigma}}, \texttt{step}_{\boldsymbol{\gamma}}, \texttt{step}_{\delta, \boldsymbol{\theta}}\)}
            \State For worker $j$, compute \(\texttt{quantity}\gets L_j(\texttt{p})\) according to \eqref{eq:local_quantity_function}
            \If{$\texttt{step}=\texttt{step}_{\delta, \boldsymbol{\theta}}$}
            \State Compute partial derivatives $G_j(\texttt{p})$ according to \eqref{eq:local_gradient_correction_partial_derivative}
            \State Add $G_j(\texttt{p})$ to $\texttt{quantity}$
            \EndIf
            \State Send $\texttt{q}_j\gets(\texttt{quantity}, 0, \texttt{step}, j)$ to the server, for all $j=1, \ldots, J$
            \State Server adds $\texttt{q}_j$ to $\mathbb{Q}$
            \State Server computes the global aggregate according to \eqref{eq:adaptive_aggregation}
            \State Server updates parameter $\boldsymbol{\phi}$ according to  \eqref{eq:mu_sigma_update}, \eqref{eq:gamma_update}, or \eqref{eq:delta_theta_update}
            \State Server sends $\texttt{p}\gets(\boldsymbol{\phi}, 0, \texttt{step})$ to all workers
            \EndFor

          \end{algorithmic}
        \end{algorithm}

        \vspace{-0.3cm}

        \begin{algorithm}[H]
          \captionsetup{font=footnotesize}
          \caption{Worker-Side Computation: the $j$-th worker}
          \label{alg:worker}
          \begin{algorithmic}[1]
            \footnotesize
            \State Initialization empty buffers $\mathbb{P},\mathbb{Q}$ for parameter vectors and local quantities, respectively,
            \State Launch three parallel threads for receiving, sending, and computing, respectively.

            \begin{tikzpicture}[overlay]
              \draw[dashed, red] (-1.4,0) rectangle (5.8,-3.6);
              \node[blue] at (2.5,0) {Receiving};
            \end{tikzpicture}

            \While{\texttt{wait=True}} %\Comment{Thread 1: receiving the new parameter vector}
            \State Receive $\texttt{p}^{\text{new}}$ from the server
            \If{$\texttt{p}^{\text{new}}\texttt{.step}\not \in \{\texttt{p}\texttt{.step}, \texttt{p}\in \mathbb{P} \}$ }
            \State Append $\texttt{p}^{\text{new}}$ to the end of the buffer $\mathbb{P}$
            \Else
            \State Find $\texttt{p}\in \mathbb{P}$ satisfying $\texttt{p}\texttt{.step}=\texttt{p}^{\text{new}}\texttt{.step}$
            \State Replace $\texttt{p}$ by $\texttt{p}^{\text{new}}$
            \EndIf
            \EndWhile

            \begin{tikzpicture}[overlay]
              \draw[dashed, red] (-1.4,0) rectangle (5.8,-1.7);
              \node[blue] at (2.5,0) {Sending};
            \end{tikzpicture}

            \While{\texttt{wait=True} and $\mathbb{Q}$ is not empty} %\Comment{Thread 2: sending the local quantity to the server}
            \State Pop the first element $\texttt{q}$ from $\mathbb{Q}$
            \State Send $\texttt{q}$ to the server
            \EndWhile

            \begin{tikzpicture}[overlay]
              \draw[dashed, red] (-1.4,0) rectangle (5.9,-4.7);
              \node[blue] at (2.5,0) {Computing};
            \end{tikzpicture}

            \While{\texttt{wait=True} and $\mathbb{P}$ is not empty} %\Comment{Thread 3: computing the local quantity}
            \State Pop the first element $\texttt{p}$ from $\mathbb{P}$.
            \State Compute the local quantity $\texttt{quantity}:=L_j(\texttt{p})$ according to \eqref{eq:local_quantity_function}.
            \If{$\texttt{p}\texttt{.step}=\texttt{step}_{\delta, \boldsymbol{\theta}}$}
            \State Compute derivatives $G_j(\texttt{p})$ according to \eqref{eq:local_gradient_correction_partial_derivative}
            \State Add $G_j(\texttt{p})$ to $\texttt{quantity}$
            \EndIf
            \State Add $(\texttt{quantity}, \texttt{p}\texttt{.iter}, \texttt{p}\texttt{.step}, j)$ to the buffer $\mathbb{Q}$

            \EndWhile

          \end{algorithmic}
        \end{algorithm}
      \end{minipage}
      \hfill
      \begin{minipage}[t]{0.52\textwidth}
        \begin{algorithm}[H]
          \captionsetup{font=footnotesize}
          \caption{Server-Side Updates}
          \label{algo:server}
          \begin{algorithmic}[1]
            \footnotesize
            \State Initialization empty buffers $\mathbb{P}$ for parameters
            \State Launch one compute thread.
            \State Launch $2J$ threads for sending and receiving with workers.
            \State Set counters $l_j\gets 0$ for all $j=1, \ldots, J$.
            \State Set $\texttt{Num.step}\gets 0$ for  $\texttt{step} \in \{\texttt{step}_{\boldsymbol{\mu}, \boldsymbol{\Sigma}}, \texttt{step}_{\boldsymbol{\gamma}}, \texttt{step}_{\delta, \boldsymbol{\theta}}\}$

            \begin{tikzpicture}[overlay]
              \draw[dashed, red] (-1.4,0) rectangle (6.7,-2.8);
              \node[blue] at (2.8,0) {Receiving};
            \end{tikzpicture}

            \While{\texttt{wait=True}}  % Monospace style
            \State Receive augmented local quantities $\texttt{q}^{\text{new}}$ from  workers
            \State Find $\texttt{q}\in \mathbb{Q}$ with $\texttt{q}\texttt{.step}=\texttt{q}^{\text{new}}\texttt{.step}$ and $\texttt{q}\texttt{.iter}=\texttt{q}^{\text{new}}\texttt{.iter}$, and replace $\texttt{q}$ by $\texttt{q}^{\text{new}}$
            \State $\texttt{Num.step}\gets\texttt{Num.step}+1$
            \EndWhile

            \begin{tikzpicture}[overlay]
              \draw[dashed, red] (-1.4,0) rectangle (6.7,-3.6);
              \node[blue] at (2.8,0) {Sending};
            \end{tikzpicture}

            \While{\texttt{wait=True}}  % Monospace style
            \ForAll{\(j = 1, \ldots, J\)} \Comment{In parallel}
            \If{$l_j<len(\mathbb{P})$}
            \State Get augmented parameter vector $\texttt{p}:=\mathbb{P}[l_j]$
            \State Send $\texttt{p}$ to the worker $j$
            \State $l_j:=l_j+1$
            \EndIf
            \EndFor
            \EndWhile

            \begin{tikzpicture}[overlay]
              \draw[dashed, red] (-1.4,0) rectangle (6.7,-9.0);
              \node[blue] at (2.8,0) {Aggregation and Updating};
            \end{tikzpicture}

            \For{\(t = 1, \ldots, T\)}
            \For {\(\texttt{step} = \texttt{step}_{\boldsymbol{\mu}, \boldsymbol{\Sigma}}, \texttt{step}_{\boldsymbol{\gamma}}, \texttt{step}_{\delta, \boldsymbol{\theta}}\)}
            \If{$\texttt{Num.step}<\texttt{Agg\_threshold}$}
            \State Get $\{\texttt{q}_j\}_{j=1}^J$ from $\mathbb{Q}$ with $\texttt{q}_j\texttt{.step}=\texttt{step}$ and $\texttt{q}_j\texttt{.worker}=j$

            \If{$\texttt{step}=\texttt{step}_{\delta, \boldsymbol{\theta}}$}
            \State Compute the corrected gradient in $\{\texttt{q}_j\}_{j=1}^J$ according to \eqref{eq:local_gradient_correction}
            \State Aggregate the Hessian in $\{\texttt{q}_j\}_{j=1}^J$ and corrected gradient according to \eqref{eq:adaptive_aggregation}
            \Else
            \State Aggregate the local quantities in $\{\texttt{q}_j\}_{j=1}^J$ according to \eqref{eq:adaptive_aggregation}
            \EndIf

            \State Update parameter vector $\boldsymbol{\phi}$ according to \eqref{eq:mu_sigma_update}, \eqref{eq:gamma_update}, or \eqref{eq:delta_theta_update} by replacing the corresponding terms
            \State Perform moving average for $\boldsymbol{\phi}$ according to \eqref{eq:moving_average}
            \State Add $\texttt{p}\gets(\boldsymbol{\phi}, t, \texttt{step})$ to $\mathbb{P}$
            \EndIf
            \EndFor
            \EndFor
            \State Set \texttt{wait=False} and broadcast to all workers
          \end{algorithmic}
        \end{algorithm}
      \end{minipage}
    }

\end{algorithmframework}

  \section{Theoretical Analysis}\label{sec:theory}
    In this section, we present the theoretical analysis of the proposed asynchronous algorithm.

    We begin by treating the linear regression coefficient vector $\boldsymbol{\gamma}$ as fixed. This assumption is made for two reasons. First, in the asynchronous setting, block-wise optimization with exact updates of more than two parameter blocks, as is the case for $\boldsymbol{\gamma}$ and $(\boldsymbol{\mu}, \boldsymbol{\Sigma})$ in this work, requires much stronger conditions on a general objective function to guarantee convergence, as illustrated in Example~\ref{example:block_wise_update}. Second, the value of $\boldsymbol{\gamma}$ changes only slightly during the iterations and remains close to its true value. This is because the estimation of $\boldsymbol{\gamma}$ achieves $\sqrt{N}$-consistency under suitable conditions, even when the covariance structure is misspecified (e.g., when the independence model is used; see \citealp{cui2004m,wu2007m}). This is also verified in the numerical experiments. Hence, for large sample sizes, the deviation in $\boldsymbol{\gamma}$ is negligible, and fixing it does not materially affect the validity of the theoretical analysis.

    \begin{example}\label{example:block_wise_update}
      Consider the problem of optimizing a quadratic function $F$ with two variables $(\boldsymbol{z}_1,\boldsymbol{z}_2)$ using block-wise parallel updates:
      \begin{equation*}
        \boldsymbol{z}_1^{t+1} = \arg\min_{\boldsymbol{z}_1} F(\boldsymbol{z}_1, \boldsymbol{z}_2^t),
        \quad
        \boldsymbol{z}_2^{t+1} = \arg\min_{\boldsymbol{z}_2} F(\boldsymbol{z}_1^t, \boldsymbol{z}_2).
      \end{equation*}
      Here, the update of $\boldsymbol{z}_2^{t+1}$ is based on the previous iterate $\boldsymbol{z}_1^t$ rather than the most recent $\boldsymbol{z}_1^{t+1}$. Thus, the block-wise parallel update can be interpreted as an asynchronous sequential update with a delay of one. Despite the simplicity of both the problem and the update rule, convergence is not guaranteed in general, even when the objective function is smooth and strongly convex, unless additional contractiveness conditions are satisfied (see Proposition~2.6 in Sec.~3.2 and Proposition~3.10 in Sec.~3.3 of \citet{bertsekas2015parallel}).
    \end{example}

    Then, since $\boldsymbol{\gamma}$ is fixed in the analysis, we  rewrite the problem, for simplicity, as
    \begin{equation*}
      J^{-1}\min_{\boldsymbol{x}_1,\boldsymbol{x}_2} \left\{f(\boldsymbol{x}_1,\boldsymbol{x}_2)=\sum_{j=1}^J f_j(\boldsymbol{x}_1,\boldsymbol{x}_2)+h(\boldsymbol{x}_1,\boldsymbol{x}_2)\right\},
    \end{equation*}
    where $\boldsymbol{x}_1$ collects the variables $(\boldsymbol{\mu}, \boldsymbol{\Sigma})$ and $\boldsymbol{x}_2$ collects $(\delta, \boldsymbol{\theta})$.

    The following analysis proceeds in two stages for the convergence. In the first stage, we establish the convergence of the basic asynchronous algorithm, which does not incorporate the improved strategies proposed in this work, for general local objective functions $f_j$ and a common term $h$. Building on this foundation, the second stage establishes the convergence of the proposed asynchronous algorithm with the improved strategies. The proof of all the results in this section is provided in Section \ref{sec:proof_results_sec_theory} of the Appendix.

    The update of the basic asynchronous algorithm can be described as
    \begin{align*}
      \boldsymbol{x}_1^{t + 1}
       & = \arg \min_{\boldsymbol{x}_1} \left\{
        J^{-1} \sum_{j=1}^J \Big[
        f_j \big(\boldsymbol{x}_1, \boldsymbol{x}_2^{\,t - \tau_{t,j}^1}\big)
        + h \big(\boldsymbol{x}_1 , \boldsymbol{x}_2^{\,t}\big)
        \Big] \right\}, \\
        \boldsymbol{x}_2^{t + 1}
         & = \boldsymbol{x}_2^t - \alpha_t
        \left( \boldsymbol{H}_2^{\{t, \tau_{t,j}^{21}, \tau_{t,j}^{22}\}} \right)^{-1}
        \boldsymbol{G}_2^{\{t, \tau_{t,j}^{21}, \tau_{t,j}^{22}\}} ,
      \end{align*}
      where $\boldsymbol{G}_2^{\{t, \tau_{t,j}^{21}, \tau_{t,j}^{22}\}}$ and
      $\boldsymbol{H}_2^{\{t, \tau_{t,j}^{21}, \tau_{t,j}^{22}\}}$ denotes the stale gradient and modified  Hessian, respectively:
      \begin{align*}
        \boldsymbol{G}_2^{\{t, \tau_{t,j}^{21}, \tau_{t,j}^{22}\}}
         & := J^{-1} \sum_{j=1}^J
        \frac{\partial f_j(\boldsymbol{x}_1^{t+1 - \tau_{t,j}^{21}},
        \boldsymbol{x}_2^{\,t - \tau_{t,j}^{22}})}{\partial\boldsymbol{x}_2}
        + \frac{\partial h(\boldsymbol{x}_1^t, \boldsymbol{x}_2^t)}{\partial\boldsymbol{x}_2} , \\
        \boldsymbol{H}_2^{\{t, \tau_{t,j}^{21}, \tau_{t,j}^{22}\}}
         & := \operatorname{mod}\!\left[
        J^{-1} \sum_{j=1}^J
        \frac{\partial^2 f_j(\boldsymbol{x}_1^{\,t+1 - \tau_{t,j}^{21}},
        \boldsymbol{x}_2^{t - \tau_{t,j}^{22}})}{\partial\boldsymbol{x}_2^2}
        + \frac{\partial^2 h(\boldsymbol{x}_1^{t+1}, \boldsymbol{x}_2^t)}{\partial\boldsymbol{x}_2^2}
        \right].
      \end{align*}
      Here, for any symmetric matrix $\boldsymbol{H}$, the operator $\operatorname{mod}(\boldsymbol{H})$ replaces negative eigenvalues with their absolute values and lifts small eigenvalues to a positive threshold $\underline{\lambda}$, thereby ensuring positive definiteness, i.e., $\operatorname{mod}(\boldsymbol{H}) \succeq \underline{\lambda}^{-1}\boldsymbol{I}$. The nonnegative integers $\tau_{t,j}^{1}$, $\tau_{t,j}^{21}$, and $\tau_{t,j}^{22}$ denote the staleness in the $j$-th worker, i.e., the number of iterations by which the information used in the updates of $\boldsymbol{x}_1^{t+1}$ and $\boldsymbol{x}_2^{t+1}$ lags behind the current iteration.

      To the best of our knowledge, neither this algorithm nor its theoretical underpinnings has been investigated in the optimization literature. Therefore, our results may hold standalone theoretical significance for the community. The proposed algorithm is closely related to the gradient-based federated asynchronous method \citep{vanli2018global,xie2019asynchronous,nguyen2022federated,feyzmahdavian2023asynchronous,xu2024fedfa,zakerinia2024communication} where the gradient-based update is applied across all coordinate blocks. In contrast, our approach partitions the variables into two distinct blocks: one block is updated asynchronously in closed form, while the other is updated asynchronously using a gradient-based method, which can lead to improved computational efficiency. From a theoretical standpoint, this setting introduces additional challenges: beyond the coupling between staleness and gradient-based update considered in prior analyses, we must also account for the interaction between the exact closed-form update and the staleness, rendering the analysis substantially more involved.
      The algorithm is also related to the asynchronous block coordinate method \citep{liu2015asynchronous,peng2016arock,wu2023delay}, which updates blocks in parallel and asynchronously. The key difference is that in block coordinate methods, each processing unit is typically assigned to update a single block, whereas in our setting, each machine is responsible for processing a subset of the data across all blocks. We now state the assumptions required for the convergence of the asynchronous algorithm.

      \begin{assumption}[Lipschitz continuous gradient]\label{assumption:lipchitz_continuous_gradient}
        Each $f_j$ and $h$ is twice continuously differentiable with the gradients of them being Lipschitz continuous, i.e., for all $\boldsymbol{x}: = (\boldsymbol{x}_1, \boldsymbol{x}_2)$, $\boldsymbol{x}': = (\boldsymbol{x}_1', \boldsymbol{x}_2')$, there exist constants $L_j$ and $L_h$ such that
        \begin{equation*}\label{eq:lipchitz_continuous_gradient}
          \left\| \frac{\partial f_j(\boldsymbol{x})}{\partial\boldsymbol{x}} - \frac{\partial f_j(\boldsymbol{x}')}{\partial\boldsymbol{x}} \right\| \leq L_{j} \left\| \boldsymbol{x} - \boldsymbol{x}' \right\|, \left\| \frac{\partial h(\boldsymbol{x})}{\partial\boldsymbol{x}} - \frac{\partial h(\boldsymbol{x}')}{\partial\boldsymbol{x}} \right\| \leq L_{h} \left\|\boldsymbol{x} - \boldsymbol{x}' \right\| 
        \end{equation*}
      \end{assumption}

      Assumption \ref{assumption:lipchitz_continuous_gradient} implies that the gradient of the global function $f$ is also Lipschitz continuous. We denote the Lipschitz constant of the gradient of $f$ by $L$ in the following. Obviously, $L \leq J^{-1}\sum_{j=1}^J L_j + L_h$.

      \begin{assumption}[Lipschitz continuous Hessian]\label{assumption:lipchitz_continuous_hessian}
         The Hessians of each $f_j$ and $h$ are Lipschitz continuous, i.e., for all $\boldsymbol{x}: = (\boldsymbol{x}_1, \boldsymbol{x}_2)$, $\boldsymbol{x}': = (\boldsymbol{x}_1', \boldsymbol{x}_2')$, there exist constants $H_j$ and $H_h$ such that
        \begin{equation*}
          \left\| \frac{\partial^2 f_j(\boldsymbol{x})}{\partial\boldsymbol{x}^2} - \frac{\partial^2 f_j(\boldsymbol{x}')}{\partial\boldsymbol{x}^2} \right\| \leq H_{j} \left\| \boldsymbol{x} - \boldsymbol{x}' \right\|, \left\| \frac{\partial^2 h(\boldsymbol{x})}{\partial\boldsymbol{x}^2} - \frac{\partial^2 h(\boldsymbol{x}')}{\partial\boldsymbol{x}^2} \right\| \leq H_{h} \left\| \boldsymbol{x} - \boldsymbol{x}' \right\| ,
        \end{equation*}
        where the norm for the Hessian is the operator norm, i.e., the largest absolute eigenvalue.
      \end{assumption}

      \begin{assumption}\label{assumption:strongly_convex}
        The global function $f$ is strongly convex, i.e., for all $\boldsymbol{x}: = (\boldsymbol{x}_1, \boldsymbol{x}_2)$, there exists a constant $\mu > 0$ such that
        \begin{align*}
          \frac{\partial^2 f(\boldsymbol{x})}{\partial\boldsymbol{x}^2} \succeq \mu \boldsymbol{I},
        \end{align*}
        where the symbol $\succeq$ denotes that if $\boldsymbol{A} \succeq \boldsymbol{B}$, then $\boldsymbol{A} - \boldsymbol{B}$ is positive semidefinite.
      \end{assumption}

      \begin{assumption}\label{assumption:uniform_bound_staleness}
        The staleness is uniformly bounded, i.e., there exists a positive constant integer $\tau$ such that $\tau_{t,j}^{1}, \tau_{t,j}^{21}, \tau_{t,j}^{22} \leq \tau$ for all $t$ and $j$.
      \end{assumption}

      Assumption~\ref{assumption:uniform_bound_staleness} implies that the average staleness across workers is bounded. Specifically, there exists a constant $\overline{\tau} > 0$ such that
        \[
          \frac{1}{J} \sum_{j=1}^J \tau_{t,j}^{1}, \quad 
          \frac{1}{J} \sum_{j=1}^J \tau_{t,j}^{21}, \quad 
          \frac{1}{J} \sum_{j=1}^J \tau_{t,j}^{22} \ \leq \ \overline{\tau},
        \]
        for all $t$. Clearly, $\overline{\tau} \leq \tau$, and in practice it can be much smaller. In the subsequent analysis, we employ $\overline{\tau}$ in place of $\tau$ whenever possible to obtain refined results.

        \begin{remark}
          These assumptions are typical in the convergence analysis of asynchronous algorithms for functions expressed as a sum of local components (see, e.g., \citet{vanli2018global,feyzmahdavian2023asynchronous,fraboni2023general}). In particular, Lipschitz continuity of the gradient and strong convexity are typically imposed to guarantee linear convergence rates. Lipschitz continuity of the Hessian is a common requirement in the analysis of Newton-type methods, as it ensures that a unit step size may be employed in the absence of staleness (cf. \cite{boyd2004convex}). The bounded staleness assumption further guarantees that the information used in parameter updates is not excessively outdated.

For the specific function in this work, \citet{shi2025decentralized} established that, under mild regularity conditions on the data distribution, the gradient is Lipschitz continuous and the objective function is strongly convex in a neighborhood of the optimum with high probability. By analogous arguments, one may also establish the Lipschitz continuity of the Hessian.
          \end{remark}

      Define the Lyapunov sequence as
      \begin{equation}\label{eq:lyapunov_function}
        V_t := f(\boldsymbol{x}_1^t, \boldsymbol{x}_2^t) - \hat{f},
      \end{equation}
      where $\hat{f}$ is the minimal value of $f$. Recall that $\underline{\lambda}$ is the positive threshold for the operator $\operatorname{mod}(\cdot)$. Then, we have the following proposition.

      \begin{proposition}\label{prop:Lyapunov_inequality}
        Suppose Assumptions \ref{assumption:lipchitz_continuous_gradient}, \ref{assumption:strongly_convex}, and \ref{assumption:uniform_bound_staleness} hold, and let $\underline{\lambda} \leq \mu$.  
        If $\alpha_t = \alpha \leq L^{-1} \underline{\lambda}$, then the sequence $\{V_{t}\}$ generated by the basic asynchronous algorithm satisfies
        \begin{equation}\label{eq:Lyapunov_inequality_1}
            V_{t+1} \leq \left(1-\alpha C_0\right) V_t 
            + \alpha^3 \tau \overline{\tau} C_1 \max_{s \in [t-3\tau,\, t]} V_s 
            + \alpha^2 \tau \overline{\tau} C_2 \max_{s \in [t-2\tau,\, t]} V_s,
        \end{equation}
        where $C_0, C_1, C_2$ are constants depending only on $L_j, L_h, \underline{\lambda}$, and $L$ (their explicit forms are provided in Equation~\eqref{eq:constants1} of the Appendix).  
        
        Furthermore, suppose Assumption \ref{assumption:lipchitz_continuous_hessian} also holds. Then, for any $0 < \epsilon < 1$, there exists $T_0$ such that if $\alpha_t = \alpha \leq 1$ for all $t \geq T_0$, the sequence $\{V_{t}\}$ generated by the improved asynchronous algorithm satisfies
        \begin{equation}\label{eq:Lyapunov_inequality_2}
            V_{t+1} \leq \left(1-\alpha C_0'\right) V_t 
            + \alpha^3 \tau \overline{\tau} C_1' \max_{s \in [t-3\tau,\, t]} V_s 
            + \alpha^2 \tau \overline{\tau} C_2' \max_{s \in [t-2\tau,\, t]} V_s,
        \end{equation}
        where $C_0' = (1-\epsilon) C_0$, $C_1' = (1-\epsilon^2) \mu^{-1} \underline{\lambda} C_1$, and $C_2' = C_2$.
        \end{proposition}

        In Equation~\eqref{eq:Lyapunov_inequality_1}, the descent term $-\alpha C_0 V_t$ is linear with respect to $\alpha$, while the staleness terms are of higher order in $\alpha$ (quadratic and cubic). Thus, for sufficiently small $\alpha$, the linear term dominates, leading to linear convergence.
        Under the additional Lipschitz condition on the Hessian, in Equation~\eqref{eq:Lyapunov_inequality_2}, larger step sizes are allowed, with $C_0'$ close to $C_0$, $C_1' < C_1$, and $C_2' = C_2$ for small $\epsilon$ when $\underline{\lambda}<\mu$, resulting in tighter convergence.

        The following lemma gives a more precise characterization of the step size choice and the resulting convergence rate of the Lyapunov sequence $V_t$, whenever $V_t$ satisfies the form of Equations~\eqref{eq:Lyapunov_inequality_1} or \eqref{eq:Lyapunov_inequality_2}.

        \begin{lemma}\label{lem:Lyapunov_contraction}
            Suppose the Lyapunov sequence $V_t$ satisfies Equation~\eqref{eq:Lyapunov_inequality_1}   for $0 < \alpha \leqslant
            \overline{\alpha}$ and $\overline{\tau}, \tau > 0$. If the step size is chosen as  $\alpha = \beta A_0 \left( \overline{\tau} + 1 \right)^{- 1} (\tau
            + 1)^{- 1}$ for some $0 < \beta < 1 $, then
            \[ V_t \leqslant \left[ 1 - \frac{(1 - \beta) \beta }{A_1  + 8 \beta^2}
     \left( \overline{\tau} + 1 \right)^{- 1} (\tau + 1)^{- 1} \right]^t V_0 , \]
  where $A_0 := \min \left\{ C_0 \left( C_2 + \sqrt{C_1 C_0} \right)^{- 1}, \overline{\alpha} \right\}$ and $A_1 := C_0^{- 1} A_0^{- 1}$.
        \end{lemma}
            
        This lemma shows that choosing the step size proportional to 
$(\overline{\tau}+1)^{-1}(\tau+1)^{-1}$ ensures linear convergence of $V_t$, 
with a rate that improves linearly as $(\overline{\tau}+1)^{-1}(\tau+1)^{-1}$ increases.  
When the effect of average staleness is negligible compared with $\tau$, 
the step size simplifies to being proportional to $(\tau+1)^{-1}$, and the convergence rate 
improves linearly with $(\tau+1)^{-1}$.  
This dependence on staleness is consistent with prior results for the incremental aggregated gradient method, 
which is a kind of asynchronous gradient descent method (see, e.g., \citet{vanli2018global,feyzmahdavian2023asynchronous}).

        With Proposition~\ref{prop:Lyapunov_inequality} and Lemma~\ref{lem:Lyapunov_contraction}, we are now ready to state the convergence result for the basic asynchronous algorithm. Since the results with and without the Lipschitz condition on the Hessian differ only in constants, we focus on the case without this assumption ({Assumption~\ref{assumption:lipchitz_continuous_hessian}}) for brevity.

      \begin{theorem}\label{thm:basic_asynchronous_algorithm}
         Suppose Assumptions \ref{assumption:lipchitz_continuous_gradient}, \ref{assumption:strongly_convex}, and \ref{assumption:uniform_bound_staleness} hold, and $\underline{\lambda} \leq \mu$.  
        If $\alpha_t = \beta A_0 \left( \overline{\tau} + 1 \right)^{- 1} (\tau
        + 1)^{- 1} $ for some $0 < \beta < 1 $, then the sequence $\{V_{t}\}$ generated by the basic asynchronous algorithm satisfies
        \[ V_t \leqslant \left[ 1 - \frac{(1 - \beta) \beta }{A_1  + 8 \beta^2}
     \left( \overline{\tau} + 1 \right)^{- 1} (\tau + 1)^{- 1} \right]^t V_0, \]
     where $A_0 := \min \left\{ C_0 \left( C_2 + \sqrt{C_1 C_0} \right)^{- 1}, \overline{\alpha} \right\}$ and $A_1 := C_0^{- 1} A_0^{- 1}$.
\begin{comment}  
     Suppose further that Assumption \ref{assumption:lipchitz_continuous_hessian} holds.  Then, for any $0<\epsilon < 1$, there exists $T_0$ such that if $\alpha_t = \beta A_0' \left( \overline{\tau} + 1 \right)^{- 1} (\tau
        + 1)^{- 1} $ for any $0 < \beta < 1 $, then
        \[ V_t \leqslant \left[ 1 - \frac{(1 - \beta) \beta }{A_1'  + 8 \beta^2}
     \left( \overline{\tau} + 1 \right)^{- 1} (\tau + 1)^{- 1} \right]^t , \]
     where $A_0' := \min \left\{ C_0' \left( C_2' + \sqrt{C_1' C_0'} \right)^{- 1}, 1 \right\}$ and $A_1' := C_0'^{- 1} A_0'^{- 1}$.
\end{comment}
      \end{theorem}

      The convergence of the improved asynchronous algorithm follows as a direct corollary of the above results. For brevity, we present only the case without the Lipschitz condition on the Hessian, since the results under this additional assumption ({Assumption~\ref{assumption:lipchitz_continuous_hessian}}) are analogous. Recall that $T_c$ denotes the iteration at which the weights in the adaptive aggregation scheme are set to be equal, $\omega$ is the exponential decay parameter of the moving average, and $M$ is the number of moving average steps.
      \begin{corollary}\label{cor:improved_asynchronous_algorithm}
        Suppose Assumptions \ref{assumption:lipchitz_continuous_gradient},
        \ref{assumption:strongly_convex}, \ref{assumption:uniform_bound_staleness}
        hold. Then, there exist some staleness-independent positive constants\footnote{Note that the constants for the improved asynchronous algorithm are different from the constants for the basic asynchronous algorithm.} $A_0, A_1$ such that,
        for all $t \geq T_c$, if $\alpha_t = \beta A_0  (\bar{\tau} + 1)^{- 1} 
        (\tau + 1)^{- 1}$ for some $0 < \beta < 1$, \ then the sequence $\{V_{t}\}$ generated by the improved asynchronous algorithm satisfies
        \[ V_t \leqslant \left[ 1 - \frac{(1 - \beta) \beta}{(M+1) (A_1 + 8 \beta^2)
           \sum_{i = 0}^M \omega^M}  (\bar{\tau} + 1)^{- 1} (\tau + 1)^{- 1}
           \right]^{t - T_c} V_{T_c}, \].
      \end{corollary}

      This corollary demonstrates that the choice of step size and the resulting convergence rate of the Lyapunov sequence $V_t$ for the improved asynchronous algorithm are similar to those of the basic asynchronous algorithm. However, our current proof framework does not enable us to quantify the improvements in stability and convergence rate precisely. We leave this as future work and illustrate the improvement empirically in the simulation section.

%%%%%%%%%%%%%%%%%%%Numerical Experiments%%%%%%%%%%
  \section{Numerical Experiments}\label{sec:simulation}
      In this section, we present simulation results to validate our algorithm and theoretical results. The simulations are based on the Matérn covariance function of the form \citep{wang2023parameterization}:
      \begin{equation}
        c(\textbf{s}_1,\textbf{s}_2) = \sigma^2 \frac{2^{1-\nu}}{\Gamma(\nu)} \left(\frac{\sqrt{2\nu}|\textbf{s}_1-\textbf{s}_2|}{\beta}\right)^{\nu} \mathcal{K}_{\nu}\left(\frac{\sqrt{2\nu}|\textbf{s}_1-\textbf{s}_2|}{\beta}\right),
      \end{equation}
      where $\nu$ is the smoothness parameter, $\beta$ is the length scale, $\sigma^2$ is the marginal variance, and $K_{\nu}$ is the modified Bessel function of the second kind. Owing to the complexity of $K_{\nu}$ \citep{geng2025gpu,takekawa2022fast,abramowitz1948handbook}, computing the derivative of the Matérn covariance function for $\nu$ is challenging and lacks efficient, readily available implementations. Consequently, $\nu$ is fixed, and the covariance parameter vector to be estimated is $\boldsymbol{\theta} = (\sigma^2, \beta)^\top$.

      Unless otherwise specified, the following experimental settings are used throughout this study. Spatial locations are generated by adding uniform jitter in the range $[-0.4, 0.4]$ to a regular grid with unit spacing, then shifting and scaling to fit within the unit square $[0,1]^2$. Knot locations are generated using the same procedure. The covariate vector follows a 5-dimensional standard Gaussian distribution, with the corresponding coefficient vector set to $\boldsymbol{\gamma} = (-1, 2, 1, 1, 1)^\top$. Measurement noise is drawn from a Gaussian distribution with zero mean and standard deviation $2$ (corresponding to the true parameter $\delta = 0.25$). The distributed dataset is obtained by randomly partitioning the dataset across workers.

    \subsection{Comparison Between Low-Rank and Independent Models}\label{sec:simulation:comparison_lowrank_independent}
      Here, we compare low-rank and independent models, supported by theoretical results demonstrating that the low-rank model consistently outperforms the independent model. All simulation results are based on $100$ Monte Carlo replications using the full local covariance structure defined in Equation~\eqref{eq:full_local_covariance}. Boxplots are used to illustrate the distribution of parameter estimates, with the sample mean indicated by a green triangle and the true parameter value marked by a red line.

      We begin by examining parameter estimation error under different covariance function parameter settings. The smoothness parameter $\nu$ is set to $0.5$, $1.5$, and $2.5$, where a larger $\nu$ corresponds to a smoother process (Stein, 2012). For each fixed $\nu$, the range parameter $\beta$ is assigned three values such that the effective range, defined as the distance at which the correlation function is $0.05$ of its maximum, is $0.1$, $0.3$, and $0.7$, corresponding to weak, medium, and strong spatial dependence, respectively. The local sample size is set to $100$, the number of workers is set to $10$, and the number of knots is set to $100$.

      \begin{figure}[h]
        \centering
        \includegraphics[width=0.47\textwidth]{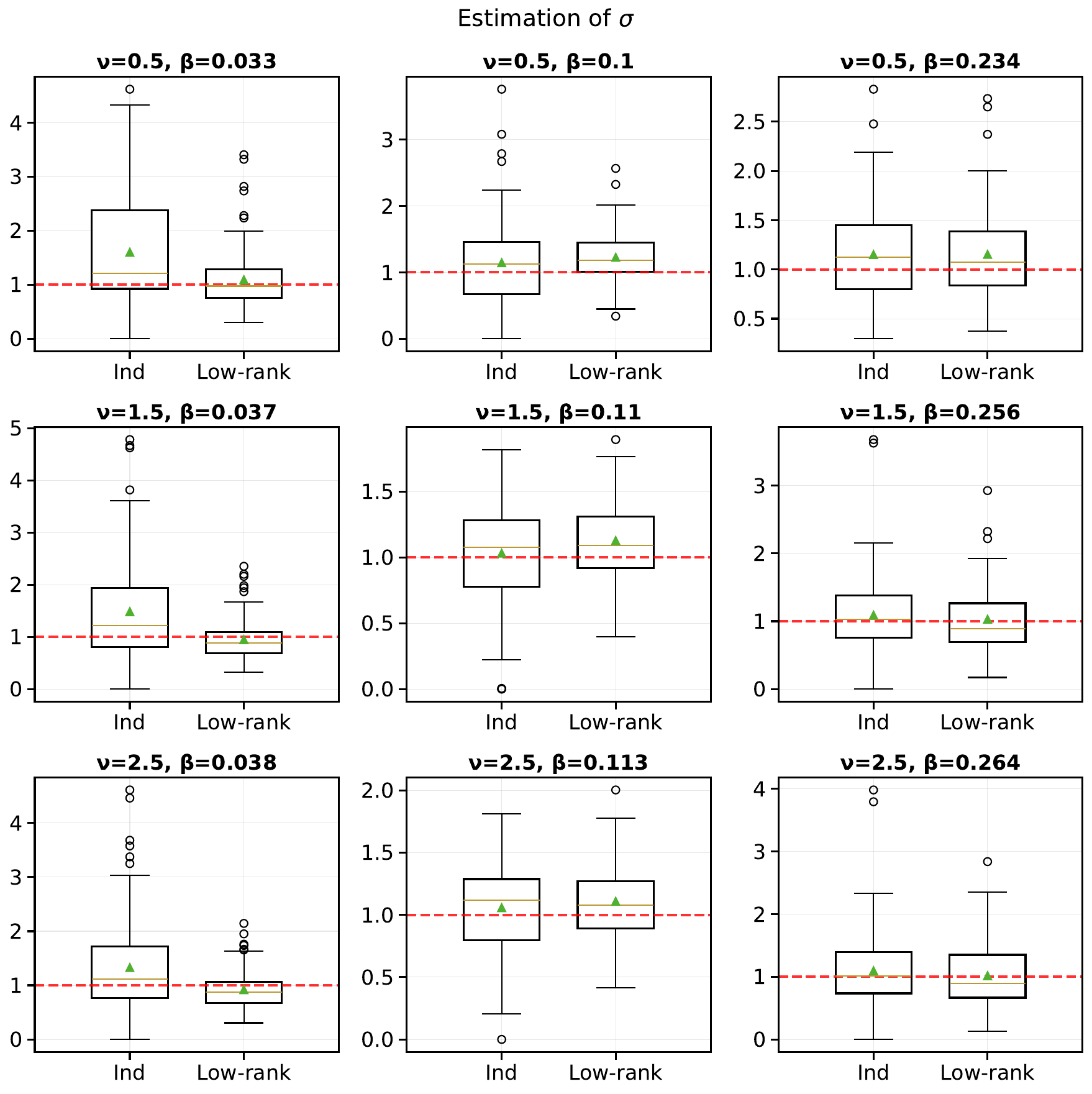}
        \hfill
        \includegraphics[width=0.47\textwidth]{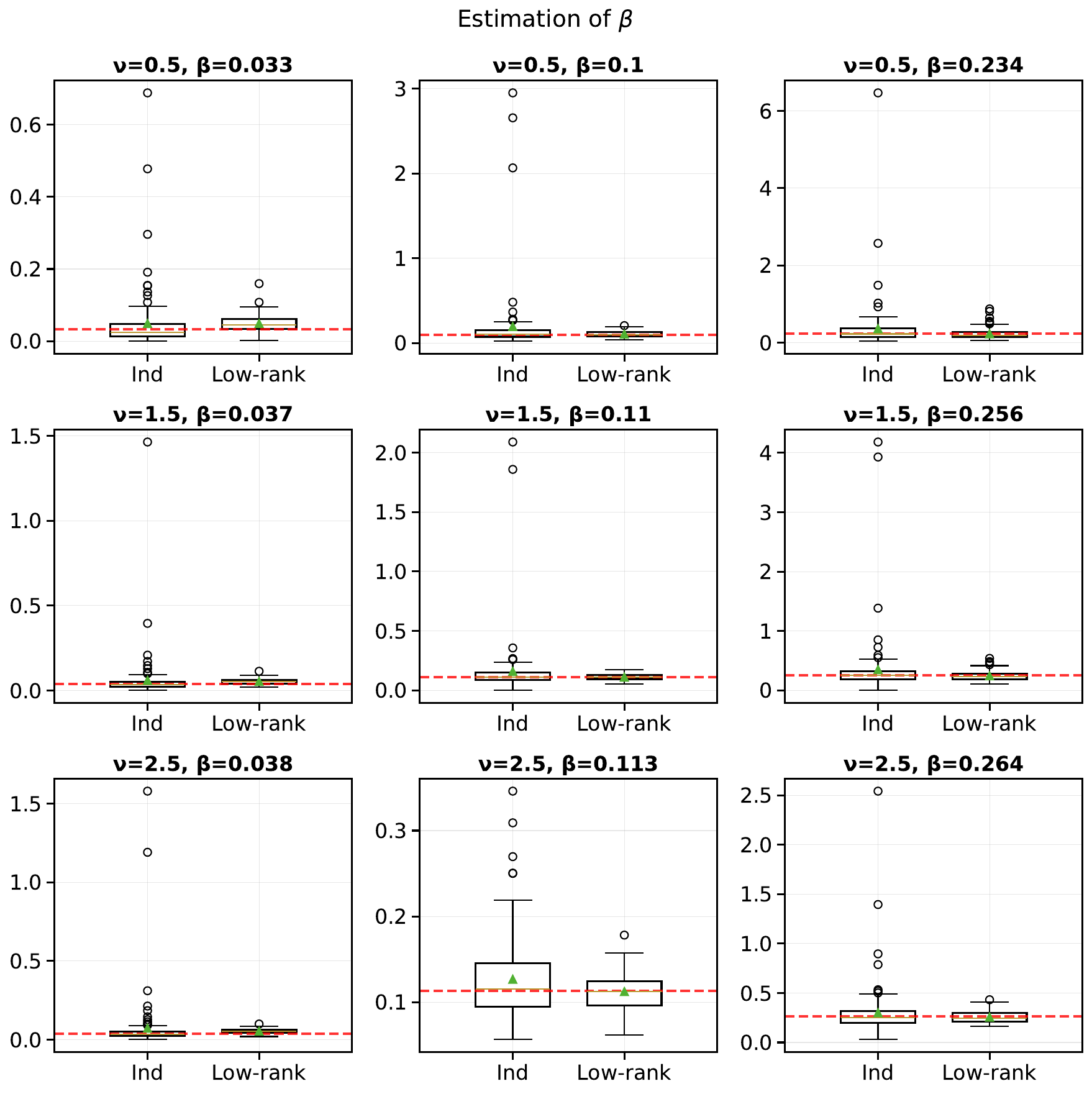}
        \caption{Boxplots of parameter estimates under varying covariance settings: $\sigma$ (left) and $\beta$ (right).}
        \label{fig:boxplot_parameter_settings}

      \end{figure}

      The results are shown in Figure \ref{fig:boxplot_parameter_settings} for the covariance parameter vector $\boldsymbol{\theta}$. In contrast, similar results for other parameters are presented in Figure \ref{fig:boxplot_parameter_settings_other}  in the Appendix for conciseness. From the figure, we observe that the parameter estimates from the low-rank model generally exhibit lower variance compared to those from the independent model, while maintaining comparable or smaller bias across different parameter settings. Moreover, the low-rank model appears more robust, as it produces fewer outliers. This robustness stems from the fact that the independent model's log-likelihood is simply a sum of local log-likelihoods, making it more sensitive to local anomalies. In contrast, the objective function $f$ in Equation~\eqref{eq:low-rank-obj} for the low-rank model includes the additional global parameters $\boldsymbol{\mu}$ and $\boldsymbol{\Sigma}$, along with a shared regularization term $h$, which helps mitigate the influence of outliers.

      We next examine parameter estimation error under varying local sample sizes. In this experiment, we fix $\nu = 2.5$ and $\beta = 0.113$, corresponding to the setting in which the independent model exhibits the fewest outliers in Figure~\ref{fig:boxplot_parameter_settings}. The number of workers is set to $10$, and the local sample size per worker is varied among $(200, 400, 800)$. To maintain comparable computational cost between models while improving the approximation of the covariance structure, the number of knots is set equal to the local sample size. The results for $\sigma$ and $\beta$ are presented in Figure~\ref{fig:boxplot_local_sample_size}, with results for other parameters shown in Figure~\ref{fig:boxplot_local_sample_size_other} in the Appendix for brevity. As the local sample size increases, the estimation errors for both $\sigma$ and $\beta$ decrease. Notably, while the low-rank model achieves accuracy comparable to the independent model in estimating $\sigma$, it significantly outperforms the independent model in estimating $\beta$.

      \begin{figure}[h]
        \centering
        \includegraphics[width=0.47\textwidth]{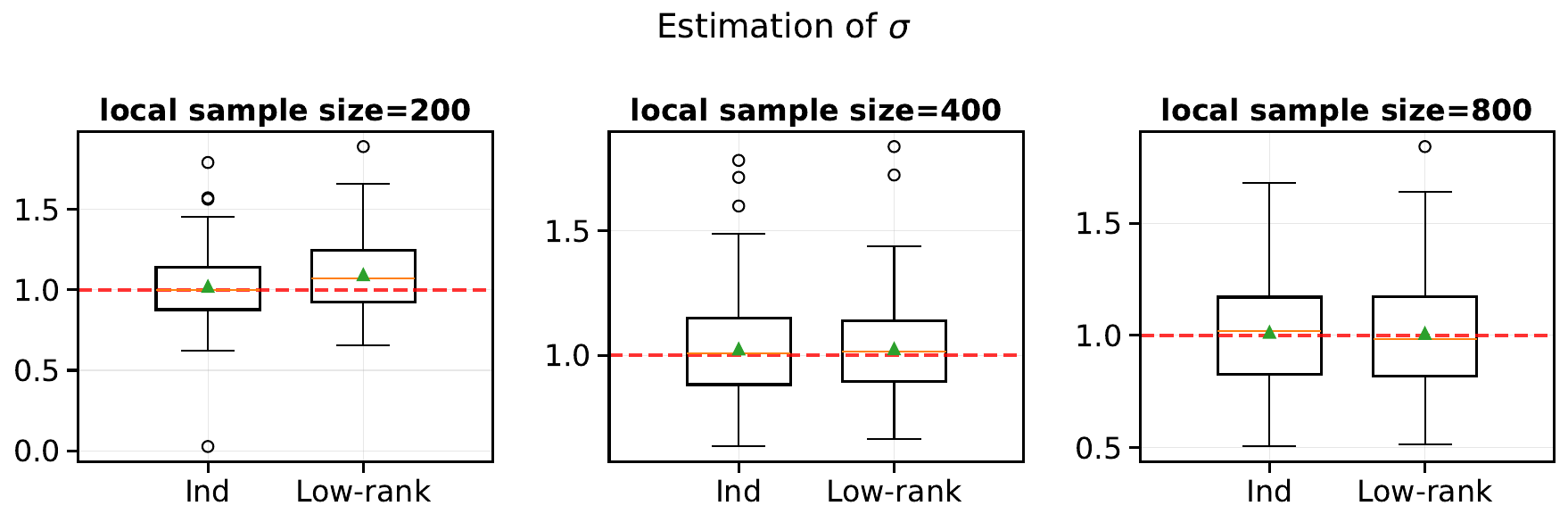}
        \hfill
        \includegraphics[width=0.47\textwidth]{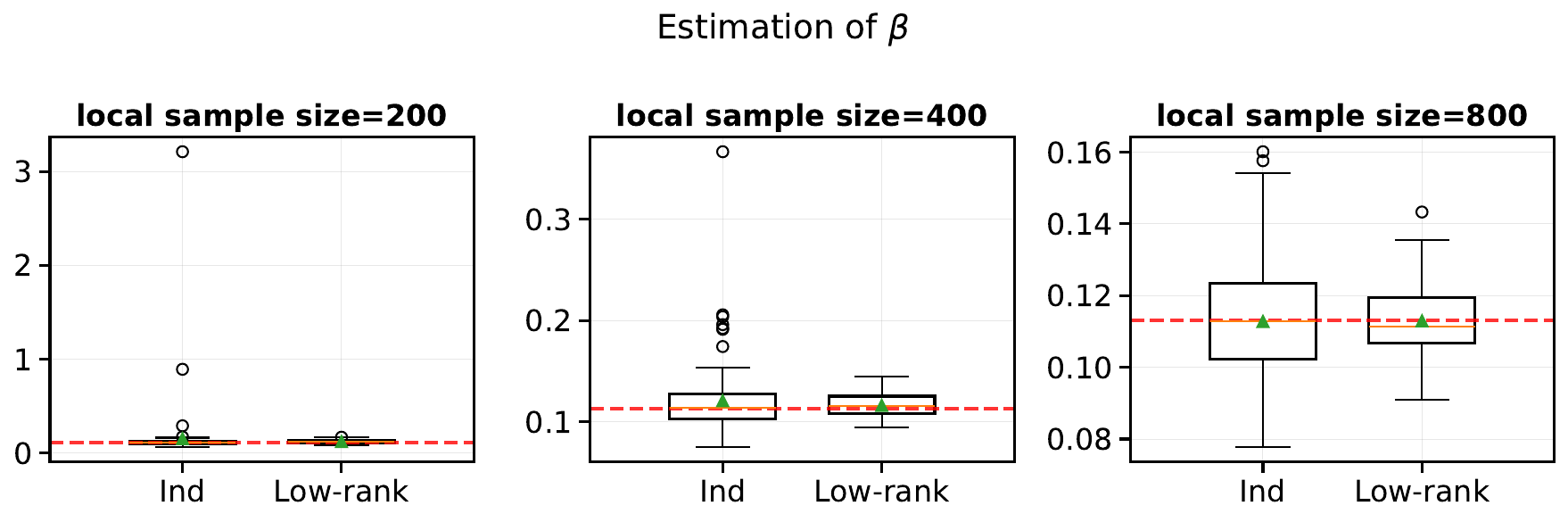}
        \caption{Boxplots of parameter estimates under varying sample sizes: $\sigma$ (left) and $\beta$ (right). }
        \label{fig:boxplot_local_sample_size}

      \end{figure}
      Then, we examine parameter estimation error under varying numbers of workers. In this experiment, similarly, we fix $\nu = 2.5$ and $\beta = 0.113$, set the local sample size to $100$, and vary the number of workers among $(10, 20, 40, 80)$. The number of knots is set to $200$, slightly exceeding the local sample size to account for the larger number of workers. The results for $\sigma^2$ and $\beta$ are presented in Figure~\ref{fig:boxplot_machine_number}, with results for other parameters shown in Figure~\ref{fig:boxplot_machine_number_other} in the Appendix for brevity. As the number of workers increases, the independent model tends to produce more outlier estimates, consistent with the explanation provided in the preceding paragraph.

      \begin{figure}[h]
        \centering
        \includegraphics[width=0.8\textwidth]{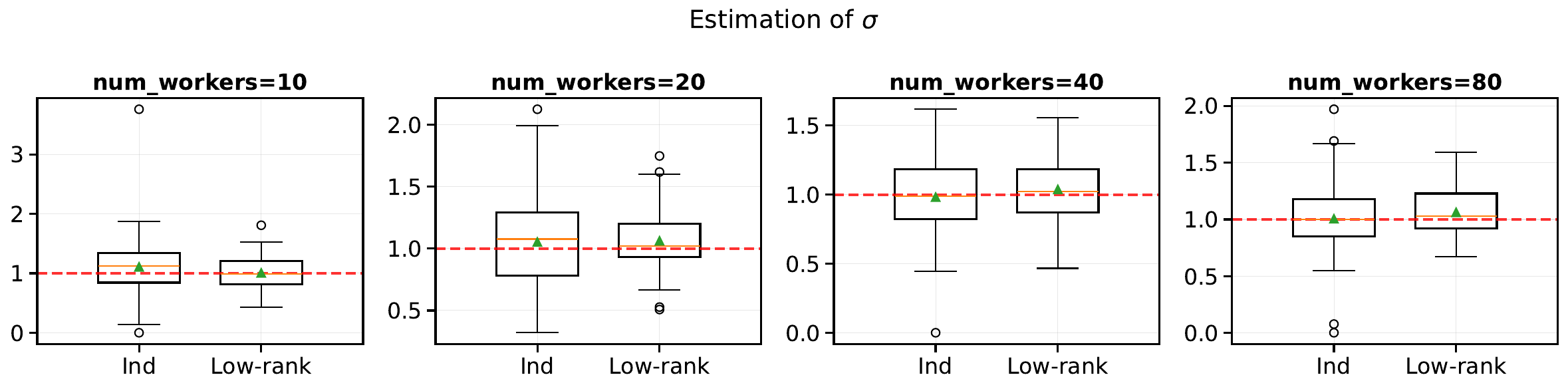}
% \hfill
        \includegraphics[width=0.8\textwidth]{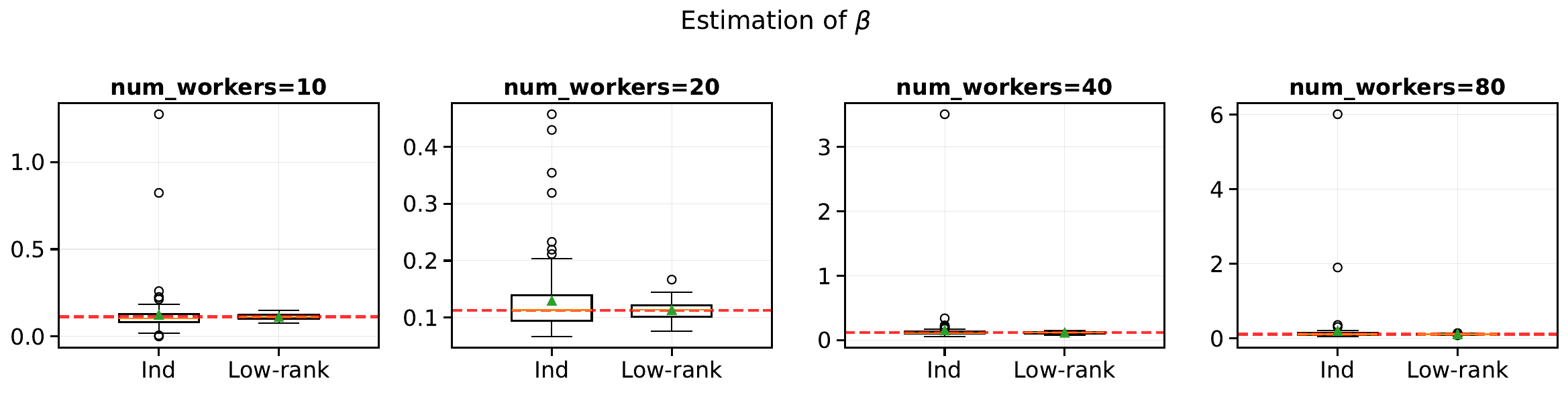}
        \caption{Boxplots of parameter estimates under varying number of workers: $\sigma$ (top) and $\beta$ (bottom). }
        \label{fig:boxplot_machine_number}

      \end{figure}

      \begin{figure}[h]
        \centering
        \includegraphics[width=0.5\textwidth]{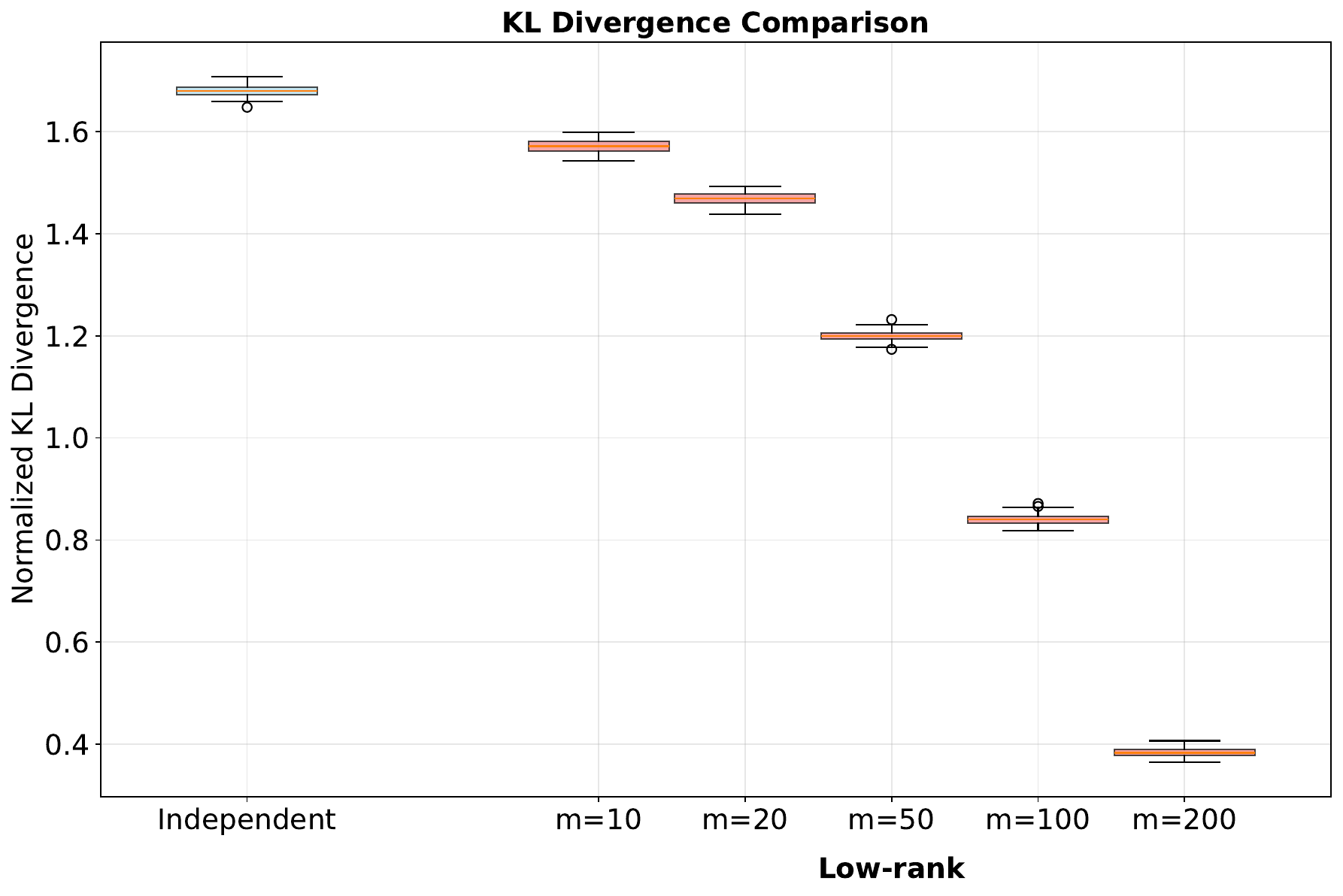}
        \caption{{Boxplots of KL divergence of independent model and low-rank model with varying knot numbers.}}
        \label{fig:kl_divergence}
      \end{figure}

{Finally, to directly verify Proposition~\ref{prop:kl_IDVSLOW}, we compare the KL divergences in a controlled experiment. We fix $\nu = 2.5$ and $\beta = 0.113$, set the local sample size to 100, use 10 workers, and vary the number of knots over $(10, 20, 50, 100, 200)$. The results, shown in Figure~\ref{fig:kl_divergence}, indicate that the KL divergence of the low-rank model is consistently smaller than that of the independent model, in agreement with Proposition~\ref{prop:kl_IDVSLOW}. Furthermore, the KL divergence of the low-rank model decreases as the number of knots increases, which aligns with intuition, even though the theory does not explicitly address this trend.}

      In conclusion, our simulation results show that the low-rank model consistently yields more accurate and robust parameter estimates than the independent model, particularly in settings with limited local data or a large number of workers. This improved performance stems from the use of shared global parameters and a regularized structure, which together mitigate the influence of outliers and enhance the stability of estimation. Overall, these findings underscore the benefits of incorporating low-rank approximations into distributed spatial inference.

    \subsection{Asynchronous vs. Synchronous Low-Rank Models}\label{sec:simulation:comparison_async_sync_lowrank}

      In this subsection, we compare the performance of asynchronous and synchronous federated modeling algorithms under the low-rank model. To assess their real-world performance in a distributed setting, we conduct experiments on Shaheen III (CPU partition), an HPE Cray EX supercomputer at King Abdullah University of Science and Technology (KAUST). Each worker corresponds to a single CPU compute node, with all or some cores per node allocated for computation, enabling us to capture both computation and communication costs in a realistic high-performance computing environment. Each CPU node on Shaheen~III is equipped with two AMD EPYC~9654 (``Genoa'') 96-core CPUs, providing 192 physical cores (384 logical threads), 384 GB of DDR5-4800 MHz memory, and is interconnected via the high-speed, low-latency HPE Slingshot 11 network. Both algorithms are implemented in Python using \texttt{MPI4py} (version 3.1.4) \citep{dalcin2021mpi4py}, compiled against Cray MPICH, the default MPI library on Shaheen III. All jobs are submitted through SLURM, with each MPI process mapped to a dedicated compute node. To emulate heterogeneous computational power among workers, each worker is assigned only a subset of the available cores on a single node.

     Unless specified, the following experimental settings apply throughout this subsection: the low-rank model is based on the full local covariance structure defined in Equation~\eqref{eq:full_local_covariance}; the number of workers is set to $10$, each holding local data with a sample size of $n = 5{,}000$. The kernel parameters $(\nu, \beta) = (0.5, 0.1)$. The hyperparameters are set based on Remark \ref{remark:hyperparameter_tuning}.

      We begin by comparing the performance of the proposed asynchronous algorithm, which integrates all strategies, with variants that employ only a single strategy. These strategies are described in Section~\ref{sec:asynchronous}. For reference, we also include results for the asynchronous algorithm without any strategy and for the baseline synchronous algorithm. The number of cores assigned to each worker follows the configuration shown in Figure~\ref{fig:core_assignment_per_worker}, ranging from a minimum of $3$ to a maximum of $62$, thereby simulating an environment with a significantly slower worker. The aggregation threshold, defined as the number of new local updates the server must receive before performing a global parameter update, is set to $2$, $4$, and $8$. The results, presented in Figure~\ref{fig:async_strategies}, show that the asynchronous algorithm, when run without any strategy or with only a single strategy, exhibits either substantial instability or markedly slower convergence. These findings underscore the importance of combining multiple strategies to achieve stable and efficient asynchronous optimization. Furthermore, an aggregation threshold of $2$ consistently yields faster convergence than thresholds of $4$ and $8$. Based on this observation, we set the aggregation threshold in subsequent experiments to $0.2$ times the number of workers.

      \begin{figure}[h]
  \centering
  \begin{subfigure}[t]{0.38\textwidth}
    \centering
    \includegraphics[width=\linewidth]{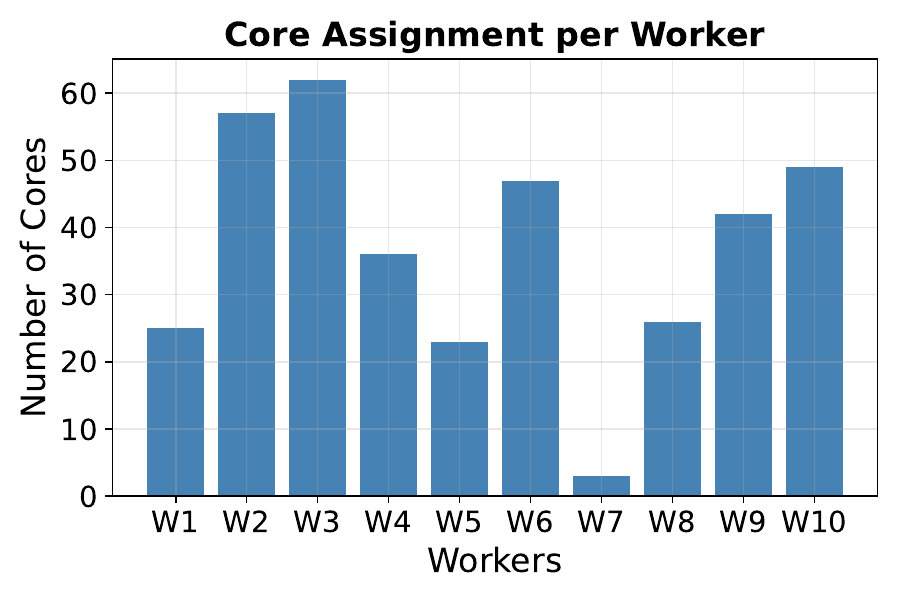}
    \caption{Core assignment per worker}
    \label{fig:core_assignment_per_worker}
  \end{subfigure}

  \vspace{1em}

  \begin{subfigure}[t]{0.92\textwidth}
    \centering
    \includegraphics[width=\linewidth]{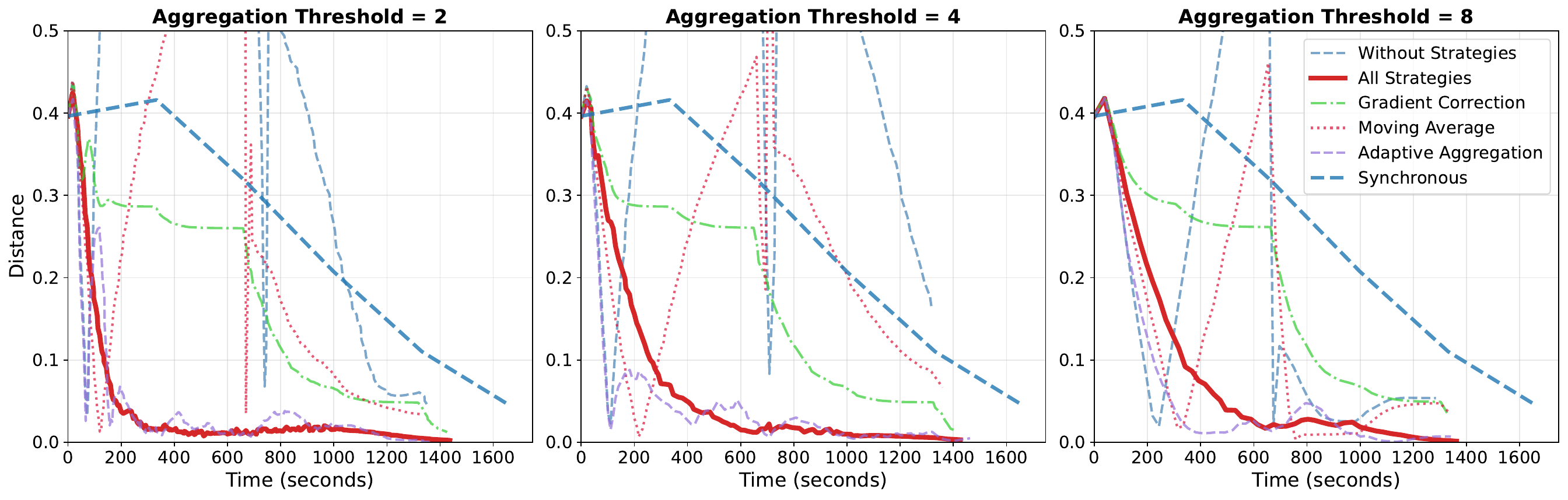}
    \caption{Comparison of asynchronous strategies}
    \label{fig:async_strategies}
  \end{subfigure}
  \caption{Performance comparison of asynchronous strategies under the core assignment.}
  \label{fig:combined_async_core}
\end{figure}

      We next compare the performance of the proposed asynchronous algorithm with the baseline synchronous algorithm under varying degrees of core allocation imbalance. The number of cores assigned to each worker is shown in the upper panel of Figure~\ref{fig:heterogeneity_comparison_linear}, ranging from equal allocations to highly uneven distributions. The corresponding results are presented in the lower panel of Figure~\ref{fig:heterogeneity_comparison_linear}, with additional results for datasets generated using different random seeds provided in Figure~\ref{fig:heterogeneity_comparison_linear_other_replications} (Appendix). As shown in the figures, when the imbalance is low, the asynchronous algorithm converges at a speed that is either slightly slower than or comparable to that of the synchronous algorithm. However, as the imbalance increases, the asynchronous algorithm converges progressively faster relative to its synchronous counterpart.

      \begin{figure}[h]
        \centering
        \includegraphics[width=0.95\textwidth]{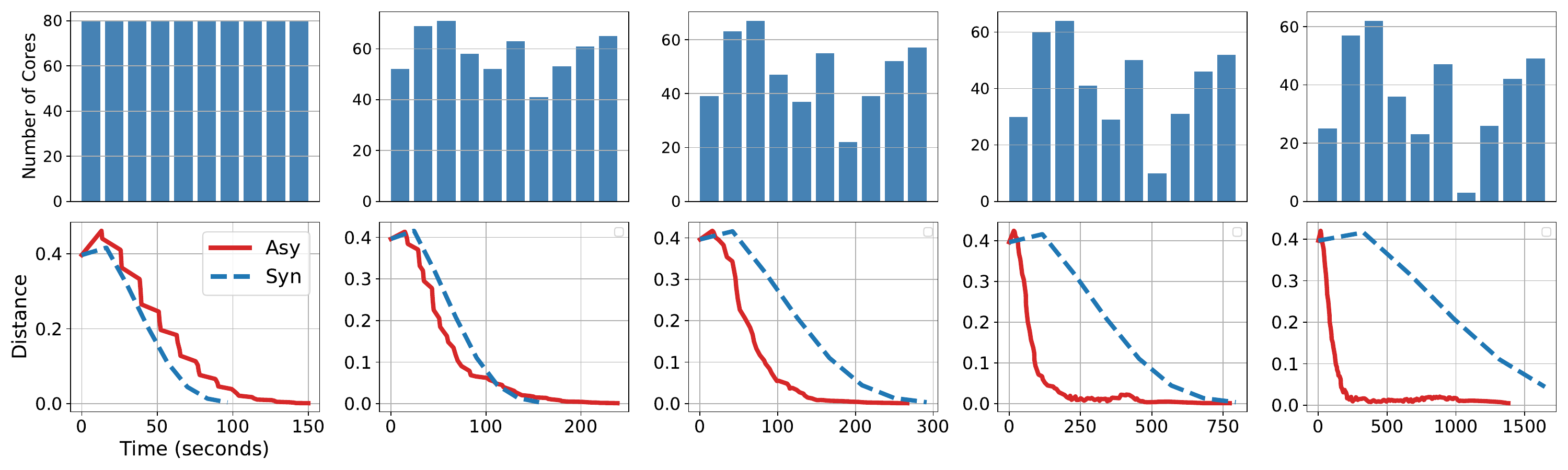}
        \caption{Asynchronous vs. synchronous algorithms under varying degrees of core allocation imbalance.}
        \label{fig:heterogeneity_comparison_linear}
      \end{figure}

      We also evaluate the performance of the proposed asynchronous algorithm in comparison to a baseline synchronous algorithm under heterogeneous local sample sizes. The number of assigned cores per worker varies, and local sample sizes may or may not reflect this variation in computational capacity. To assess the impact of this relationship, we consider three configurations, as illustrated in the upper panel of Figure~\ref{fig:unequal_sample_size}: 1) a core-aligned setting, where workers with more cores are assigned more samples; 2) an unaligned setting, where sample sizes vary independently of core counts; and 3) a core-inverse setting, where workers with fewer cores are assigned more samples. The corresponding convergence results are shown in the lower panel of Figure~\ref{fig:unequal_sample_size}. As the misalignment between core counts and local sample sizes increases, the convergence advantage of the asynchronous algorithm over the synchronous one becomes more pronounced. This is primarily because the per-iteration computation time of the synchronous algorithm is dictated by the slowest worker, which becomes increasingly slower as the mismatch increases. In contrast, the per-iteration time of the asynchronous algorithm is largely unaffected by this unalignment, and its convergence rate per iteration remains relatively stable.

      \begin{figure}[h]
        \centering
        \includegraphics[width=0.75\textwidth]{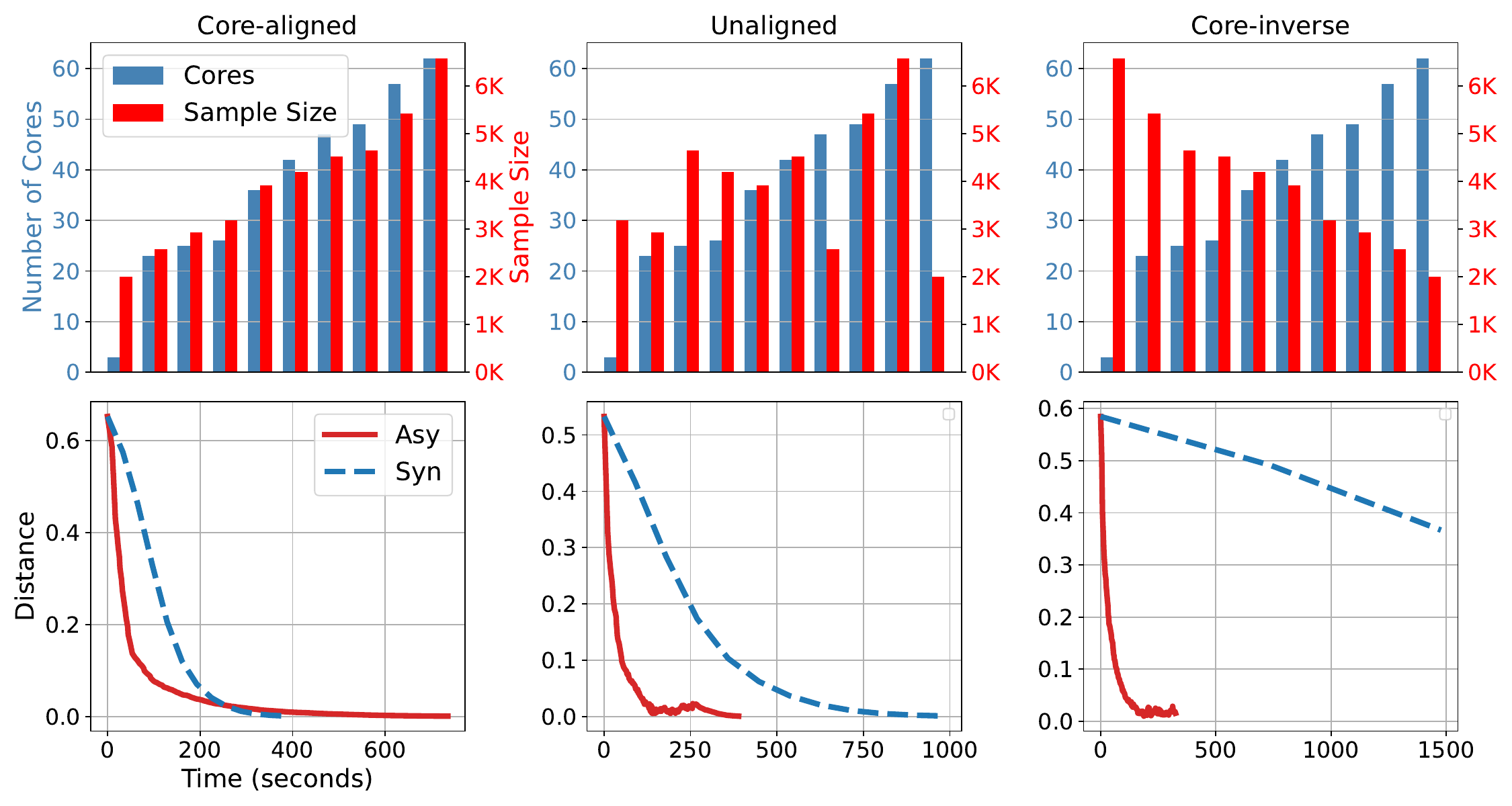}
        \caption{Comparison of asynchronous and synchronous algorithms under unequal local sample sizes.}
        \label{fig:unequal_sample_size}
      \end{figure}

      We further compare the performance of the proposed asynchronous algorithm with the baseline synchronous algorithm for varying numbers of workers: $5$, $20$, and $40$. The number of cores assigned to each worker follows the configuration shown in the upper panel of Figure~\ref{fig:varying_num_workers}, reflecting the fact that as the number of workers increases, the likelihood of some receiving fewer cores also increases, leading to greater computational imbalance. The convergence results are presented in the lower panel of Figure~\ref{fig:varying_num_workers}. As shown, the asynchronous algorithm consistently converges across all settings. Moreover, as the number of workers increases and the imbalance becomes more pronounced, the asynchronous algorithm converges faster than the synchronous algorithm.

      \begin{figure}[h]
        \centering
        \includegraphics[width=0.8\textwidth]{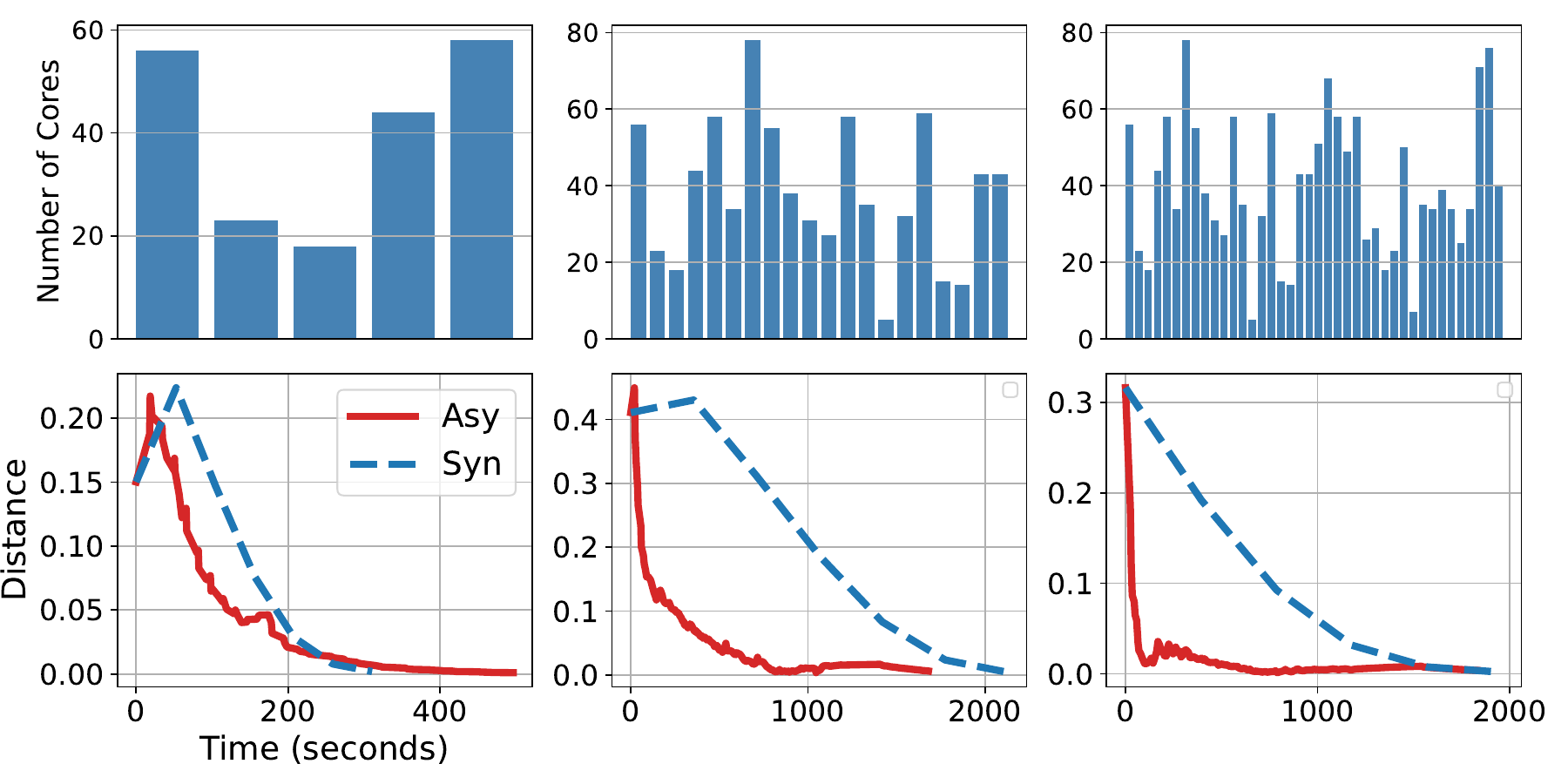}
        \caption{Comparison of asynchronous and synchronous algorithms under varying numbers of workers.}
        \label{fig:varying_num_workers}
      \end{figure}

      Moreover, we evaluate how the prediction error evolves over time for the proposed asynchronous algorithm compared to the baseline synchronous algorithm, under varying degrees of core allocation imbalance. The core assignments, shown in the upper panel of Figure~\ref{fig:prediction_error}, range from balanced to highly imbalanced configurations. To improve predictive performance, the number of knots is set to 500, slightly larger than in previous experiments. Prediction error is computed using 100 out-of-sample data points. The results, presented in Figure~\ref{fig:prediction_error}, show that the prediction error for both algorithms decreases over time. When the core imbalance is low, the asynchronous algorithm achieves a prediction error reduction rate that is comparable to or slightly slower than that of the synchronous algorithm. However, as the imbalance increases, the asynchronous algorithm exhibits a progressively faster reduction in prediction error relative to the synchronous counterpart, highlighting its advantage in heterogeneous computing environments.

      \begin{figure}[h]
        \centering
        \includegraphics[width=0.95\textwidth]{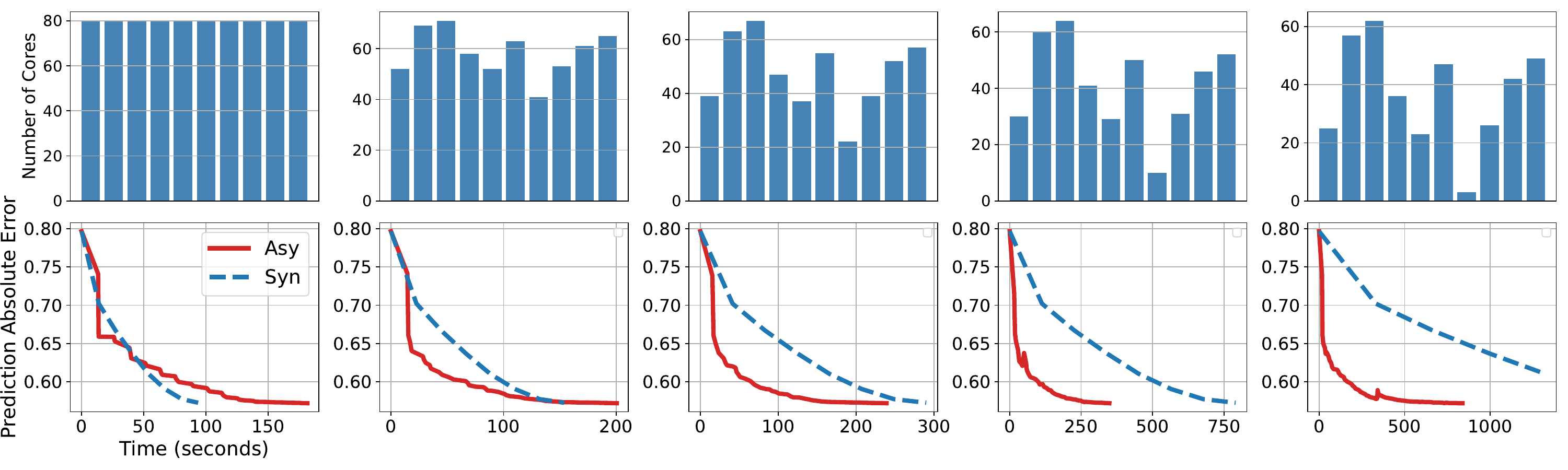}
        \caption{Prediction error comparison of asynchronous and synchronous algorithms under different degrees of core allocation imbalance.}
        \label{fig:prediction_error}
      \end{figure}

      Finally, we evaluate the robustness of the proposed asynchronous algorithm in comparison to the baseline synchronous algorithm under various experimental conditions, including different local sample sizes, kernel parameters, and data partitioning schemes.

      We first vary the local sample size per worker, considering values of 6,000, 8,000, and 10,000. The core allocation per worker follows the setup illustrated in the leftmost subfigure of Figure~\ref{fig:varying_num_workers}, where the maximum number of cores is slightly increased to ensure computation time remains within a reasonable range for larger datasets. The convergence results are shown in the three right subfigures of Figure~\ref{fig:varying_sample_sizes}. Across all sample sizes, the asynchronous algorithm achieves faster convergence than the synchronous algorithm, highlighting its robustness to changes in local data volume.

      \begin{figure}[h]
        \centering
        \includegraphics[width=0.95\textwidth]{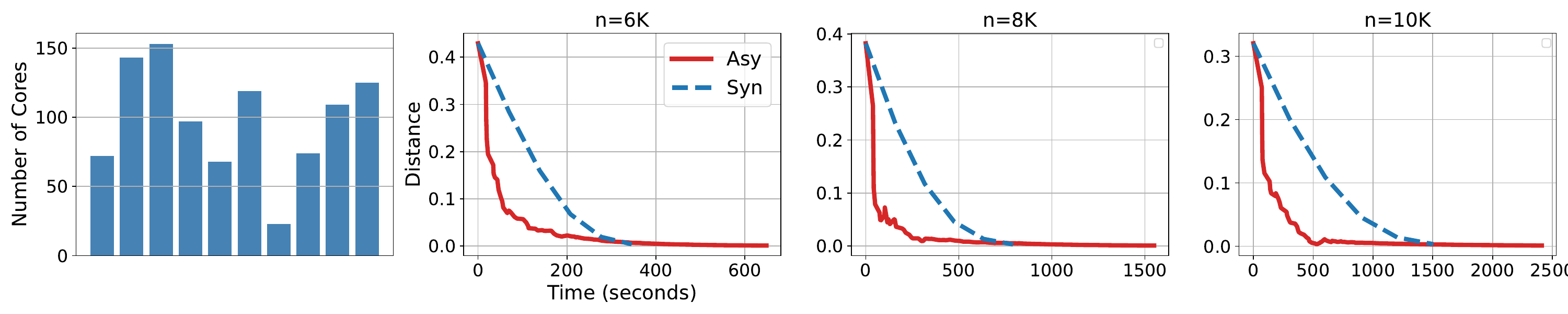}
        \caption{Comparison of asynchronous and synchronous algorithms under varying local sample sizes per worker.}
        \label{fig:varying_sample_sizes}
      \end{figure}

   Next, we examine the performance of the algorithms under varying kernel parameters, using the core assignment setup shown in Figure~\ref{fig:core_assignment_per_worker}. We consider two values for the smoothness parameter, $\nu \in \{0.5, 1.5\}$, and for each $\nu$, three values of the range parameter $\beta$, yielding six kernel settings: $(0.5, 0.033)$, $(0.5, 0.1)$, $(0.5, 0.234)$, $(1.5, 0.021)$, $(1.5, 0.063)$, and $(1.5, 0.148)$. A larger $\nu$ corresponds to smoother spatial processes (Stein, 2012), while a larger $\beta$ indicates a longer range of spatial correlation. The selected $\beta$ values produce effective spatial ranges of approximately 0.1, 0.3, and 0.7, respectively. As shown in Figure~\ref{fig:varying_kernel_parameters}, the asynchronous algorithm consistently achieves faster convergence across all kernel configurations, demonstrating strong robustness to changes in spatial dependence characteristics.

      \begin{figure}[h]
        \centering
        \includegraphics[width=0.8\textwidth]{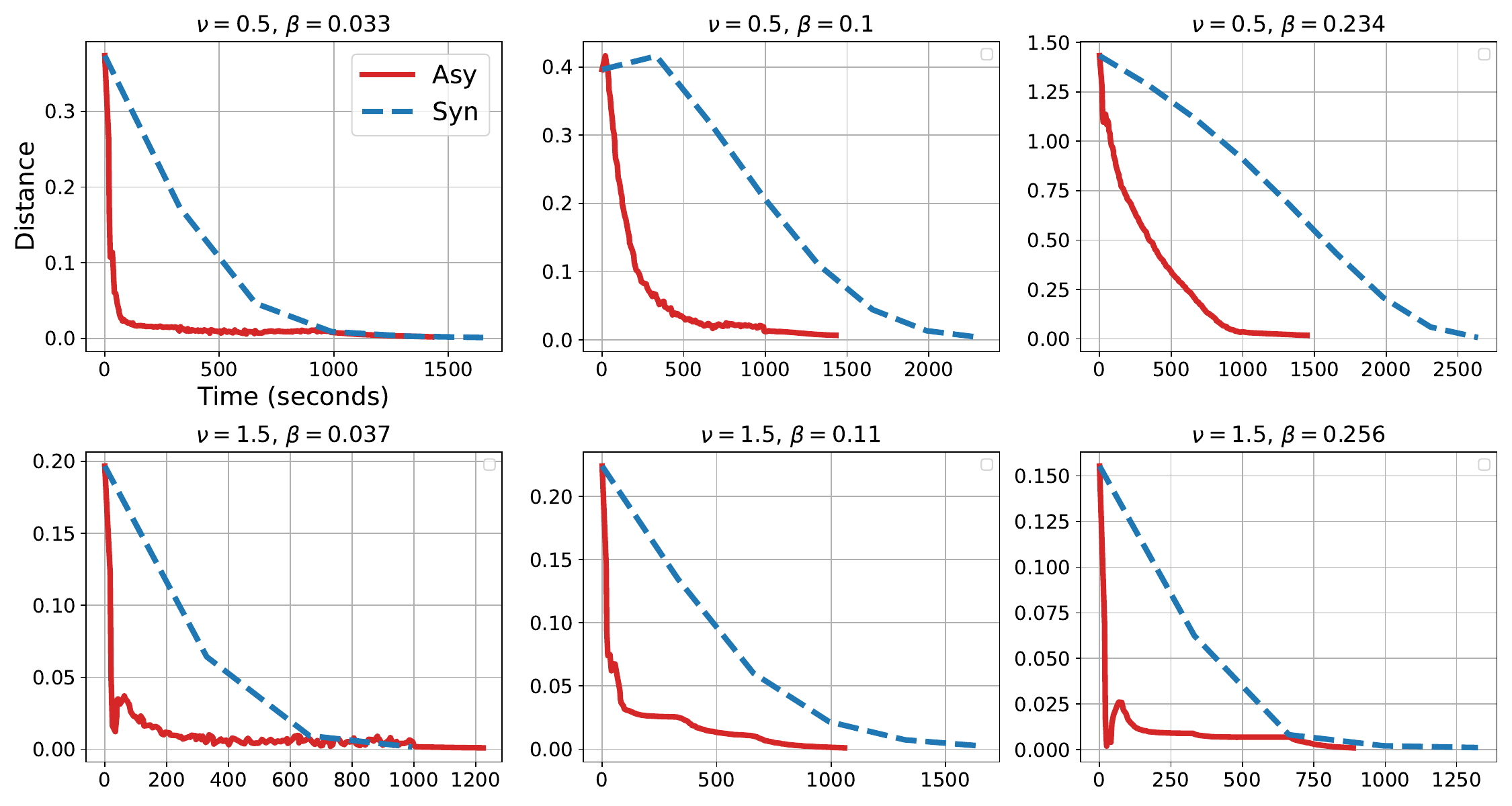}
        \caption{Comparison of asynchronous and synchronous algorithms under varying kernel parameters.}
        \label{fig:varying_kernel_parameters}
      \end{figure}

      Lastly, we assess the impact of different data partitioning strategies, again using the core assignments in Figure~\ref{fig:core_assignment_per_worker}. Three partitioning schemes are considered: 1) random partitioning, where data points are distributed uniformly without considering spatial location; 2) area-based partitioning, which groups data by spatial proximity; and 3) random neighboring, a hybrid method where each subset is formed by random selection and then augmented with 99 spatially neighboring points for each selected point. The local sample size is fixed at 5,000 per worker. We use 10 workers for the random and random neighboring schemes, and 9 workers for the area-based scheme. As shown in Figure~\ref{fig:partition}, the asynchronous algorithm outperforms the synchronous counterpart under all three partitioning strategies, further demonstrating its robustness to heterogeneous data distributions.

      \begin{figure}[h]
        \centering
        \includegraphics[width=0.8\textwidth]{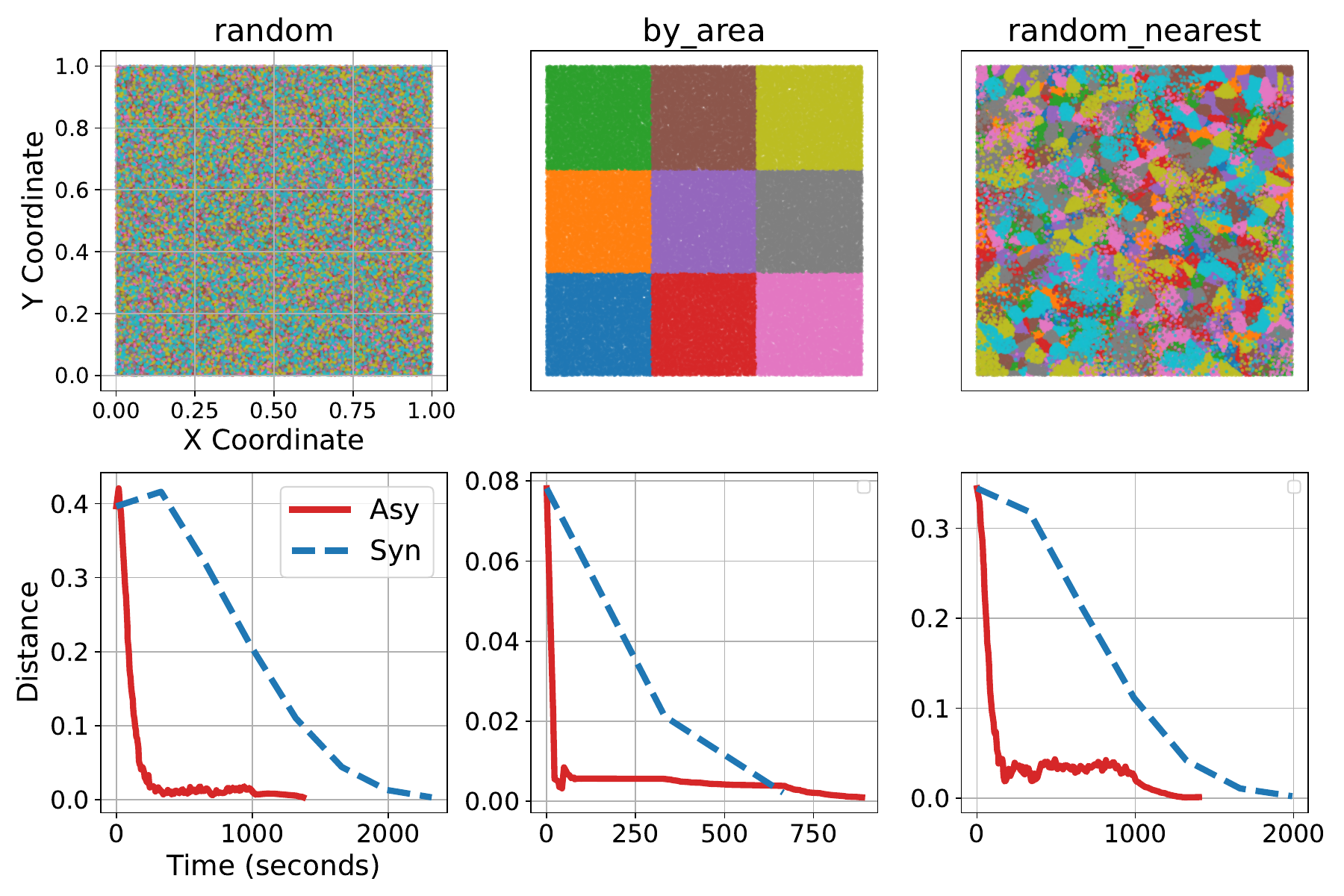}
        \caption{Comparison of asynchronous and synchronous algorithms under different data partitioning ways.}
        \label{fig:partition}
      \end{figure}

      In summary, our simulation results show that the proposed asynchronous algorithm achieves comparable or slightly worse performance when computational resources are well balanced. However, in scenarios with significant computational imbalance, whether caused by heterogeneous processing power or unequal local sample sizes, the asynchronous algorithm consistently outperforms its synchronous counterpart. This performance advantage holds across various parameter settings, sample sizes, and data partitioning schemes, highlighting the robustness and practical benefits of the proposed asynchronous algorithm.

  \section{Conclusion}

  In this article, we present an asynchronous federated modeling framework for spatial data using low-rank Gaussian process approximations. By adopting a block-wise optimization scheme and addressing key challenges of asynchronous updates, including gradient inconsistencies, staleness, and update instability, our method enables efficient and scalable inference in distributed and heterogeneous environments.
  We established theoretical guarantees for the proposed algorithm, showing that it achieves linear convergence with an explicit characterization of its dependence on staleness. This result not only supports the practical implementation of asynchronous federated modeling but also provides insights that may be of potential standalone theoretical significance in optimization theory.
  Extensive numerical experiments demonstrate that the asynchronous algorithm performs comparably to the synchronous approach under balanced computational resources and consistently outperforms it in scenarios with heterogeneous processing power or unequal data partitioning. These findings highlight the robustness, scalability, and practical utility of the proposed method for distributed spatial inference. {Specifically,} our analysis presumes bounded staleness across workers; linear convergence holds when the step size is chosen appropriately, i.e., scaling inversely with the maximum staleness $\tau$, and the rate improves as this bound tightens. In practice, our asynchronous recipe, i.e., local gradient correction, staleness-aware aggregation, and a moving average of parameters, stabilizes updates and maintains accuracy. A step size proportional to $1/\tau$ is a robust default in our setting. Practically, the method matches synchronous baselines on balanced hardware and wins under heterogeneity, reflecting the compute/communication trade-off that favors asynchronous when stragglers are present.

  Future work may proceed in several directions to further strengthen the proposed asynchronous federated modeling framework. First, a precise theoretical analysis is needed to quantify the gains of the improved asynchronous algorithm over the basic version, particularly in terms of convergence rate and stability. Second, extending the framework to multi-resolution low-rank models \citep{nychka2015multiresolution,katzfuss2017multi} would enable the capture of spatial dependencies across multiple scales more effectively, thereby improving approximation accuracy. {Third}, exploring fully decentralized communication architectures \citep{gabrielli2023survey,shi2025decentralized} would remove the reliance on a central server, thereby enhancing scalability, robustness, and applicability in distributed spatial inference. {Fourth, while we currently place knots on a predefined grid, adapting data-dependent methods such as variational optimization \citep{titsias2009variational} or support points \citep{song2025large} to the federated setting could increase the approximation accuracy of the low-rank model. Fifth, adopting hierarchical or tile low-rank structures \citep{mondal2023tile,salvana2022parallel,abdulah2018parallel} for local covariance matrices could further reduce computation costs on individual workers while maintaining accuracy. }

\acks{
This work is based upon work supported by King Abdullah University of Science and
Technology Research Funding (KRF) under Award No. ORFS-2022-CRG11-5069.}

  \appendix
  \section{}\label{sec:appendix}
%the numbering of the appendix should in the way of A.1, A.2, A.3, etc. and the figure and table should be numbered in the way of A.1, A.2, A.3, etc beginning from 1.
\renewcommand{\thesubsection}{A.\arabic{subsection}}
      \setcounter{figure}{0}
\renewcommand{\thefigure}{A.\arabic{figure}} %the figure should be numbered in the way of A.1, A.2, A.3, etc beginning from 1.
\setcounter{equation}{0}
\renewcommand{\theequation}{A\arabic{equation}}

    \subsection{Derivation of the objective in Equation \eqref{eq:low-rank-obj}.}\label{sec:appendix:additional_details}
      In this subsection, we provide the derivation of the objective function $f$ in Equation \eqref{eq:low-rank-obj}. The derivation is based on the evidence lower bound (ELBO) in variational inference, specifically, let \( q : \mathbb{R}^m \to \mathbb{R} \) be a density function, the ELBO is defined as a function of  $q$:
      \[
      \text{ELBO} (q) := \int q (\boldsymbol{\eta}) \log \frac{p (\boldsymbol{z}
      | \boldsymbol{\eta}  ) p (\boldsymbol{} \boldsymbol{\eta})}{q
        (\boldsymbol{\eta})} \mathrm{d} \boldsymbol{\eta} =\mathbb{E}_q \log p (\boldsymbol{z} |
        \boldsymbol{\eta}  ) -  \text{KL} (q(\boldsymbol{\eta}) \| p (\boldsymbol{\eta})). \]
        According to Jensen's inequality, $ \text{ELBO}$ is a lower bound of the
        log-likelihood, that is, $\text{ELBO} \leqslant \log p(\boldsymbol{z})$, with equality achieved if and only if $q (\boldsymbol{\eta}) = p (\boldsymbol{}
        \boldsymbol{\eta} | \boldsymbol{z})$. Since $p (
        \boldsymbol{\eta} | \boldsymbol{z})$ is a Gaussian density, $q (\boldsymbol{\eta})$ can be parameterized by a mean vector $\boldsymbol{\mu}$ and a covariance matrix $\boldsymbol{\Sigma}$. Additionally, with  the conditional independence of $\boldsymbol{z}_1,\ldots,\boldsymbol{z}_J$ given $\boldsymbol{\eta}$, we have  $\log p (\boldsymbol{z} |
        \boldsymbol{\eta}  ) = \sum_j \log p (\boldsymbol{z}_j | \boldsymbol{\eta}
        )$. This leads to \vspace{-10pt}
        \begin{equation}\label{eq:ob1}
          \log p(\boldsymbol{z}) = \max_{
            \boldsymbol{\mu}, \boldsymbol{\Sigma}} \left\{ \sum_j \mathbb{E}_{\mathcal{N} (\boldsymbol{\eta} |
            \boldsymbol{\mu}, \boldsymbol{\Sigma}  )} \left\{ \log p (\boldsymbol{z}_j |
            \boldsymbol{\eta}  )\right\}-  \text{KL} (\mathcal{N} (\boldsymbol{\eta} |
            \boldsymbol{\mu}, \boldsymbol{\Sigma}  ) \|  p (\boldsymbol{\eta})) \right\},\vspace{-5pt}
          \end{equation}
          where $\delta=\tau^{-2}$.  The negative ELBO on the right-hand side of Equation \eqref{eq:ob1} is exactly the objective function $f$ in Equation \eqref{eq:low-rank-obj}.

    \subsection{Additional Simulation Results}\label{sec:appendix:additional_simulation_results}

          In this subsection, we provide additional simulation results to supplement the main text. Specifically, Figures~\ref{fig:boxplot_parameter_settings_other}, \ref{fig:boxplot_local_sample_size_other}, and \ref{fig:boxplot_machine_number_other} present boxplots of parameter estimates for additional parameters under varying covariance settings, local sample sizes, and numbers of workers, respectively. For the parameter $\boldsymbol{\gamma}$, we summarize the effect of the five covariates using their average, defined as $\gamma_{\text{avg}} := \sum_{i=1}^5 \gamma_i/5$, allowing the five-dimensional vector to be visualized in a single plot.

\vspace{-10pt}
          \begin{figure}[h]
            \centering
            \includegraphics[width=0.43\textwidth]{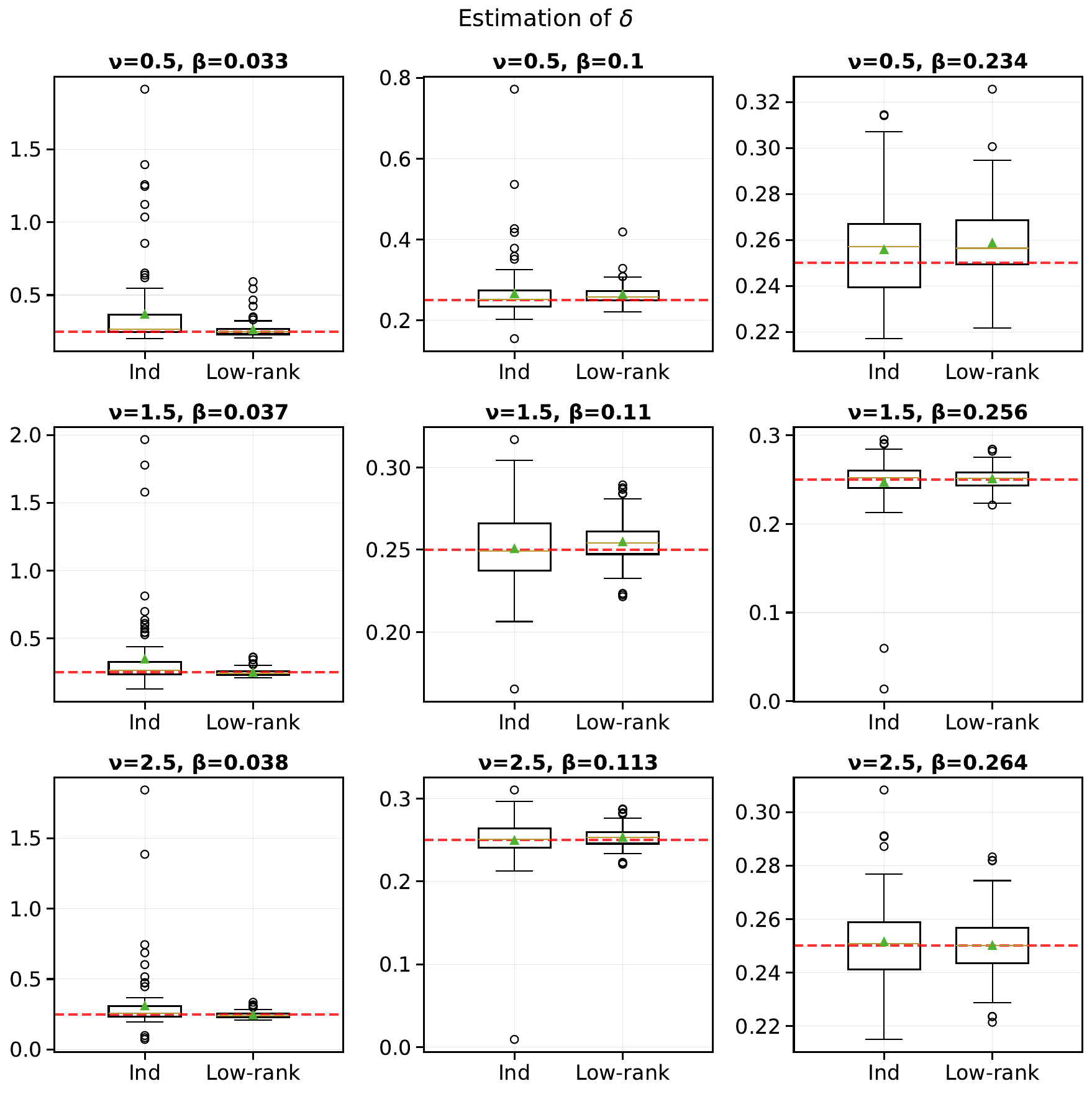}
            \includegraphics[width=0.43\textwidth]{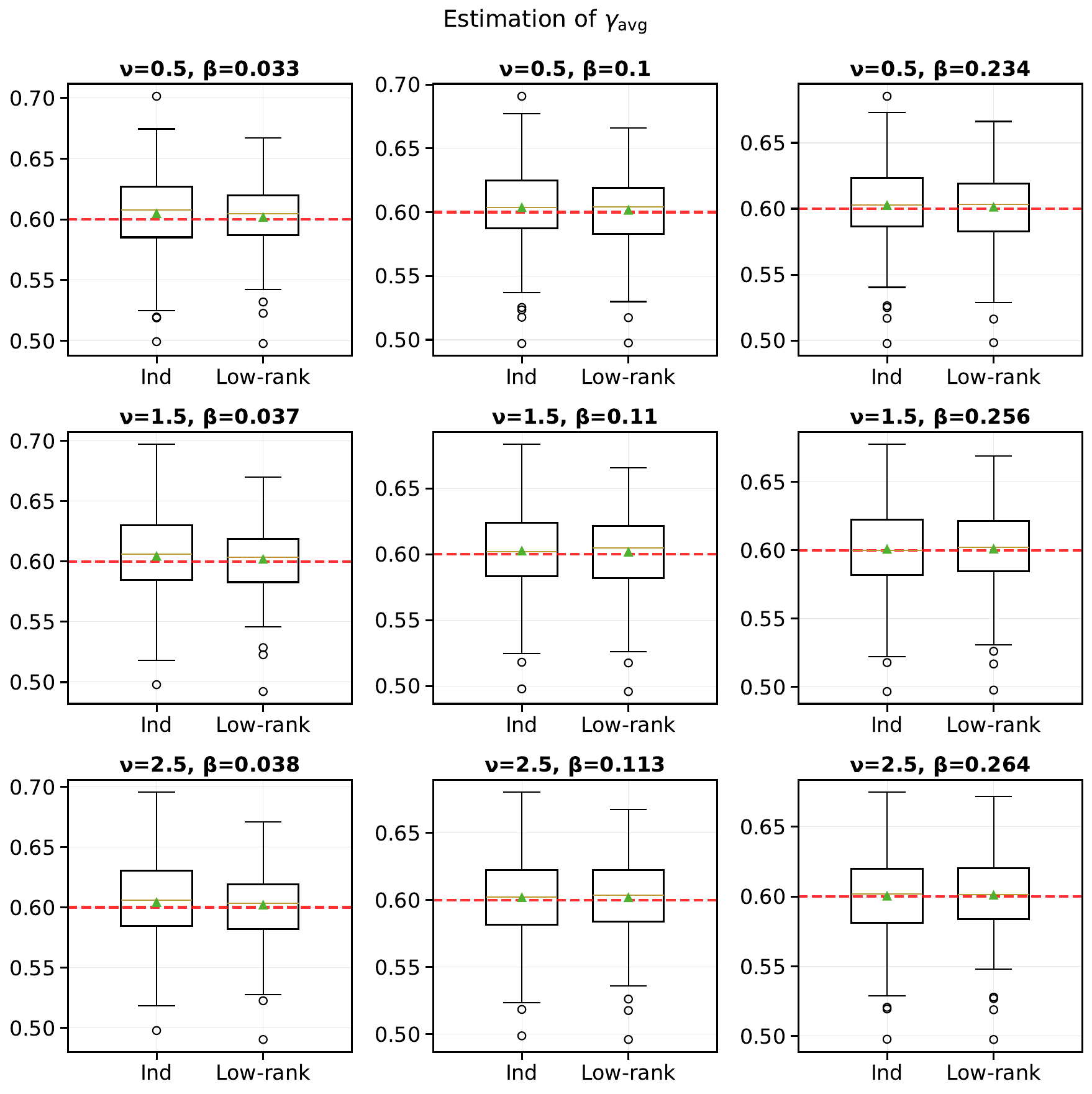}
            \caption{Boxplots of parameter estimates under varying covariance function settings: $\delta$ (left) and $\gamma_{\text{avg}}:=\sum_{i=1}^5 \gamma_i/5$ (right). }
            \label{fig:boxplot_parameter_settings_other}

          \end{figure}

          \begin{figure}[h]
            \centering
            \includegraphics[width=0.43\textwidth]{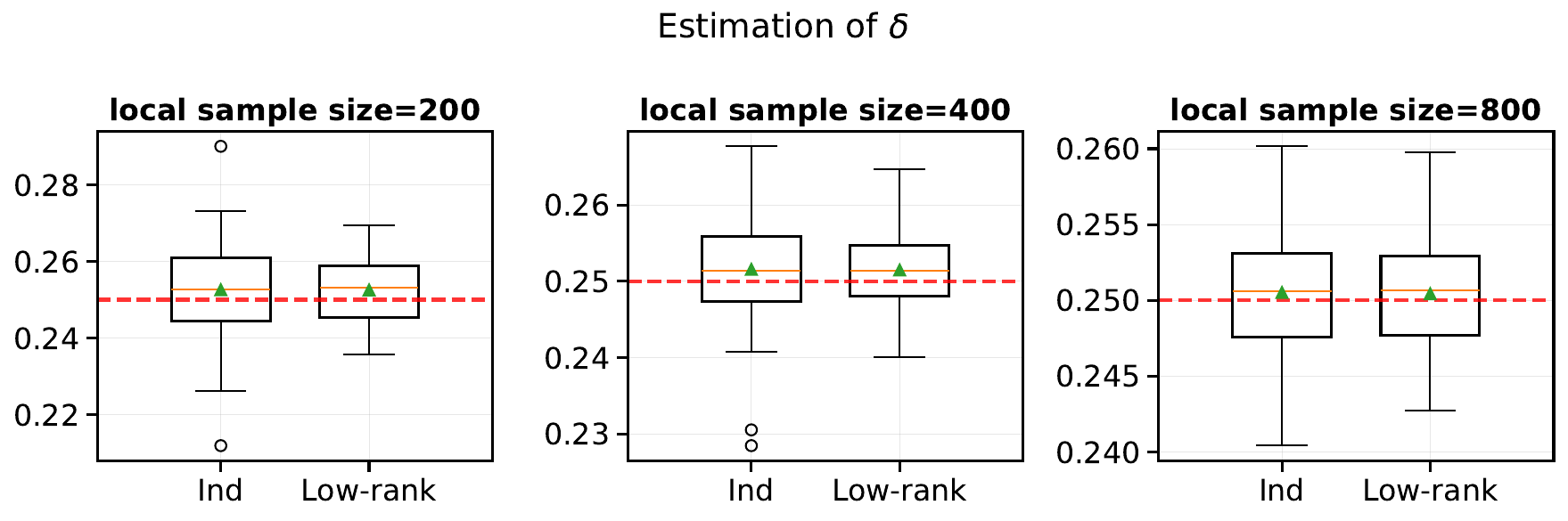}
            \includegraphics[width=0.43\textwidth]{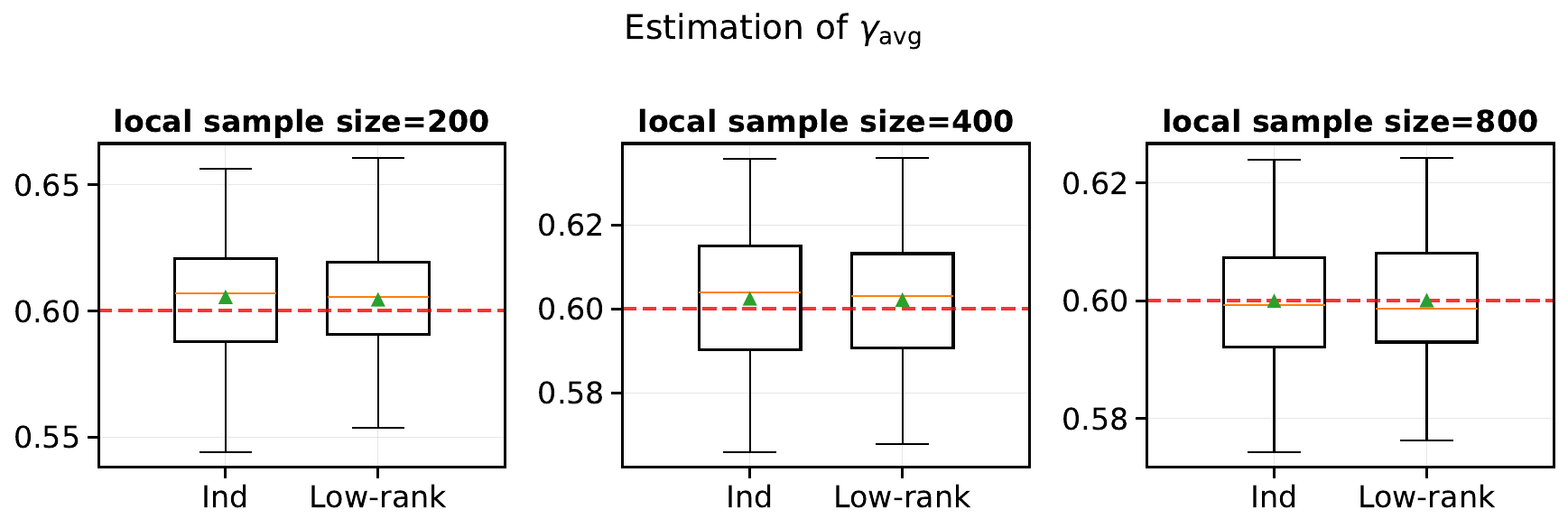}
            \caption{Boxplots of parameter estimates under varying local sample size settings: $\delta$ (left) and $\gamma_{\text{avg}}:=\sum_{i=1}^5 \gamma_i/5$ (right). }
            \label{fig:boxplot_local_sample_size_other}

          \end{figure}

          \begin{figure}[h]
            \centering
            \includegraphics[width=0.43\textwidth]{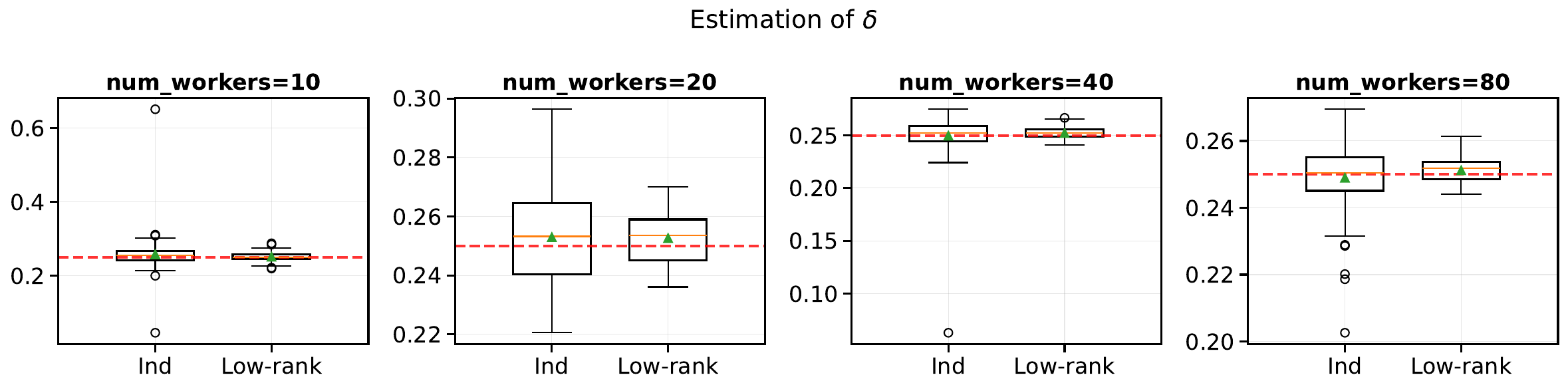}
            \includegraphics[width=0.43\textwidth]{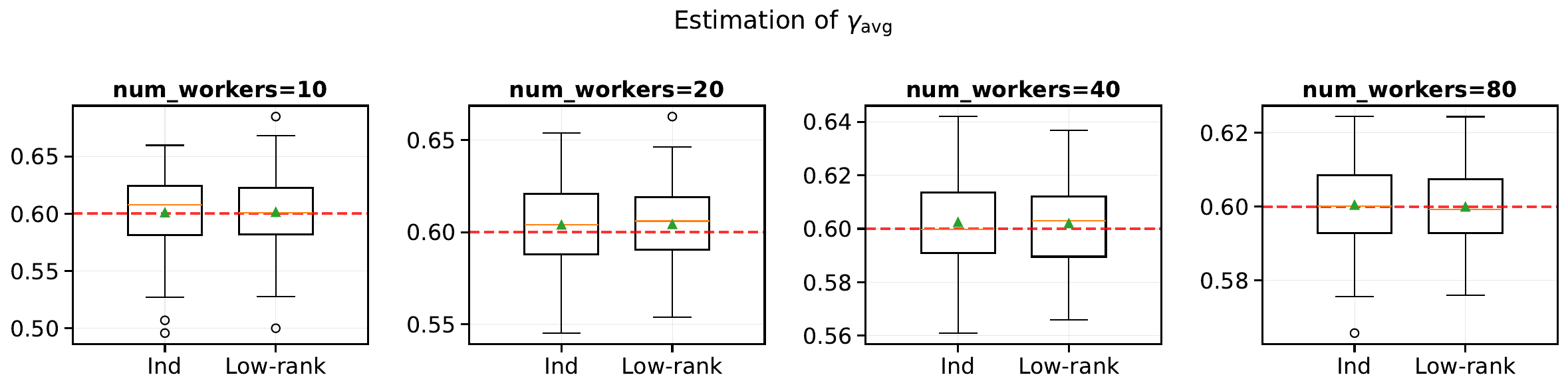}
            \caption{Boxplots of parameter estimates under varying number of workers: $\delta$ (left) and $\gamma_{\text{avg}}:=\sum_{i=1}^5 \gamma_i/5$ (right). }
            \label{fig:boxplot_machine_number_other}

          \end{figure}

          \begin{figure}[h]
            \centering
            \includegraphics[width=0.9\textwidth]{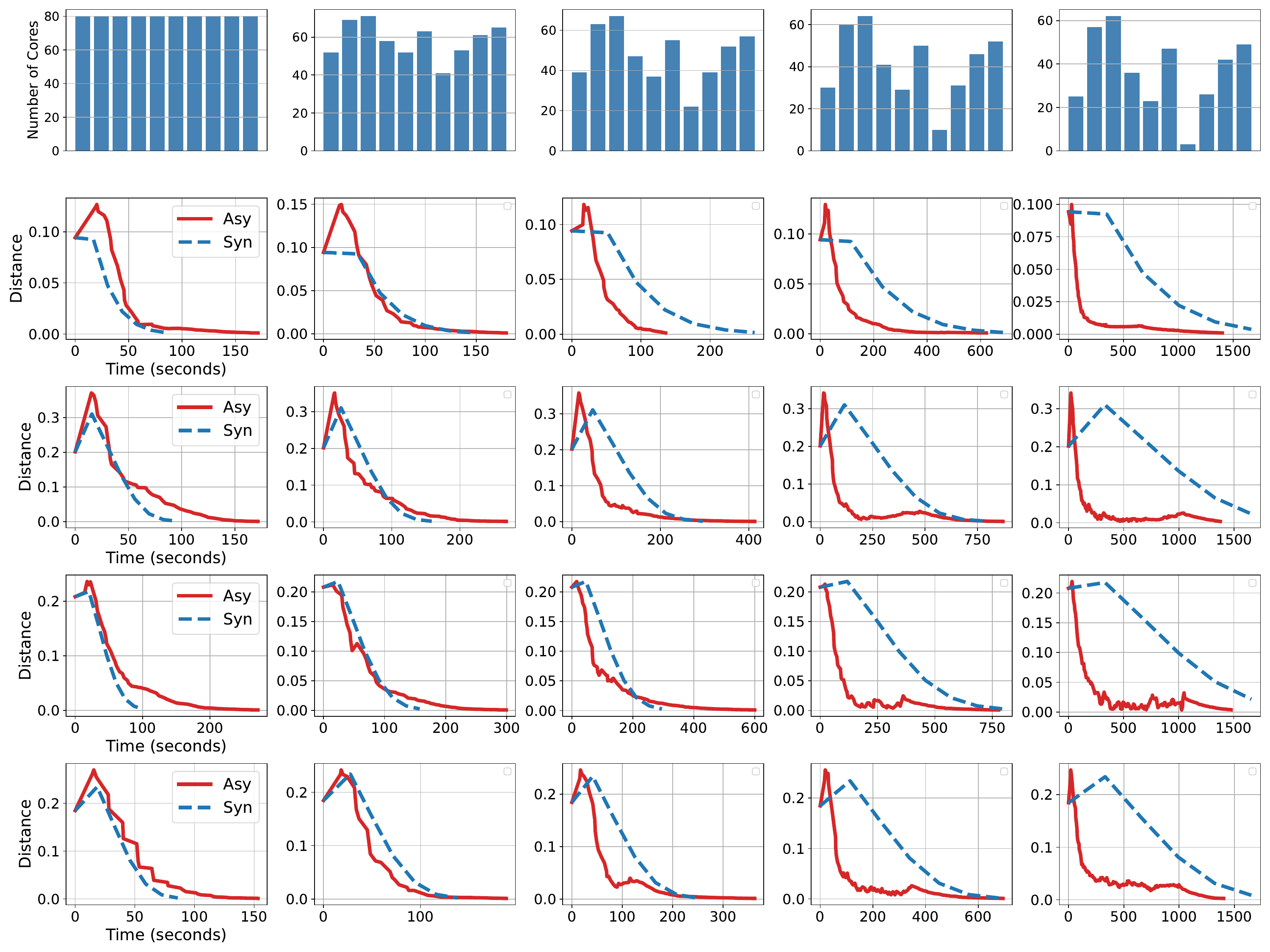}
            \caption{Comparison of asynchronous and synchronous algorithms under varying degrees of core allocation imbalance and different seeds.}
            \label{fig:heterogeneity_comparison_linear_other_replications}
          \end{figure}

    \subsection{Proofs}\label{sec:appendix:proofs}
          This subsection provides proofs of the theorems and propositions from the main text.

      \subsubsection{Proof of Proposition \ref{prop:kl_IDVSLOW}}\label{sec:proof_prop_kl_IDVSLOW}
          We need the following lemma to prove Proposition \ref{prop:kl_IDVSLOW}.
          \begin{lemma}\label{lem:det_sum}
            Suppose that \(\boldsymbol{X}_1, \ldots, \boldsymbol{X}_J\) are positive semidefinite matrices. Then,
            \[
            \det\left(\boldsymbol{I} + \sum_{j=1}^J \boldsymbol{X}_j \right) \leq \prod_{j=1}^J \det(\boldsymbol{I} + \boldsymbol{X}_j).
            \]
          \end{lemma}
          \begin{proof}
            It suffices to prove the case \( J = 2 \), i.e.,
            $\det(\boldsymbol{I} + \boldsymbol{X}_1 + \boldsymbol{X}_2) \leq \det(\boldsymbol{I} + \boldsymbol{X}_1) \det(\boldsymbol{I} + \boldsymbol{X}_2).$
            This inequality follows from the identity \vspace{-5pt}
            \[
            \det(\boldsymbol{I} + \boldsymbol{X}_1 + \boldsymbol{X}_2) = \det(\boldsymbol{I} + \boldsymbol{X}_1) \det\left( \boldsymbol{I} + (\boldsymbol{I} + \boldsymbol{X}_1)^{-\frac{1}{2}} \boldsymbol{X}_2 (\boldsymbol{I} + \boldsymbol{X}_1)^{-\frac{1}{2}} \right).\vspace{-5pt}
            \]
            and the fact that
            $
            0 \leq \lambda_i \left( (\boldsymbol{I} + \boldsymbol{X}_1)^{-\frac{1}{2}} \boldsymbol{X}_2 (\boldsymbol{I} + \boldsymbol{X}_1)^{-\frac{1}{2}} \right) \leq \lambda_i(\boldsymbol{X}_2),
            $
            where \(\lambda_i(\cdot)\) denotes the \(i\)-th largest eigenvalue, which follows from the Loewner order and properties of positive semidefinite matrices.
          \end{proof}

\vspace{-10pt}

          \begin{proof}[Proof of Proposition \ref{prop:kl_IDVSLOW}]\label{proof:kl_IDVSLOW}
            Denote the covariance matrices of the low-rank model with full local covariance and the independent model by \( \boldsymbol{C}_1 \) and \( \boldsymbol{C}_2 \), respectively. For each pair \( (i,j) \), define
           $ \boldsymbol{C}_{ij} = \boldsymbol{C}_{\boldsymbol{\theta}}(\mathcal{S}_i, \mathcal{S}_j), 
            \widetilde{\boldsymbol{C}}_{ij} = \boldsymbol{C}_{\boldsymbol{\theta}}(\mathcal{S}_i, \mathcal{S}^*) \boldsymbol{C}_{\boldsymbol{\theta}}^{-1}(\mathcal{S}^*, \mathcal{S}^*) \boldsymbol{C}_{\boldsymbol{\theta}}(\mathcal{S}^*, \mathcal{S}_j)$,
            where \( \mathcal{S}_i \) is the set of spatial locations in the \( i \)-th worker, and \( \mathcal{S}^* \) is the set of knots. Then the block structures of the covariance matrices are
            \[
            (\boldsymbol{C}_1)_{ij} =
            \begin{cases}
              \boldsymbol{C}_{ii}, & \text{if } i = j, \\
              \widetilde{\boldsymbol{C}}_{ij}, & \text{if } i \neq j,
            \end{cases}
            \quad
            (\boldsymbol{C}_2)_{ij} =
            \begin{cases}
              \boldsymbol{C}_{ii}, & \text{if } i = j, \\
              \boldsymbol{0}, & \text{if } i \neq j.
            \end{cases}
            \]
            The original covariance matrix \( \boldsymbol{C} \) has entries \( \boldsymbol{C}_{ij} = \boldsymbol{C}_{\boldsymbol{\theta}}(\mathcal{S}_i, \mathcal{S}_j) \).

            The KL divergences from the original \( P \) to the low-rank models \( P_1 \) and \( P_2 \) are
            \[
            \mathrm{KL}(P \,\|\, P_1) = \frac{1}{2} \left( \mathrm{tr}(\boldsymbol{C}_1^{-1} \boldsymbol{C}) - N + \log \frac{\det \boldsymbol{C}_1}{\det \boldsymbol{C}} \right),
            \quad
            \mathrm{KL}(P \,\|\, P_2) = \frac{1}{2} \left( \mathrm{tr}(\boldsymbol{C}_2^{-1} \boldsymbol{C}) - N + \log \frac{\det \boldsymbol{C}_2}{\det \boldsymbol{C}} \right).
            \]
            To show that \( \mathrm{KL}(P \,\|\, P_1) \leq \mathrm{K}(P \,\|\, P_2) + m \), it suffices to prove the following:
            \[
            \mathrm{tr}(\boldsymbol{C}_1^{-1} \boldsymbol{C}) \leq \mathrm{tr}(\boldsymbol{C}_2^{-1} \boldsymbol{C}) + m,
            \quad
            \det \boldsymbol{C}_1 \leq \det \boldsymbol{C}_2.
            \]
            In the following, we let \( \boldsymbol{R}_i = \boldsymbol{C}_{ii} - \widetilde{\boldsymbol{C}}_{ii} \), \( \boldsymbol{U}_i = \boldsymbol{C}_{\boldsymbol{\theta}}(\mathcal{S}_i, \mathcal{S}^*) \in \mathbb{R}^{n_i \times m} \), and \( \boldsymbol{C}^* = \boldsymbol{C}_{\boldsymbol{\theta}}(\mathcal{S}^*, \mathcal{S}^*) \). Define
            \[
            \boldsymbol{R} = \mathrm{diag}(\boldsymbol{R}_1, \ldots, \boldsymbol{R}_J) \in \mathbb{R}^{N \times N},
            \quad
            \boldsymbol{U} = \begin{bmatrix}
              \boldsymbol{U}_1^\top & \cdots & \boldsymbol{U}_J^\top
            \end{bmatrix}^\top \in \mathbb{R}^{N \times m}.
            \]
            Then, we can express the full local covariance matrix and the local covariance matrix of the \( i \)-th worker as
            \[
            \boldsymbol{C}_1 = \boldsymbol{R} + \boldsymbol{U} \boldsymbol{C}^{*-1} \boldsymbol{U}^\top,
            \quad
            \boldsymbol{C}_{ii} = \boldsymbol{R}_i + \boldsymbol{U}_i \boldsymbol{C}^{*-1} \boldsymbol{U}_i^\top.
            \]

            First, we show that the trace inequality holds.
            Note that \( \mathrm{tr}(\boldsymbol{C}_2^{-1} \boldsymbol{C}) = N = \mathrm{tr}(\boldsymbol{C}_1^{-1} \boldsymbol{C}_1) \), and hence:
            \[
            \mathrm{tr}(\boldsymbol{C}_1^{-1} \boldsymbol{C}) - \mathrm{tr}(\boldsymbol{C}_2^{-1} \boldsymbol{C})
            = \mathrm{tr}(\boldsymbol{C}_1^{-1} (\boldsymbol{C} - \boldsymbol{C}_1)).
            \]
            Using the Woodbury matrix identity:
            \[
            \vspace{-10pt}
            \boldsymbol{C}_1^{-1} = \boldsymbol{R}^{-1} - \boldsymbol{R}^{-1} \boldsymbol{U} \left( \boldsymbol{C}^* + \boldsymbol{U}^\top \boldsymbol{R}^{-1} \boldsymbol{U} \right)^{-1} \boldsymbol{U}^\top \boldsymbol{R}^{-1},
            \vspace{-10pt}
            \]
            we obtain:
            \begin{align*}
              \mathrm{tr}(\boldsymbol{C}_1^{-1} (\boldsymbol{C} - \boldsymbol{C}_1))
               & = \mathrm{tr} \left( \boldsymbol{U}^\top \boldsymbol{R}^{-1} (\boldsymbol{C}_1 - \boldsymbol{C}) \boldsymbol{R}^{-1} \boldsymbol{U} \left( \boldsymbol{C}^* + \boldsymbol{U}^\top \boldsymbol{R}^{-1} \boldsymbol{U} \right)^{-1} \right) \\
               & = \mathrm{tr} \left( \boldsymbol{U}^\top \boldsymbol{R}^{-1} (\boldsymbol{U} \boldsymbol{C}^{*-1} \boldsymbol{U}^\top - \boldsymbol{C}) \boldsymbol{R}^{-1} \boldsymbol{U} \left( \boldsymbol{C}^* + \boldsymbol{U}^\top \boldsymbol{R}^{-1} \boldsymbol{U} \right)^{-1} \right) \\
               & \quad + \mathrm{tr} \left( \boldsymbol{U}^\top \boldsymbol{R}^{-1} \boldsymbol{U} \left( \boldsymbol{C}^* + \boldsymbol{U}^\top \boldsymbol{R}^{-1} \boldsymbol{U} \right)^{-1} \right).
            \end{align*}
            Since \( \boldsymbol{C} \succeq \boldsymbol{U} \boldsymbol{C}^{*-1} \boldsymbol{U}^\top \), the first term is non-positive, and we have:
            \[
            \mathrm{tr}(\boldsymbol{C}_1^{-1} \boldsymbol{C}) - \mathrm{tr}(\boldsymbol{C}_2^{-1} \boldsymbol{C}) \leq \mathrm{tr} \left( \boldsymbol{U}^\top \boldsymbol{R}^{-1} \boldsymbol{U} \left( \boldsymbol{C}^* + \boldsymbol{U}^\top \boldsymbol{R}^{-1} \boldsymbol{U} \right)^{-1} \right).
            \]
            Let \( \boldsymbol{A} = \boldsymbol{C}^{*-1/2} \boldsymbol{U}^\top \boldsymbol{R}^{-1} \boldsymbol{U} \boldsymbol{C}^{*-1/2} \). Then the above becomes
            $
            \mathrm{tr}(\boldsymbol{A} (\boldsymbol{I} + \boldsymbol{A})^{-1}) \leq m,
            $, 
            as \(\boldsymbol{A} (\boldsymbol{I} + \boldsymbol{A})^{-1}\) has eigenvalues in \( [0,1) \).

            Second, we show the determinant inequality holds using the matrix determinant lemma:
            \[
            \vspace{-10pt}
            \det(\boldsymbol{C}_1) = \det(\boldsymbol{R}) \cdot \det \left( \boldsymbol{I} + \boldsymbol{C}^{*-1/2} \boldsymbol{U}^\top \boldsymbol{R}^{-1} \boldsymbol{U} \boldsymbol{C}^{*-1/2} \right),
            \]
            \[
            \det(\boldsymbol{C}_{ii}) = \det(\boldsymbol{R}_i) \cdot \det \left( \boldsymbol{I} + \boldsymbol{C}^{*-1/2} \boldsymbol{U}_i^\top \boldsymbol{R}_i^{-1} \boldsymbol{U}_i \boldsymbol{C}^{*-1/2} \right).
            \vspace{-10pt}
            \]
            Therefore,  \vspace{-10pt}
            \[
            \det(\boldsymbol{C}_2) = \prod_{i=1}^J \det(\boldsymbol{C}_{ii}) = \det(\boldsymbol{R}) \cdot \prod_{i=1}^J \det \left( \boldsymbol{I} + \boldsymbol{C}^{*-1/2} \boldsymbol{U}_i^\top \boldsymbol{R}_i^{-1} \boldsymbol{U}_i \boldsymbol{C}^{*-1/2} \right).
            \]
            \vspace{-10pt}
            By Lemma \ref{lem:det_sum}, we have
            \[
            \det \left( \boldsymbol{I} + \sum_{i=1}^J \boldsymbol{C}^{*-1/2} \boldsymbol{U}_i^\top \boldsymbol{R}_i^{-1} \boldsymbol{U}_i \boldsymbol{C}^{*-1/2} \right)
            \leq \prod_{i=1}^J \det \left( \boldsymbol{I} + \boldsymbol{C}^{*-1/2} \boldsymbol{U}_i^\top \boldsymbol{R}_i^{-1} \boldsymbol{U}_i \boldsymbol{C}^{*-1/2} \right),
            \]
            which implies \( \det(\boldsymbol{C}_1) \leq \det(\boldsymbol{C}_2) \).
          \end{proof}

      \subsubsection{Proof of results in Section \ref{sec:theory}}\label{sec:proof_results_sec_theory}

          Since each $f_j$ is $L_j$-smooth and $h$ is $L_h$-smooth, the aggregate
          $
          f := \frac{1}{J}\sum_{j=1}^J f_j + h
          $
          is $L$-smooth with
          $
          L \;\le\; \frac{1}{J}\sum_{j=1}^J L_j \;+\; L_h.
          $ Likewise, if each $\nabla^2 f_j$ is $H_j$-Lipschitz and $\nabla^2 h$ is $H_h$-Lipschitz, then $\nabla^2 f$ is $H$-Lipschitz with
          $
          H \;\le\; \frac{1}{J}\sum_{j=1}^J H_j \;+\; H_h.
          $

          For notational convenience, write $\nabla f := \partial f/\partial \boldsymbol{x}$ and
          $\nabla^2 f := \partial^2 f/\partial \boldsymbol{x}^2$, with the block gradients/Hessians
          \vspace{-10pt}
          \[
          \nabla_1 f := \frac{\partial f}{\partial \boldsymbol{x}_1},\quad
          \nabla_2 f := \frac{\partial f}{\partial \boldsymbol{x}_2},\quad
          \nabla_{11} f := \frac{\partial^2 f}{\partial \boldsymbol{x}_1^2},\quad
          \nabla_{22} f := \frac{\partial^2 f}{\partial \boldsymbol{x}_2^2},
          \vspace{-10pt}
          \]
          and analogously for each $f_j$ and for $h$. For any positive definite matrix $\boldsymbol{H}$ and vectors
          $\boldsymbol{a},\boldsymbol{b}$, denote the $\boldsymbol{H}$-inner product by
          $\langle \boldsymbol{a},\boldsymbol{b}\rangle_{\boldsymbol{H}} := \boldsymbol{a}^\top \boldsymbol{H}\boldsymbol{b}$.

          \begin{proof}[Proof of Proposition \ref{prop:Lyapunov_inequality}]
            Let
           $ \widehat{\boldsymbol{x}}_1^{\,t+1} = \arg\min_{\boldsymbol{x}_1} f(\boldsymbol{x}_1, \boldsymbol{x}_2^{\,t})$,
            and decompose the one-step decrease in function value as
            \begin{align*}
               & f(\boldsymbol{x}_1^t, \boldsymbol{x}_2^t) - f(\boldsymbol{x}_1^{\,t+1}, \boldsymbol{x}_2^{\,t+1}) \\
              = & \underbrace{f(\boldsymbol{x}_1^t, \boldsymbol{x}_2^t) - f(\widehat{\boldsymbol{x}}_1^{t+1}, \boldsymbol{x}_2^{t})}_{\text{First term}}
              + \underbrace{f(\widehat{\boldsymbol{x}}_1^{t+1}, \boldsymbol{x}_2^{t}) - f({\boldsymbol{x}}_1^{t+1}, \boldsymbol{x}_2^{t})}_{\text{Second term}}  + \underbrace{f(\boldsymbol{x}_1^{t+1}, \boldsymbol{x}_2^{t}) - f(\boldsymbol{x}_1^{t+1}, \boldsymbol{x}_2^{t+1})}_{\text{Third term}}.
            \end{align*}
            We bound each term separately.

            \textbf{First Term Bound.}
            By the definition of $\widehat{\boldsymbol{x}}_1^{t+1}$ and Lemma \ref{lem:gradient_gap_inequality_lower}, we have \vspace{-10pt}
            \begin{equation}\label{eq:first_term_bound}
              f(\boldsymbol{x}_1^t, \boldsymbol{x}_2^t) - f(\widehat{\boldsymbol{x}}_1^{t+1}, \boldsymbol{x}_2^{t}) \geq  \frac{1}{2\kappa L}\|\nabla_1 f(\boldsymbol{x}_1^{t}, \boldsymbol{x}_2^{t})\|^2.
              \vspace{-10pt}
            \end{equation}
            where we define the condition number $\kappa := \frac{L}{\mu}$.

            \textbf{Second Term Bound.} By Lemma \ref{lem:block_descent_inequality_upper} and the optimal condition: $\nabla_1f (\widehat{\boldsymbol{x}}_1^{t + 1}, \boldsymbol{x}_2^t)=\boldsymbol{0}$ ,\vspace{-10pt}
            \[ f (\widehat{\boldsymbol{x}}_1^{t + 1}, \boldsymbol{x}_2^t) - f
               (\boldsymbol{x}_1^{t + 1}, \boldsymbol{x}_2^t) \geqslant - \frac{L }{2} \|
               \widehat{\boldsymbol{x}}_1^{t + 1} - \boldsymbol{x}_1^{t + 1} \|^2 \vspace{-10pt}\]
            Since \vspace{-10pt}
            \begin{equation*}
                \begin{gathered}
                \hat{\boldsymbol{x}}_1^{t+1}=\arg \min _{\boldsymbol{x}_1} f\left(\boldsymbol{x}_1, \boldsymbol{x}_2^t\right), \\
                \boldsymbol{x}_1^{t+1}=\arg \min _{\boldsymbol{x}_1} J^{-1} \sum_{j=1}^J\left[f_j\left(\boldsymbol{x}_1, \boldsymbol{x}_2^{t-\tau_{t, j}^1}\right)\right]+h\left(\boldsymbol{x}_1, \boldsymbol{x}_2^t\right),
                \end{gathered}\vspace{-10pt}
            \end{equation*}
            Lemma \ref{lem:block_minimizer_sensitivity_2} yields \vspace{-10pt}
            \[ \| \widehat{\boldsymbol{x}}_1^{t + 1} - \boldsymbol{x}_1^{t + 1} \| \leqslant
               \mu^{- 1} \frac{1}{J} \sum_{j = 1}^J L_j \| \boldsymbol{x}_2^t -
               \boldsymbol{x}_2^{t - \tau_{t, j}^1} \|\vspace{-10pt} \]
            Thus, with  ${L^{\mathrm{msq}}} := \frac{1}{J} \sum_{j=1}^J L_j^2$,\vspace{-10pt}
            $$
            f\left(\hat{\boldsymbol{x}}_1^{t+1}, \boldsymbol{x}_2^{t+1}\right)-f\left(\boldsymbol{x}_1^{t+1}, \boldsymbol{x}_2^{t+1}\right) \geqslant-\frac{1}{2} \mu^{-2} L L^{\mathrm{msq}} \frac{1}{J} \sum_j\left\|\boldsymbol{x}_2^t-\boldsymbol{x}_2^{t-\tau_{t, j}^1}\right\|^2.\vspace{-10pt}
            $$

            To control $\frac{1}{J} \sum_j \| \boldsymbol{x}_2^t - \boldsymbol{x}_2^{t -
            \tau_{t, j}^1} \|^2$, note that
        \begin{equation}\label{eq:iteration_difference_second_term}
                \begin{aligned}
                & \frac{1}{J} \sum_j\left\|\boldsymbol{x}_2^t-\boldsymbol{x}_2^{t-\tau_{t, j}^1}\right\|^2 
                \stackrel{(a)}{\leqslant}  \frac{1}{J} \sum_j \tau_{t, j}^1 \sum_{s=t-\tau}^{t-1}\left\|\boldsymbol{x}_2^{s+1}-\boldsymbol{x}_2^s\right\|^2 \\
                \stackrel{(b)}{\leqslant} & \alpha^2 \overline{\tau} \tau \underline{\lambda}^{-1} \sum_{s=t-\tau}^{t-1}\left\|J^{-1} \sum_{j=1}^J \nabla_2 f_j\left(\boldsymbol{x}_1^{s+1-\tau_{s, j}^{21}}, \boldsymbol{x}_2^{s-\tau_{s, j}^{22}}\right)+\nabla_2 h\left(\boldsymbol{x}_1^{s+1}, \boldsymbol{x}_2^s\right)\right\|^2 \\
                \stackrel{(c)}{\leqslant} & 2 \alpha^2 \overline{\tau} \underline{\lambda}^{-1} \sum_{s=t-\tau}^{t-1}\left(J^{-1} \sum_{j=1}^J L_j\left(\sqrt{\left\|\boldsymbol{x}_1^{s+1-\tau_{s, j}^{21}}-\boldsymbol{x}_1^*\right\|^2+\left\|\boldsymbol{x}_2^{s+1-\tau_{s, j}^{22}}-\boldsymbol{x}_2^*\right\|^2}\right)\right)^2 \\
                & +2 \alpha^2 \overline{\tau} \underline{\lambda}^{-1} L_h \sum_{s=t-\tau}^{t-1}\left(\left\|\boldsymbol{x}_1^{s+1}-\boldsymbol{x}_1^*\right\|^2+\left\|\boldsymbol{x}_2^s-\boldsymbol{x}_2^*\right\|^2\right) \\
                \stackrel{(d)}{\leqslant} & 4 \alpha^2 \tau \overline{\tau} \underline{\lambda}^{-1}\left({L^{\mathrm{msq}}}+L_h^2\right) \max _{s \in[t-2 \tau, t]}\left\{\left\|\boldsymbol{x}_1^s-\boldsymbol{x}_1^*\right\|^2+\left\|\boldsymbol{x}_2^s-\boldsymbol{x}_2^*\right\|^2\right\} \\
                \stackrel{(e)}{\leqslant} & 4 \alpha^2 \overline{\tau} \tau \mu^{-1}\underline{\lambda}^{-1}\left({L^{\mathrm{msq}}}+L_h^2\right) \max _{s \in[t-2 \tau, t]}\left[f\left(\boldsymbol{x}_1^s, \boldsymbol{x}_2^s\right)-\widehat{f}\right]
                \end{aligned}
            \end{equation}
            where $\widehat{f}$  is the minimal value and  $ \frac{1}{J} \sum_{j=1}^J \tau_{t, j}^1\le \overline{\tau}$ for any $t$.
            In the above Equation \eqref{eq:iteration_difference_second_term}, the inequalities $(a)$ follows from the Cauchy--Schwarz inequality, $(b)$ follows from that $mod(H) \succeq \underline{\lambda}^{-1}\boldsymbol{I}$, $(c)$ follows from the Lipschitz continuity of $f_j$ and $h$, $(d)$ follows from the Cauchy--Schwarz inequality and $(e)$ follows the strong convexity of $f$.

            Defining $B:=2 \underline{\lambda}^{-1} \mu^{-3} L L^{\mathrm{msq}}\left({L^{\mathrm{msq}}}+L_h^2\right)$, we obtain
            \begin{equation}\label{eq:first_term_bound_final}
            f\left(\hat{\boldsymbol{x}}_1^{t+1}, \boldsymbol{x}_2^{t+1}\right)-f\left(\boldsymbol{x}_1^{t+1}, \boldsymbol{x}_2^{t+1}\right) \geqslant-\alpha^2 \tau \overline{\tau}^1 B \max _{s \in[t-2 \tau, t]}\left[f\left(\boldsymbol{x}_1^s, \boldsymbol{x}_2^s\right)-\widehat{f}\right]
            \end{equation}

            \textbf{Third Term Bound.}
            For notational convenience, we define
            \[
            \boldsymbol{H}_2^t:=\boldsymbol{H}_2^{\{t, 0, 0\}}, \quad \boldsymbol{G}_2^t := \boldsymbol{G}_2^{\{t, 0, 0\}}, \quad \boldsymbol{a} := [\boldsymbol{H}_2^t]^{-1} \boldsymbol{G}_2^t,
            \quad
            \boldsymbol{b} := (\boldsymbol{H}_2^{\{t, \tau_{t,j}^{21}, \tau_{t,j}^{22}\}})^{-1}
            \boldsymbol{G}_2^{\{t, \tau_{t,j}^{21}, \tau_{t,j}^{22}\}}.
            \]
            The third term is lower bounded according to either Lemma \ref{lem:block_descent_inequality_upper}  or Lemma \ref{lem:block_descent_inequality_Hessian}.
            
            \textbf{Case 1:} When the inequality is from Lemma \ref{lem:block_descent_inequality_upper}, we have \vspace{-10pt}
            \begin{equation}\label{eq:third_term_bound_1}
                f\left(\boldsymbol{x}_1^t, \boldsymbol{x}_2^t\right)-f\left(\boldsymbol{x}_1^t, \boldsymbol{x}_2^{t+1}\right) \geqslant \alpha\left\langle\boldsymbol{G}_2^t, \boldsymbol{b}\right\rangle-\alpha^2 \frac{L}{2}\|\boldsymbol{b}\|^2 \vspace{-5pt}
            \end{equation}
Let \vspace{-5pt}
$$\tilde{\boldsymbol{H}}_2:=\boldsymbol{H}_2^{\left\{t, \tau_{t, j}^{21}, \tau_{t, j}^{22}\right\}} \quad \tilde{\boldsymbol{G}}_2:=\boldsymbol{G}_2^{\left\{t, \tau_{t, j}^{21} \tau_{t, j}^{22}\right\}}.\vspace{-5pt} $$ 
then we have  $$\left\langle\boldsymbol{G}_2^t, \boldsymbol{b}\right\rangle=\left\langle\boldsymbol{G}_2^t, \tilde{\boldsymbol{G}}_2\right\rangle_{\left(\tilde{\boldsymbol{H}}_2\right)^{-1 }} \quad \|\boldsymbol{b}\|^2 \leqslant \underline{\lambda}^{-1}\left\|\tilde{\boldsymbol{G}}_2\right\|_{\left(\tilde{\boldsymbol{H}}_2\right)^{-1 }},$$ we have, if $\alpha \leqslant L^{-1} \underline{\lambda}$,
where the later inequality is due to $\underline{\lambda} \boldsymbol{I} \preceq \tilde{\boldsymbol{H}}_2$. Thus,
\begin{equation}\label{eq:third_term_bound_2}
    \alpha\left\langle\boldsymbol{G}_2^t, \boldsymbol{b}\right\rangle-\alpha^2 \frac{L}{2}\|\boldsymbol{b}\|^2 \geqslant \frac{\alpha}{2}\left(\left\|\boldsymbol{G}_2^t\right\|_{\tilde{\boldsymbol{H}}_2^{-1}}^2-\left\|\boldsymbol{G}_2^t-\tilde{\boldsymbol{G}}_2\right\|_{\tilde{\boldsymbol{H}}_2^{-1}}^2\right) \geqslant \frac{\alpha}{2} \overline{\lambda}^{-1}\left\|\boldsymbol{G}_2^t\right\|^2-\frac{\alpha}{2}(\underline{\lambda})^{-1}\left\|\boldsymbol{G}_2^t-\tilde{\boldsymbol{G}}_2\right\|^2 .
\end{equation}
where the later inequality is due to $\underline{\lambda} \boldsymbol{I} \preceq \tilde{\boldsymbol{H}}_2 \preceq \overline{\lambda} \boldsymbol{I}$.

By the smoothness of each $f_j$, we have \vspace{-10pt}
\begin{equation}\label{eq:gradient_difference_bound}
\left\|\boldsymbol{G}_2^t-\tilde{\boldsymbol{G}}_2\right\|^2 \leqslant L^{\mathrm{msq}}\left(J^{-1} \sum_{j=1}^J\left\|\boldsymbol{x}_1^{t+1}-\boldsymbol{x}_1^{t+1-\tau_{t, j}^{21}}\right\|^2+J^{-1} \sum_{j=1}^J\left\|\boldsymbol{x}_2^t-\boldsymbol{x}_2^{t-\tau_{t, j}^{22}}\right\|^2\right) \vspace{-10pt}
\end{equation}
and it is sufficient to bound the two terms on the right-hand side.

According to the proof of Equation \eqref{eq:iteration_difference_second_term}, the second term in the right hand side of Equation \eqref{eq:gradient_difference_bound} is bounded by \vspace{-5pt}
\begin{equation}\label{eq:gradient_difference_bound_x2}
\frac{1}{J} \sum_j\left\|\boldsymbol{x}_2^t-\boldsymbol{x}_2^{t-\tau_{t, j}^1}\right\|^2 \leqslant 4 \alpha^2 \tau \overline{\tau}^{-1} \mu^{-1} \underline{\lambda}^{-1}\left(L^{\mathrm{msq}}+L_h^2\right) \max _{s \in[t-2 \tau, t]}\left[f\left(\boldsymbol{x}_1^s, \boldsymbol{x}_2^s\right)-\widehat{f}\right] \vspace{-5pt}
\end{equation}

As for the first term in the right-hand side of Equation \eqref{eq:gradient_difference_bound}. Let $t_j:=t-\tau_{t, j}^{21}$, recall that \vspace{-5pt}
$$
\begin{gathered}
\boldsymbol{x}_1^{t+1}=\arg \min _{\boldsymbol{x}_1} J^{-1} \sum_{l=1}^J f_j\left(\boldsymbol{x}_1, \boldsymbol{x}_2^{t-\tau_{t, l}^1}\right)+h\left(\boldsymbol{x}_1, \boldsymbol{x}_2^{t}\right), \\
\boldsymbol{x}_1^{t+1-\tau_{t, j}^{21}}=\arg \min _{\boldsymbol{x}_1} J^{-1} \sum_{l=1}^J f_l\left(\boldsymbol{x}_1, \boldsymbol{x}_2^{t_j-\tau_{t_j, l}^1}\right)+h\left(\boldsymbol{x}_1, \boldsymbol{x}_2^{t_j}\right),
\end{gathered} \vspace{-5pt}
$$
Lemma \ref{lem:block_minimizer_sensitivity_2} and Young's inequality yield
\begin{equation}\label{eq:iteration_difference_first_term_bound1}
    \begin{aligned}
        \left\|\boldsymbol{x}_1^{t+1}-\boldsymbol{x}_1^{t+1-\tau_{t, i}^{21}}\right\|^2 \leqslant & 3 \mu^{-2}\left(\frac{1}{J} \sum_{l=1}^J L_l\left\|\boldsymbol{x}_2^{t_j}-\boldsymbol{x}_2^{t_j-\tau_{t_j, l}^1}\right\|\right)^2 \\
        & +3 \mu^{-2}\left(\frac{1}{J} \sum_{l=1}^J L_l\left\|\boldsymbol{x}_2^{t}-\boldsymbol{x}_2^{t-\tau_{t, l}^1}\right\|\right)^2+3 \mu^{-2}\left(L_h+\overline{L}\right)^2\left\|\boldsymbol{x}_2^{t}-\boldsymbol{x}_2^{t_j}\right\|^2
        \end{aligned}
\end{equation}
For the three terms in the right-hand side of Equation \eqref{eq:iteration_difference_first_term_bound1}, by the triangle and Cauchy-Swarchtz inequality (Recall that $L^{\mathrm{msq}}=\frac{1}{J} \sum_{l=1}^J L_l^2$.), we have
$$
\begin{gathered}
    \left(\frac{1}{J} \sum_{l=1}^J L_l\left\|\boldsymbol{x}_2^{t_j}-\boldsymbol{x}_2^{t_j-\tau_{t_j-1, l}^1}\right\|\right)^2 \leqslant L^{\mathrm{msq}} \overline{\tau} \sum_{s=t-2 \tau}^{t-1}\left\|\boldsymbol{x}_2^{s+1}-\boldsymbol{x}_2^s\right\|^2 \\
    \left(\frac{1}{J} \sum_{l=1}^J L_l\left\|\boldsymbol{x}_2^{t}-\boldsymbol{x}_2^{t-\tau_{t, l}^1}\right\|\right) \leqslant L^{\mathrm{msq}} \overline{\tau} \sum_{s=t-2 \tau}^{t-1}\left\|\boldsymbol{x}_2^{s+1}-\boldsymbol{x}_2^s\right\|^2 \\
    \left\|\boldsymbol{x}_2^{t}-\boldsymbol{x}_2^{t_j}\right\|^2 \leqslant \tau_{t, j}^{21} \sum_{s=t-\tau}^{t-1}\left\|\boldsymbol{x}_2^{s+1}-\boldsymbol{x}_2^s\right\|^2
    \end{gathered}
$$
Thus, \vspace{-5pt}
$$
\frac{1}{J} \sum_j\left\|\boldsymbol{x}_1^{t+1}-\boldsymbol{x}_1^{t+1-\tau_{t, j}^{21}}\right\|^2 \leqslant 3 \overline{\tau} \mu^{-2}\left[2 L^{\mathrm{msq}}+\left(L_h+\overline{L}\right)^2\right] \sum_{s=t-2 \tau}^{t-1}\left\|\boldsymbol{x}_2^{s+1}-\boldsymbol{x}_2^s\right\|^2 \vspace{-5pt}
$$
According to the proof of Equation \eqref{eq:iteration_difference_second_term},
$$
\sum_{s=t-2 \tau}^{t-1}\left\|\boldsymbol{x}_2^{s+1}-\boldsymbol{x}_2^s\right\|^2 \leqslant 8 \eta^2 \tau \mu^{-1} \underline{\lambda}^{-1}\left(L^{\mathrm{msq}}+L_h^2\right) \max_{s \in[t-3 \tau, t]}\left[f\left(\boldsymbol{x}_1^s, \boldsymbol{x}_2^s\right)-\hat{f}\right]
$$
Thus, \vspace{-5pt}
\begin{equation}\label{eq:gradient_difference_bound_x1}
    \begin{aligned}
        & \frac{1}{J} \sum_j\left\|\boldsymbol{x}_1^t-\boldsymbol{x}_1^{t-\tau_{t, j}^{21}}\right\|^2 \\
        & \leqslant 24 \alpha^2 \tau \overline{\tau} \mu^{-3} \underline{\lambda}^{-1}\left(L^{\mathrm{msq}}+L_h^2\right)\left[2 L^{\mathrm{msq}}+\left(L_h+\overline{L}\right)^2\right]\max_{s \in[t-3 \tau, t]}\left[f\left(\boldsymbol{x}_1^s, \boldsymbol{x}_2^s\right)-\hat{f}\right]
        \end{aligned} \vspace{-5pt}
\end{equation}
Therefore, by combining Equations \eqref{eq:gradient_difference_bound}, \eqref{eq:gradient_difference_bound_x2} and   \eqref{eq:gradient_difference_bound_x1}, 
\begin{equation}\label{eq:gradient_difference_bound_final}
    \begin{aligned}
        & \left\|\boldsymbol{G}_2^t-\tilde{\boldsymbol{G}}_2\right\|^2 \\
        & \leqslant 4 \alpha^2 \tau \overline{\tau} \mu^{-1} \underline{\lambda}^{-1} L^{\mathrm{msq}}\left(L^{\mathrm{msq}}+L_h^2\right)\left[1+6 \mu^{-2}\left[2 L^{\mathrm{msq}}+\left(L_h+\overline{L}\right)^2\right]\right] \max _{s \in[t-3 \tau, t]}\left[f\left(\boldsymbol{x}_1^s, \boldsymbol{x}_2^s\right)-\hat{f}\right]
        \end{aligned}
\end{equation}

We then bound $\left\|\boldsymbol{G}_2^t\right\|^2$ as follows. By the Cauchy-Schwarz inequality and Young's inequality,
$$
\left\|\nabla_2 f\left(\boldsymbol{x}_1^t, \boldsymbol{x}_2^t\right)\right\|^2 \leqslant 3\left\|\nabla_2 f\left(\boldsymbol{x}_1^t, \boldsymbol{x}_2^t\right)-\nabla_2 f\left(\hat{\boldsymbol{x}}_1^{t+1}, \boldsymbol{x}_2^t\right)\right\|^2+3\left\|\nabla_2 f\left(\hat{\boldsymbol{x}}_1^{t+1}, \boldsymbol{x}_2^t\right)-\boldsymbol{G}_2^t\right\|^2+3\left\|\boldsymbol{G}_2^t\right\|^2
$$
The smoothness and strong convexity of $f$ yields
$$
\begin{gathered}
\left\|\nabla_2 f\left(\boldsymbol{x}_1^t, \boldsymbol{x}_2^t\right)-\nabla_2 f\left(\hat{\boldsymbol{x}}_1^{t+1}, \boldsymbol{x}_2^t\right)\right\|^2 \leqslant L^2\left\|\boldsymbol{x}_1^t-\hat{\boldsymbol{x}}_1^{t+1}\right\|^2 \leqslant \kappa^2\left\|\nabla_1 f\left(\boldsymbol{x}_1^t, \boldsymbol{x}_2^t\right)\right\| \\
\left\|\nabla_2 f\left(\hat{\boldsymbol{x}}_1^{t+1}, \boldsymbol{x}_2^t\right)-\boldsymbol{G}_2^t\right\|^2 \leqslant L\left\|\hat{\boldsymbol{x}}_1^{t+1}-\boldsymbol{x}_1^{t+1}\right\|^2
\end{gathered}
$$
Recall that we have proved that (refer to Equation \eqref{eq:first_term_bound_final})
$$
\frac{L}{2}\left\|\hat{\boldsymbol{x}}_1^{t+1}-\boldsymbol{x}_1^{t+1}\right\|^2 \leqslant \alpha^2 \tau \overline{\tau}^1 B \max _{s \in[t-2 \tau, t]}\left[f\left(\boldsymbol{x}_1^s, \boldsymbol{x}_2^s\right)-\widehat{f}\right]
$$
Therefore,
$$
\left\|\boldsymbol{G}_2^t\right\|^2 \geqslant \frac{1}{3}\left\|\nabla_2 f\left(\boldsymbol{x}_1^t, \boldsymbol{x}_2^t\right)\right\|^2-3 \kappa^2\left\|\nabla_1 f\left(\boldsymbol{x}_1^t, \boldsymbol{x}_2^t\right)\right\|-2 \alpha^2 \tau \overline{\tau}^1 B \max _{s \in[t-2 \tau, t]}\left[f\left(\boldsymbol{x}_1^s, \boldsymbol{x}_2^s\right)-\widehat{f}\right]
$$
Applying Lemma \ref{lem:linear_combination_lower_bound} yields
\begin{equation}\label{eq:combining_bound}
    \begin{aligned}
        & \frac{1}{2} \alpha \overline{\lambda}^{-1}\left\|\boldsymbol{G}_2^t\right\|^2+\frac{1}{2 \kappa L}\left\|\nabla_1 f\left(\boldsymbol{x}_1^t, \boldsymbol{x}_2^t\right)\right\|^2 \geqslant A(\alpha) \alpha\left\|\nabla f\left(\boldsymbol{x}_1^t, \boldsymbol{x}_2^t\right)\right\|^2-\alpha^3 \tau \overline{\tau}^1 B \overline{\lambda}^{-1} \max _{s \in[t-2 \tau, t]}\left[f\left(\boldsymbol{x}_1^s, \boldsymbol{x}_2^s\right)-\right. \\
        & \left.\widehat{f}\right]
        \end{aligned}        
\end{equation}
where $A(\alpha)$ defined as
$$
A(\alpha):= \begin{cases}{\left[\kappa L\left(18 \kappa^2+1\right)\right]^{-1} \alpha^{-1},} & \text { if } \alpha \geqslant \overline{\lambda}\left(6 L \kappa^3\right)^{-1} \\ \frac{1}{6} \overline{\lambda}^{-1}, & \text { otherwise }\end{cases}
$$
which is bounded below by $A:=\min \left\{\left[\kappa L\left(18 \kappa^2+1\right)\right]^{-1}, \frac{1}{6} \overline{\lambda}^{-1}\right\}
$

Moreover, Lemma \ref{lem:gradient_gap_inequality_upper} gives
\begin{equation}\label{eq:gradient_gap_inequality_upper_1}
    \left\|\nabla f\left(\boldsymbol{x}_1^t, \boldsymbol{x}_2^t\right)\right\|^2 \geqslant 2 \mu\left(f\left(\boldsymbol{x}_1^t, \boldsymbol{x}_2^t\right)-\hat{f}\right) .
\end{equation}

\textbf{Summarizing} all the related bounds, i.e., Equations \eqref{eq:first_term_bound}, \eqref{eq:first_term_bound_final}, \eqref{eq:third_term_bound_1}, \eqref{eq:third_term_bound_2}, \eqref{eq:gradient_difference_bound_final}, \eqref{eq:combining_bound}, \eqref{eq:gradient_gap_inequality_upper_1}, gives: if $\alpha \leqslant L^{-1} \underline{\lambda}$, then
$$
\begin{aligned}
& f\left(\boldsymbol{x}_1^t, \boldsymbol{x}_2^t\right)-f\left(\boldsymbol{x}_1^{t+1}, \boldsymbol{x}_2^{t+1}\right) \\
\geqslant & \alpha C_0\left(f\left(\boldsymbol{x}_1^t, \boldsymbol{x}_2^t\right)-\hat{f}\right)-\alpha^3 \tau \overline{\tau} C_1 \max _{s \in[t-3 \tau, t]}\left[f\left(\boldsymbol{x}_1^s, \boldsymbol{x}_2^s\right)-\hat{f}\right]-\alpha^2 \tau \overline{\tau} C_2 \max _{s \in[t-2 \tau, t]}\left[f\left(\boldsymbol{x}_1^s, \boldsymbol{x}_2^s\right)-\hat{f}\right]
\end{aligned}
$$
where $C_0, C_1, C_2$ are constants defined as follows:
\begin{equation}\label{eq:constants1}
    \begin{gathered}
        C_0:=\mu \min \left\{2\left[\kappa L\left(18 \kappa^2+1\right)\right]^{-1}, \frac{1}{3} \overline{\lambda}^{-1}\right\} \\
        C_1:=2 \mu^{-1} \underline{\lambda}^{-2} L^{\mathrm{msq}}\left(L^{\mathrm{msq}}+L_h^2\right)\left[1+6 \mu^{-2}\left[2 L^{\mathrm{msq}}+\left(L_h+\overline{L}\right)^2\right]\right] \\
        C_2:=2(\overline{\lambda}^{-1}+1) \underline{\lambda}^{-1} \mu^{-2} L L^{\mathrm{msq}}\left(L^{\mathrm{msq}}+L_h^2\right)
        \end{gathered} 
\end{equation}
Namely, with $V_t:=f\left(\boldsymbol{x}_1^t, \boldsymbol{x}_2^t\right)-\hat{f}$,
\begin{equation}\label{eq:Liypnov1}
V_{t+1} \leqslant\left(1-\alpha C_0\right) V_t+\alpha^3 \tau \overline{\tau} C_1 \max _{s \in[t-3 \tau, t]} V_s+\alpha^2 \tau \overline{\tau} C_2 \max _{s \in[t-2 \tau, t]} V_s .
\end{equation}

            \textbf{Case 2:} If the inequality is from Lemma \ref{lem:block_descent_inequality_Hessian}, then we have
            \begin{equation}\label{eq:third_term_case_2_bound_1}
                \begin{aligned}
                  f(\boldsymbol{x}_1^{t+1}, \boldsymbol{x}_2^t) - f(\boldsymbol{x}_1^{t+1}, \boldsymbol{x}_2^{\,t+1})
                   & \overset{(a)}{\geq} \alpha \langle \boldsymbol{a}, \boldsymbol{b} \rangle_{\boldsymbol{H}_2^t}
                  - \frac{1}{2} \alpha^2 \|\boldsymbol{b}\|^2_{\boldsymbol{H}_2^t}
                  - \frac{1}{6} \alpha^3 H^3 \mu^{-3/2} \|\boldsymbol{b}\|^3_{\boldsymbol{H}_2^t} \\
                   & \overset{(b)}{\geq} \frac{\alpha}{2} \Big( \|\boldsymbol{a}\|^2_{\boldsymbol{H}_2^t}
                  - \|\boldsymbol{b}-\boldsymbol{a}\|^2_{\boldsymbol{H}_2^t} \Big)
                  - \frac{2}{3} \alpha^3 H^3 \mu^{-3/2} \Big( \|\boldsymbol{a}\|^3_{\boldsymbol{H}_2^t}
                  + \|\boldsymbol{b}-\boldsymbol{a}\|^3_{\boldsymbol{H}_2^t} \Big).
                \end{aligned}
              \end{equation}
            In Equation \eqref{eq:third_term_case_2_bound_1}, inequality $(a)$ follows from Equation \eqref{eq:block_descent_inequality_Hessian2} in Lemma \ref{lem:block_descent_inequality_Hessian}, inequality $(b)$ follows from Lemmas \ref{lem:quadratic_surrogate} and \ref{lem:cubic_norm}.

            Thus, whenever (this condition will be satisfied eventually due to the convergence), for any given $\epsilon<1$,
$$
\frac{2}{3} \eta^2 H^3 \sqrt{\mu}^{-3}\|\boldsymbol{a}\|_{\boldsymbol{H}_2^t} \leqslant \frac{\epsilon}{2}, \frac{2}{3} \eta^2 H^3 \sqrt{\mu}^{-3}\|\boldsymbol{b}-\boldsymbol{a}\|_{\boldsymbol{H}_2^t}^2 \leqslant \frac{\epsilon}{2},
$$
we have
$$
f\left(\boldsymbol{x}_1^t, \boldsymbol{x}_2^t\right)-f\left(\boldsymbol{x}_1^t, \boldsymbol{x}_2^{t+1}\right) \geqslant \frac{\eta}{2}(1-\epsilon)\left(\|\boldsymbol{a}\|_{\boldsymbol{H}_2^t}^2-\|\boldsymbol{b}-\boldsymbol{a}\|_{\boldsymbol{H}_2^t}^2\right) .
$$

The term $\|\boldsymbol{a}\|_{\boldsymbol{H}_2^t}^2$ is lower bounded by
$$
\|\boldsymbol{a}\|_{\boldsymbol{H}_2^t}^2=\left(\boldsymbol{G}_2^t\right)^{\top}\left[\boldsymbol{H}_2^t\right]^{-1} \boldsymbol{G}_2^t \geqslant L^{-1}\left\|\boldsymbol{G}_2^t\right\|^2
$$

The term $\|\boldsymbol{b}-\boldsymbol{a}\|_{\boldsymbol{H}_2^t}^2$ is caused by delays; we relate it to the delays. Since $\boldsymbol{b}-\boldsymbol{a}$ can be decomposed as
$$
\boldsymbol{b}-\boldsymbol{a}=\left[\tilde{\boldsymbol{H}}_2^{-1}-\left[\boldsymbol{H}_2^t\right]^{-1}\right]\left[\tilde{\boldsymbol{G}}_2-\boldsymbol{G}_2^t\right]+\left[\boldsymbol{H}_2^t\right]^{-1}\left[\tilde{\boldsymbol{G}}_2-\boldsymbol{G}_2^t\right]+\left[\tilde{\boldsymbol{H}}_2^{-1}-\left[\boldsymbol{H}_2^t\right]^{-1}\right] \boldsymbol{G}_2^t
$$
and $\left[\tilde{\boldsymbol{H}}_2^{-1}-\left[\boldsymbol{H}_2^t\right]^{-1}\right]=\tilde{\boldsymbol{H}}_2^{-1}\left(\tilde{\boldsymbol{H}}_2-\boldsymbol{H}_2^t\right)\left[\boldsymbol{H}_2^t\right]^{-1}$, we have for any given $0<\epsilon \leqslant 1$,
$$
\begin{aligned}
    \|\boldsymbol{b}-\boldsymbol{a}\|_{\boldsymbol{H}_2^t}^2 \leqslant & \left(2+3 \epsilon^{-1}\right) \underline{\lambda}^{-2} \mu^{-1}\left\|\tilde{\boldsymbol{H}}_2-\boldsymbol{H}_2^t\right\|^2\left\|\tilde{\boldsymbol{G}}_2-\boldsymbol{G}_2^t\right\|^2 \\
    & \left(1+\frac{\epsilon}{3}\right) \mu^{-1}\left\|\tilde{\boldsymbol{G}}_2-\boldsymbol{G}_2^t\right\|^2+\left(2+3 \epsilon^{-1}\right) \underline{\lambda}^{-2} \mu^{-3}\left\|\tilde{\boldsymbol{H}}_2-\boldsymbol{H}_2^t\right\|^2\left\|\boldsymbol{G}_2^t\right\|^2
    \end{aligned}
$$
Thus, whenever (this condition will be satisfied eventually due to the convergence)
$$
\left(2+3 \epsilon^{-1}\right) \underline{\lambda}^{-2} \mu^{-1}\left\|\tilde{\boldsymbol{H}}_2-\boldsymbol{H}_2^t\right\|^2 \leqslant \frac{\epsilon}{3} \text { and }\left(2+3 \epsilon^{-1}\right) \underline{\lambda}^{-2} \mu^{-2}\left\|\boldsymbol{G}_2^t\right\|^2 \leqslant \frac{\epsilon}{3}\left(H^{\mathrm{msq}}\right)^{-1} L^{\mathrm{msq}},
$$
we have
$$
\|\boldsymbol{b}-\boldsymbol{a}\|_{\boldsymbol{H}_2^t}^2 \leqslant\left(1+\frac{2 \epsilon}{3}\right) \mu^{-1}\left\|\tilde{\boldsymbol{G}}_2-\boldsymbol{G}_2^t\right\|^2+\mu^{-1}\left(H^{\mathrm{msq}}\right)^{-1} L^{\mathrm{msq}}\left\|\tilde{\boldsymbol{H}}_2-\boldsymbol{H}_2^t\right\|^2 .
$$
The bound $\left\|\tilde{\boldsymbol{G}}_2-\boldsymbol{G}_2^t\right\|^2$ are already given in Equation. The bound of $\left\|\tilde{\boldsymbol{H}}_2-\boldsymbol{H}_2^t\right\|^2$ can be established similarly, which is
$$
\begin{aligned}
& \left\|\tilde{\boldsymbol{H}}_2-\boldsymbol{H}_2^t\right\|^2 \\
& \leqslant 4 \alpha^2 \tau \overline{\tau} \mu^{-1} \underline{\lambda}^{-1} H^{\mathrm{msq}}\left(L^{\mathrm{msq}}+L_h^2\right)\left[1+6 \mu^{-2}\left[2 L^{\mathrm{msq}}+\left(L_h+\overline{L}\right)^2\right]\right] \max_{s \in[t-3 \tau-1, t]} \left[f\left(\boldsymbol{x}_1^s, \boldsymbol{x}_2^s\right)-\hat{f}\right]
\end{aligned}
$$
Thus,
$$
\begin{aligned}
& \|\boldsymbol{b}-\boldsymbol{a}\|_{\boldsymbol{H}_2^t}^2 \leqslant(1+\epsilon) 4 \alpha^2 \tau \bar{\tau} \mu^{-1} \underline{\lambda}^{-1} L^{\mathrm{msq}}\left(L^{\mathrm{msq}}+L_h^2\right)\left[1+6 \mu^{-2}\left[2 L^{\mathrm{msq}}+\left(L_h+\bar{L}\right)^2\right]\right] \max _{s \in[t-3 \tau, t]}\left[f\left(\boldsymbol{x}_1^s, \boldsymbol{x}_2^s\right)-\right. \\
& \hat{f}]
\end{aligned}
$$
Similarly, we have
$$
\begin{gathered}
\frac{1}{2} \alpha L^{-1}\left|\boldsymbol{G}_2^t\right|^2+\frac{1}{2 \kappa L}\left|\nabla_1 f\left(\boldsymbol{x}_1^t, \boldsymbol{x}_2^t\right)\right|^2 \geqslant A \alpha\left|\nabla f\left(\boldsymbol{x}_1^t, \boldsymbol{x}_2^t\right)\right|^2-\alpha^3 \tau \bar{\tau}^1 B L^{-1} \max _{s \in[t-2 \tau, t]}\left[f\left(\boldsymbol{x}_1^s, \boldsymbol{x}_2^s\right)-f^*\right] \\
A:=\min \left\{\left[\kappa L\left(18 \kappa^2+1\right)\right]^{-1}, \frac{1}{6} L^{-1}\right\}
\end{gathered}
$$

In this case, similarly,  summarizing all the related bounds gives: if $\alpha \leqslant 1$, then
$$
V_{t+1} \leqslant\left(1-\alpha C_0^{\prime}\right) V_t+\alpha^3 \tau \bar{\tau} C_1^{\prime} \max _{s \in[t-3 \tau, t]} V_s+\alpha^2 \tau \bar{\tau} C_2^{\prime} \max _{s \in[t-2 \tau, t]} V_s
$$
where $C_0^{\prime}, C_1^{\prime}, C_2^{\prime}$ are constants defined as follows:
$$
\begin{gathered}
C_0^{\prime}:=\mu(1-\epsilon) \min \left\{2\left[\kappa L\left(18 \kappa^2+1\right)\right]^{-1}, \frac{1}{3} L^{-1}\right\} \\
C_1^{\prime}:=\left(1-\epsilon^2\right) \mu^{-1} \underline{\lambda} C_1 \\
C_2^{\prime}:=2\left(L^{-1}+1\right) \underline{\lambda}^{-1} \mu^{-2} L L^{\mathrm{msq}}\left(L^{\mathrm{msq}}+L_h^2\right)
\end{gathered}
$$

If $L=\bar{\lambda}$ and $\underline{\lambda} \leqslant \mu$, then $C_0^{\prime}=(1-\epsilon) C_0, C_1^{\prime}=(1-\epsilon^2) \mu^{-1} \underline{\lambda} C_1, C_2^{\prime}=C_2$. Note also that in the second case, $\alpha \leqslant 1$, and in the first case $\alpha \leqslant L^{-1} \underline{\lambda} \leqslant \kappa^{-1}$, then the convergence rate of the first case can be faster as long as $\epsilon$ is sufficiently small.

          \end{proof}
\vspace{-10pt}
          \begin{proof}[Proof of Lemma \ref{lem:Lyapunov_contraction}]
            If $\tau = 0$, we get $V_t \leqslant \left( {1 - \alpha C_0}  \right)^t V_0$
            immediately.
            
            Now, we suppose that $\tau \geqslant 1$. Assume that $V_s/V_0 \leqslant v^s , s
            \leqslant t$ for some $0 < v < 1$, then
            \begin{eqnarray*}
              V_{t + 1}/V_0 & \leqslant & (1 - C_0 \alpha) v^t + C_1 \alpha^3 \tau
              \overline{\tau} v^{t - 3 \tau} +  C_2 \alpha^2 \tau \overline{\tau} v^{t -
              2 \tau}\\
              & = & v^t \left[ 1 - C_0 \alpha + C_1 \alpha^3 \tau \overline{\tau} v^{-
              3 \tau} +  C_2 \alpha^2 \tau \overline{\tau} v^{- 2 \tau} \right]
            \end{eqnarray*}
            Define function
            \[ \chi (\alpha, v) \:= 1 - C_0 \alpha + C_1 \alpha^3 \tau
               \overline{\tau} v^{- 3 \tau} +  C_2 \alpha^2 \tau \overline{\tau} v^{- 2
               \tau} \]
            For each fixed $\alpha$, $\chi (v, \alpha)$ is decreasing with $v$, and
            \[ \lim_{v \rightarrow 0} \chi (v, \alpha) = \infty, \lim_{v \rightarrow 1}
               \chi (v, \alpha) = 1 - C_0 \alpha + C_1 \alpha^3 \tau \overline{\tau} + 
               C_2 \alpha^2 \tau \overline{\tau} . \]
            Thus, the set
            \[ \mathcal{V} (\alpha) : = \{ v \in (0, 1), \chi (v, \alpha) \leqslant v
               \} \]
            is not empty as long as $\chi (1, \alpha) < 1$, namely, $C_1 \alpha^2 \tau
            \overline{\tau} +  C_2 \alpha  \tau \overline{\tau} < C_0$, which is
            equivalent to
            \[ \alpha < \alpha_{\max} \:= \frac{2 C_0}{C_2 \tau \overline{\tau} +
               \sqrt{\left( C_2 \tau \overline{\tau} \right)^2 + 4 C_1 C_0 \tau
               \overline{\tau}}}  \]
            Here, $\text{} \alpha_{\max} $ solves $C_1 \alpha^2 \tau \overline{\tau} + 
            C_2 \alpha  \tau \overline{\tau} = C_0$. Therefore, for $\alpha \leqslant
            \alpha_{\max}$ and $v \in \mathcal{V} (\alpha)$, we have
            \[ V_{t + 1}/V_0 \leqslant v^t \chi (v, \alpha) \leqslant v^{t + 1} \]
            We can simplify $\alpha_{\max}$ through the following equation:
            \[ \frac{2 C_0}{C_2 \tau \overline{\tau} + \sqrt{\left( C_2 \tau
               \overline{\tau} \right)^2 + 4 C_1 C_0 \tau \overline{\tau}}} > \frac{2
               C_0}{C_2 \tau \overline{\tau} + C_2 \tau \overline{\tau} + 2 \sqrt{C_1
               C_0 \tau \overline{\tau}}} > \frac{C_0}{C_2 + \sqrt{C_1 C_0}}
               \frac{1}{\tau \overline{\tau}} \]
            Next, we find the optimal $v \in \mathcal{V} (\alpha)$ for fixed $\alpha <
            \alpha_{\max}$, which is
            $ v^{\star} (\alpha) : = \min  \{ \mathcal{V} (\alpha) \}, $
            Let $v = 1 - \xi$, then, $\chi (v, \alpha) \leqslant v$ gives
            \begin{equation}\label{eq:lem_lhs}
                \xi + C_1 \alpha^3 \tau \overline{\tau}  \left( 1 + \frac{\xi}{1 - \xi}
               \right)^{3 \tau} +  C_2 \alpha^2 \tau \overline{\tau}  \left( 1 +
               \frac{\xi}{1 - \xi} \right)^{2 \tau} \leqslant C_0 \alpha
            \end{equation}
            Note that for any $a, x > 0$
            \[ (1 + a)^x \leqslant e^{a x} \leqslant 1 + (e - 1) a x < 1 + 2 a x \]
            as long as $a x \leqslant 1$. Here, the first inequality is due to $(1 + a)
            \leqslant e^a$, and the second inequality is due to that $\frac{e^x - 1}{x}$
            is increaseing in for positive $x$, which implies that $\frac{e^{a x} - 1}{a
            x} \leqslant \frac{e^1 - 1}{1}$ if $a x \leqslant 1.$
            
            Thus, whenever $3 \tau \xi (1 - \xi)^{- 1} \leqslant 1$, namely $(3 \tau +
            1) \xi \leqslant 1$,
            \[ \left( 1 + \frac{\xi}{1 - \xi} \right)^{3 \tau} \leqslant 1 + \frac{6
               \tau \xi}{1 - \xi} \leqslant 1 + 8 \tau \xi, \left( 1 + \frac{\xi}{1 -
               \xi} \right)^{2 \tau} < 1 + \frac{4 \tau \xi}{1 - \xi} \leqslant 1 + 6
               \tau \xi \]
            In this case, this implies that the left-hand side of Equation \eqref{eq:lem_lhs} is less
            than $\xi + C_1 \alpha^3 \tau \overline{\tau}  (1 + 8 \tau \xi) +  C_2 \alpha^2
               \tau \overline{\tau} (1 + 6 \tau \xi),$ 
            which further implies that it suffices that
            \[ \left( 1 + 8 C_1 \alpha^3 \tau  \overline{\tau} \tau + 6 C_2 \alpha^2
               \tau \overline{\tau} \tau \right) \xi \leqslant C_0 \alpha - C_1 \alpha^3
               \tau \overline{\tau} - C_2 \alpha^2 \tau \overline{\tau}, \]
            namely,
            \[ \xi \leqslant \xi_{\max} \:= \frac{C_0 - C_1 \alpha^2 \tau
               \overline{\tau} - C_2 \alpha  \tau \overline{\tau}}{1 + 8 C_1 \alpha^3
               \tau  \overline{\tau} \tau + 6 C_2 \alpha^2 \tau \overline{\tau} \tau}
               \alpha \]
            In summary, when $\alpha < \min \left\{ \alpha_{\max}, \overline{\alpha}
            \right\}$, $\xi_{\max}$ is positive and
            \[ V_t/V_0 \leqslant (1 - \min \{ \xi_{\max}, (3 \tau + 1)^{- 1} \})^t . \]

    Note that
  $ \alpha_{\max} \geqslant C_0 \left( C_2 + \sqrt{C_1 C_0} \right)^{- 1}
     \frac{1}{\tau \overline{\tau}}, \geqslant A_0 \overline{\tau}^{- 1} \tau^{- 1} \geqslant
     A_0 \left( \overline{\tau} + 1 \right)^{- 1} (\tau + 1)^{- 1} .$ Thus, if $\alpha = \beta A_0 \left( \overline{\tau} + 1 \right)^{- 1} (\tau
  + 1)^{- 1}$, then $\alpha \leqslant \beta \alpha_{\max}$, which implies,
  \[ C_1 \alpha^2 \tau \overline{\tau} + C_2 \alpha  \tau \overline{\tau} <
     \beta \left[ C_1 \alpha^2_{\max} \tau \overline{\tau} + C_2 \alpha_{\max}
     \tau \overline{\tau} \right] = \beta C_0 \]
  Thus,
  \[ C_0 - C_1 \alpha^2 \tau \overline{\tau} - C_2 \alpha  \tau
     \overline{\tau} > (1 - \beta) C_0, 8 C_1 \alpha^3 \tau  \overline{\tau}
     \tau + 6 C_2 \alpha^2 \tau \overline{\tau} \tau < 8 \beta C_0  \alpha
     (\tau + 1) . \]
  From this, we have
  \[ \xi_{\max} > (\tau + 1)^{- 1} \frac{(1 - \beta) C_0 \alpha (\tau + 1)}{1
     + 8 \beta C_0 \alpha  (\tau + 1)} \geqslant (\tau + 1)^{- 1} \left(
     \overline{\tau} + 1 \right)^{- 1} \frac{(1 - \beta) \beta }{C_0^{- 1}
     A_0^{- 1} + 8 \beta^2} \]
  Moreover, since (note that $C_0  \leqslant 1$)
  \[ \frac{(1 - \beta) \beta }{C_0^{- 1} A_0^{- 1} + 8 \beta^2} < \frac{1}{4},
  \]
  we have
  \[ \frac{1}{3 \tau + 1} \geqslant (\tau + 1)^{- 1} \frac{1}{4} > (\tau +
     1)^{- 1} \left( \overline{\tau} + 1 \right)^{- 1} \frac{(1 - \beta) \beta
     }{C_0^{- 1} A_0^{- 1} + 8 \beta^2} . \]
  Therefore,
  \[ \min \{ \xi_{\max}, (3 \tau + 1)^{- 1} \} \geqslant (\tau + 1)^{- 1}
     \left( \overline{\tau} + 1 \right)^{- 1} \frac{(1 - \beta) \beta }{C_0^{-
     1} A_0^{- 1} + 8 \beta^2}, \]
     which completes the proof.
          \end{proof}

          \begin{proof}[Proof of Corollary \ref{cor:improved_asynchronous_algorithm}]
            Let $T_C = T_c$. According to the proof of and the convexity of $f$, we
            have, for $t \geq T_C$ and \ some constants ${C_0} , C_1, C_0$ independent
            of staleness.
            \[ V_{t + 1} \leqslant \frac{\tilde{V}_{t + 1} + \omega   V_t + \cdots +
               \omega ^M V_{t - M}}{1 + \omega   + \cdots + \omega ^M} \]
            \[ \tilde{V}_{t + 1} \leqslant \left( {1 - \alpha C_0}  \right) V_t +
               \alpha^3 \overline{\tau} \tau C_1 \max_{t \in [t - 3 \tau, t]} V_s +
               \alpha^2 \overline{\tau} \tau C_2 \max_{t \in [t - 2 \tau, t]} V_s \]
            Now suppose that $V_s / V_0 \leqslant c^s$ for some $0 < c < 1$ and $s
            \leqslant t$, then
            \[ V_{t + 1} / V_0 \leqslant c^{t + 1} c^{- M - 1} \gamma \]
            where (recall $\chi (\alpha, c) : = \left( {1 - \alpha C_0}  \right) +
            \alpha^3 \overline{\tau} \tau C_1 c^{- 3 \tau} + \alpha^2 \overline{\tau}
            \tau C_1 c^{- 2 \tau}$)
            \[ \gamma := \frac{\chi (\alpha, c) + \omega   + \cdots + \omega ^M}{1
               + \omega   + \cdots + \omega ^M}, \]
            For given $\alpha$, let $v$ and $c (v)$ be
            \[ \chi (\alpha, v) \leqslant v, c (v) : = {\left( \frac{v + \omega   +
               \cdots + \omega ^M}{1 + \omega   + \cdots + \omega ^M} \right)^{1 / (M +
               1)}}  \]
            Then, $c (v) \geqslant \gamma^{1 / (M + 1)}$, we have $V_{t + 1} / V_0
            \leqslant c^{t + 1}$. Since if $v = 1 - \xi$, we have
            \[ {\left( \frac{v + \omega   + \cdots + \omega ^M}{1 + \omega   + \cdots +
               \omega ^M} \right)^{1 / (M + 1)}}  \leqslant 1 - \frac{\xi}{(M + 1) (1 +
               \omega   + \cdots + \omega ^M)} \]
            The left-hand side term of the equation bounds the convergence rate. 
            The proof is done. 
          \end{proof}
      \subsubsection{Technical Lemmas}
          The following lemmas are used in the proof of Theorem.
          \begin{lemma}\label{lem:young_cauchy}
            For any $0<\lambda<1$ and vectors $\boldsymbol{a},\boldsymbol{b}$ of the same dimension,
            \[
            2\lvert \langle \boldsymbol{a},\boldsymbol{b}\rangle \rvert
            \;\le\; \lambda \|\boldsymbol{a}\|^2 + \lambda^{-1}\|\boldsymbol{b}\|^2,
            \qquad
            \|\boldsymbol{a}+\boldsymbol{b}\|^2
            \;\ge\; (1-\lambda)\|\boldsymbol{a}\|^2 - (\lambda^{-1}-1)\|\boldsymbol{b}\|^2 .
            \]
          \end{lemma}

          \begin{proof}
            By Young's inequality  and Cauchy--Schwarz inequality,
            $
            2\lvert \langle \boldsymbol{a},\boldsymbol{b}\rangle \rvert
            \le \lambda \|\boldsymbol{a}\|^2 + \lambda^{-1}\|\boldsymbol{b}\|^2,
            $
            which proves the first claim. For the second, expand and then lower-bound the cross term using
            $2\langle \boldsymbol{a},\boldsymbol{b}\rangle \ge -\big(\lambda \|\boldsymbol{a}\|^2 + \lambda^{-1}\|\boldsymbol{b}\|^2\big)$:
            \[
            \|\boldsymbol{a}+\boldsymbol{b}\|^2
            = \|\boldsymbol{a}\|^2 + \|\boldsymbol{b}\|^2 + 2\langle \boldsymbol{a},\boldsymbol{b}\rangle
            \;\ge\; (1-\lambda)\|\boldsymbol{a}\|^2 + \big(1-\lambda^{-1}\big)\|\boldsymbol{b}\|^2
            = (1-\lambda)\|\boldsymbol{a}\|^2 - (\lambda^{-1}-1)\|\boldsymbol{b}\|^2.
            \]
          \end{proof}
\vspace{-10pt}

          \begin{lemma}\label{lem:quadratic_surrogate}
            Let $\boldsymbol{a},\boldsymbol{b}$ be vectors of the same dimension and $C>0$.
            If $\alpha \leq C^{-1}$, then
            \[
            \alpha \langle \boldsymbol{a}, \boldsymbol{b} \rangle - \tfrac{C}{2}\alpha^2 \| \boldsymbol{b} \|^2
            \;\leq\; \tfrac{\alpha}{2}\Big( \|\boldsymbol{a}\|^2 - \|\boldsymbol{b}-\boldsymbol{a}\|^2 \Big).
            \]
            Moreover, for any positive definite matrix $\boldsymbol{H}$, the same bound holds with respect to the $\boldsymbol{H}$-inner product:
            \[
            \alpha \langle \boldsymbol{a}, \boldsymbol{b} \rangle_{\boldsymbol{H}}
            - \tfrac{C}{2}\alpha^2 \|\boldsymbol{b}\|_{\boldsymbol{H}}^2
            \;\leq\; \tfrac{\alpha}{2}\Big( \|\boldsymbol{a}\|_{\boldsymbol{H}}^2 - \|\boldsymbol{b}-\boldsymbol{a}\|_{\boldsymbol{H}}^2 \Big).
            \]
          \end{lemma}

          \begin{proof}
            Expanding $\|\boldsymbol{b}\|^2 = \|\boldsymbol{b}-\boldsymbol{a}\|^2 + 2\langle \boldsymbol{a},\boldsymbol{b}-\boldsymbol{a}\rangle + \|\boldsymbol{a}\|^2$, we obtain
            \begin{align*}
              \alpha \langle \boldsymbol{a}, \boldsymbol{b} \rangle - \tfrac{C}{2}\alpha^2 \|\boldsymbol{b}\|^2
               & = \Big(\alpha - \tfrac{C}{2}\alpha^2\Big)\|\boldsymbol{a}\|^2
              - \tfrac{C}{2}\alpha^2 \|\boldsymbol{b}-\boldsymbol{a}\|^2
              + (\alpha - C\alpha^2)\langle \boldsymbol{a}, \boldsymbol{b}-\boldsymbol{a}\rangle.
            \end{align*}
            Applying Lemma~\ref{lem:young_cauchy} with parameter $\lambda=\tfrac{1}{2}$ to bound the last term and using $\alpha \le C^{-1}$, we find
            \[
            \alpha \langle \boldsymbol{a}, \boldsymbol{b} \rangle - \tfrac{C}{2}\alpha^2 \|\boldsymbol{b}\|^2
            \;\le\; \tfrac{\alpha}{2}\|\boldsymbol{a}\|^2 - \tfrac{\alpha}{2}\|\boldsymbol{b}-\boldsymbol{a}\|^2,
            \]
            which is the desired inequality. The extension to the $\boldsymbol{H}$-inner product follows by replacing norms and inner products accordingly.
          \end{proof}
\vspace{-10pt}
          \begin{lemma}\label{lem:cubic_norm}
            For any two vectors $\boldsymbol{a}$ and $\boldsymbol{b}$ of the same dimension,
            $
            \|\boldsymbol{a} + \boldsymbol{b}\|^3 \leq 4 \bigl( \|\boldsymbol{a}\|^3 + \|\boldsymbol{b}\|^3 \bigr).
            $
            Similarly, for any positive definite matrix $\boldsymbol{H}$, the same bound holds with respect to the $\boldsymbol{H}$-norm:
            \[
                \|\boldsymbol{a} + \boldsymbol{b}\|^3_{\boldsymbol{H}} \leq 4 \bigl( \|\boldsymbol{a}\|^3_{\boldsymbol{H}} + \|\boldsymbol{b}\|^3_{\boldsymbol{H}} \bigr).
                \]
          \end{lemma}
          \begin{proof}
            By the triangle inequality,
            $
            \|\boldsymbol{a} + \boldsymbol{b}\| \leq \|\boldsymbol{a}\| + \|\boldsymbol{b}\|
            $,
            so
            $
            \|\boldsymbol{a} + \boldsymbol{b}\|^3 \leq (\|\boldsymbol{a}\| + \|\boldsymbol{b}\|)^3.
            $
            Expanding and using the inequality $3\|\boldsymbol{a}\|\|\boldsymbol{b}\|(\|\boldsymbol{a}\| + \|\boldsymbol{b}\|) \leq 3(\|\boldsymbol{a}\|^3 + \|\boldsymbol{b}\|^3)$ gives
            $(\|\boldsymbol{a}\| + \|\boldsymbol{b}\|)^3 \leq 4(\|\boldsymbol{a}\|^3 + \|\boldsymbol{b}\|^3)$,
            which completes the proof. The extension to the $\boldsymbol{H}$-norm follows by replacing norms accordingly.
          \end{proof}

          \begin{lemma}\label{lem:linear_combination_lower_bound}
            Suppose $a, b, c, \beta, \varsigma, \alpha_1, \alpha_2$ are all positive numbers. If
            $a \geq \beta b - \varsigma c,$
            then
            \[
            \alpha_1 a + \alpha_2 c \geq
            \begin{cases}
              \frac{\alpha_2 \beta}{2 (\varsigma + \beta)} (b + c), & \text{if } \frac{\alpha_2}{2} \leq \alpha_1 \varsigma,\\[2mm]
              \min\{\alpha_1 \beta, \frac{\alpha_2}{2}\} (b + c), & \text{otherwise}.
            \end{cases}
            \]
          \end{lemma}

          \begin{proof}
            Suppose that $\alpha_2 \leq 2 \alpha_1 \varsigma$. There are two cases.

            \textbf{Case 1:} $\beta b - \varsigma c \leq 0$. Then $\frac{\beta}{\varsigma + \beta}(b + c) \leq c$, so
            $\alpha_1 a + \alpha_2 c \geq \alpha_2 c \geq \frac{\alpha_2 \beta}{\varsigma + \beta} (b + c)$,

            \textbf{Case 2:} $\beta b - \varsigma c > 0$. Then $\frac{\varsigma}{\varsigma + \beta} (b+c) < b$ and $\frac{\beta}{\varsigma + \beta} (b+c) > c$, so
            \[
            \begin{aligned}
              \alpha_1 a + \alpha_2 c & \geq \alpha_1 \beta b - (\alpha_1 \varsigma - \alpha_2)_+ c \\
               & \geq
              \begin{cases}
                \frac{\alpha_2 \beta}{\varsigma + \beta} (b+c), & \text{if } \alpha_2 \leq \alpha_1 \varsigma,\\
                \frac{\alpha_1 \beta \varsigma}{\varsigma + \beta} (b+c), & \text{if } \alpha_1 \varsigma \leq\alpha_2 \leq 2\alpha_1 \varsigma.
              \end{cases}
            \end{aligned}
            \]

            In both cases, we have
$            \alpha_1 a + \alpha_2 c \geq \frac{\alpha_2 \beta}{2(\varsigma + \beta)} (b + c).$
            Now suppose that $\alpha_2 > 2 \alpha_1 \varsigma$, then the following inequalities complete the proof
            \[
            \alpha_1 a + \alpha_2 c \geq \alpha_1 \beta b + (\alpha_2 - \alpha_1 \varsigma) c \geq \min\{\alpha_1 \beta, \alpha_2 - \alpha_1 \varsigma\} (b+c)\geq \min\{\alpha_1 \beta, \frac{\alpha_2}{2}\} (b+c).
            \]
            \end{proof}

\vspace{-10pt}
          \begin{lemma}\label{lem:gradient_gap_inequality_upper}
            Suppose $F : \mathcal{X} \to \mathbb{R}$ is $\mu_F$-strongly convex with minimizer $\boldsymbol{x}^\ast$ over a convex set $\mathcal{X}$. Then for any $\boldsymbol{x} \in \mathcal{X}$,
           $ 
            F(\boldsymbol{x}) - F(\boldsymbol{x}^\ast) \leq \frac{1}{2 \mu_F} \|\nabla F(\boldsymbol{x})\|^2.
            $
          \end{lemma}

          \begin{proof}
            By strong convexity of $F$, for any $\boldsymbol{x}$ we have
            $
            F(\boldsymbol{x}^\ast) \geq F(\boldsymbol{x}) + \langle \nabla F(\boldsymbol{x}), \boldsymbol{x}^\ast - \boldsymbol{x} \rangle + \frac{\mu_F}{2} \|\boldsymbol{x}^\ast - \boldsymbol{x}\|^2.
            $
            Rearranging gives
            \begin{equation}\label{eq:gap_inequality_proof}
              F(\boldsymbol{x}) - F(\boldsymbol{x}^\ast) \leq \langle \nabla F(\boldsymbol{x}), \boldsymbol{x} - \boldsymbol{x}^\ast \rangle - \frac{\mu_F}{2} \|\boldsymbol{x} - \boldsymbol{x}^\ast\|^2.
            \end{equation}
            Consider the right-hand side of (\ref{eq:gap_inequality_proof}) as a function of $\boldsymbol{h} := \boldsymbol{x} - \boldsymbol{x}^\ast$:
            $
            \phi(\boldsymbol{h}) := \langle \nabla F(\boldsymbol{x}), \boldsymbol{h} \rangle - \frac{\mu_F}{2} \|\boldsymbol{h}\|^2.
            $
            This is a concave quadratic in $\boldsymbol{h}$, which achieves its maximum at $\boldsymbol{h} = \frac{1}{\mu_F} \nabla F(\boldsymbol{x})$, with value
            $
            \max_{\boldsymbol{h}} \phi(\boldsymbol{h}) = \frac{1}{2 \mu_F} \|\nabla F(\boldsymbol{x})\|^2.
            $
 Hence,
            $
            F(\boldsymbol{x}) - F(\boldsymbol{x}^\ast) \leq \frac{1}{2 \mu_F} \|\nabla F(\boldsymbol{x})\|^2.
            $
          \end{proof}

\vspace{-10pt}
          \begin{lemma}\label{lem:gradient_gap_inequality_lower}
            Suppose that $F : \mathcal{X} \to \mathbb{R}$ is $\mu_F$-strongly convex and $L_F$-smooth over a convex set $\mathcal{X}$, with condition number $\kappa = L_F / \mu_F$. Then for any $\boldsymbol{x} \in \mathcal{X}$,
            $
            F(\boldsymbol{x}) - F(\boldsymbol{x}^\ast) \geq \frac{1}{2 \kappa L_F} \|\nabla F(\boldsymbol{x})\|^2.
            $
          \end{lemma}

          The proof of Lemma~\ref{lem:gradient_gap_inequality_lower} can be found in \cite{shi2025decentralized}, and is omitted here.

          \begin{lemma}\label{lem:block_descent_inequality_upper}
            Assume that $F: \mathcal{X} \to \mathbb{R}$ is $L_F$-smooth. Then for any $\boldsymbol{x}, \boldsymbol{y} \in \mathbb{R}^d$,
            \[
            F(\boldsymbol{y}) - F(\boldsymbol{x}) \leq \langle \nabla F(\boldsymbol{x}), \boldsymbol{y} - \boldsymbol{x} \rangle + \frac{L_F}{2} \|\boldsymbol{y} - \boldsymbol{x}\|^2.
            \]
          \end{lemma}

          \begin{lemma}\label{lem:block_descent_inequality_Hessian}
            Assume that $F: \mathcal{X} \to \mathbb{R}$ is three times continuously differentiable and its Hessian $\nabla^2 F$ is $H_F$-Lipschitz. Then for any $\boldsymbol{x}, \boldsymbol{y} \in \mathcal{X}$,
            \begin{equation*}\label{eq:block_descent_inequality_Hessian}
              F(\boldsymbol{y}) - F(\boldsymbol{x}) \leq \langle \nabla F(\boldsymbol{x}), \boldsymbol{y} - \boldsymbol{x} \rangle + \frac{1}{2} \langle \boldsymbol{y} - \boldsymbol{x}, \nabla^2 F(\boldsymbol{x}) (\boldsymbol{y} - \boldsymbol{x}) \rangle + \frac{1}{6} H_F \|\boldsymbol{y} - \boldsymbol{x}\|^3.
            \end{equation*}
          \end{lemma}

          \begin{remark}
            Let $\boldsymbol{H}_F := \nabla^2 F(\boldsymbol{x})$. Then the Hessian version inequality can be expressed as
            \[
            F(\boldsymbol{y}) - F(\boldsymbol{x}) \leq \langle \boldsymbol{H}_F^{-1} \nabla F(\boldsymbol{x}), \boldsymbol{y} - \boldsymbol{x} \rangle_{\boldsymbol{H}_F} + \frac{1}{2} \|\boldsymbol{y} - \boldsymbol{x}\|^2_{\boldsymbol{H}_F} + \frac{1}{6} H_F \|\boldsymbol{y} - \boldsymbol{x}\|^3.
            \]
            Moreover, since
            \[
            \|\boldsymbol{y} - \boldsymbol{x}\|^2 = (\boldsymbol{H}_F^{1/2} (\boldsymbol{y} - \boldsymbol{x}))^\top \boldsymbol{H}_F^{-1} \boldsymbol{H}_F^{1/2} (\boldsymbol{y} - \boldsymbol{x}) \leq \lambda_{\min}^{-1}(\boldsymbol{H}_F) \|\boldsymbol{y} - \boldsymbol{x}\|^2_{\boldsymbol{H}_F},
            \]
            we can further have the following inequality
            \begin{equation}\label{eq:block_descent_inequality_Hessian2}
              F(\boldsymbol{y}) - F(\boldsymbol{x}) \leq \langle \boldsymbol{H}_F^{-1} \nabla F(\boldsymbol{x}), \boldsymbol{y} - \boldsymbol{x} \rangle_{\boldsymbol{H}_F} + \frac{1}{2} \|\boldsymbol{y} - \boldsymbol{x}\|^2_{\boldsymbol{H}_F} + \frac{1}{6} H_F \lambda_{\min}^{-3/2}(\boldsymbol{H}_F) \|\boldsymbol{y} - \boldsymbol{x}\|^3_{\boldsymbol{H}_F}.
            \end{equation}

          \end{remark}

\vspace{-10pt}
          \begin{lemma}\label{lem:block_minimizer_sensitivity_2}
            Suppose each $f_j$ is $L_j$-smooth and $h$ is $L_h$-smooth. Assume further that the global $ f:= \frac{1}{J} \sum_{j=1}^J f_j + h$ is $\mu$-strongly convex.  For sequences $(\boldsymbol{x}_2^j)_{j=0}^J$ and $(\boldsymbol{y}_2^j)_{j=0}^J$, define
            \[
            \boldsymbol{x}_1^{\ast} := \arg\min_{\boldsymbol{x}_1} \Big\{ \frac{1}{J} \sum_{j=1}^J f_j(\boldsymbol{x}_1, \boldsymbol{x}_2^j) + h(\boldsymbol{x}_1, \boldsymbol{x}_2^0) \Big\}, \quad
            \boldsymbol{y}_1^{\ast} := \arg\min_{\boldsymbol{y}_1} \Big\{ \frac{1}{J} \sum_{j=1}^J f_j(\boldsymbol{y}_1, \boldsymbol{y}_2^j) + h(\boldsymbol{y}_1, \boldsymbol{y}_2^0) \Big\},
            \]
            then, with $\overline{L} := J^{-1}\sum_{j=1}^J L_j$,
            \[
            \|\boldsymbol{x}_1^{\ast} - \boldsymbol{y}_1^{\ast}\|
            \le \frac{1}{\mu} \cdot \frac{1}{J} \sum_{j=0}^J
             L_j \Big(\|\boldsymbol{x}_2^j - \boldsymbol{x}_2^0\|
            +  \|\boldsymbol{y}_2^j- \boldsymbol{y}_2^0\| \Big)
            + \frac{1}{\mu} (L_h+\overline{L}) \|\boldsymbol{x}_2^0 - \boldsymbol{y}_2^0\|.
            \]
          \end{lemma}

          \begin{proof}
            Due to the strong convexity, we have
            $
            \mu\|\boldsymbol{y}_1^{\ast} - \boldsymbol{x}_1^{\ast}\| \leq \|\nabla_1 f(\boldsymbol{y}_1^{\ast}, \boldsymbol{x}_2^0) - \nabla_1 f(\boldsymbol{x}_1^{\ast}, \boldsymbol{x}_2^0)\|.
            $
            Bounding the right-hand side by the triangle inequality gives
            \begin{equation}\label{eq:gradient_bound_1}
                \|\nabla_1 f(\boldsymbol{y}_1^{\ast}, \boldsymbol{x}_2^0) - \nabla_1 f(\boldsymbol{x}_1^{\ast}, \boldsymbol{x}_2^0)\| \leq \|\nabla_1 f(\boldsymbol{y}_1^{\ast}, \boldsymbol{y}_2^0)\| + \|\nabla_1 f(\boldsymbol{y}_1^{\ast}, \boldsymbol{x}_2^0) - \nabla_1 f(\boldsymbol{y}_1^{\ast}, \boldsymbol{y}_2^0)\| + \|\nabla_1 f(\boldsymbol{x}_1^{\ast}, \boldsymbol{x}_2^0)\|.
            \end{equation}
            Smoothness of $f_j$ and $h$ gives
            \begin{equation}\label{eq:gradient_bound_2}
            \|\nabla_1 f(\boldsymbol{y}_1^{\ast}, \boldsymbol{x}_2^0) - \nabla_1 f(\boldsymbol{y}_1^{\ast}, \boldsymbol{y}_2^0)\| \leq (L_h+\overline{L}) \|\boldsymbol{x}_2^0 - \boldsymbol{y}_2^0\|.
            \end{equation}
            By the optimal condition, 
            \[
            \|\nabla_1 f(\boldsymbol{y}_1^{\ast},\boldsymbol{y}_2^{0})\|
            = \Bigg\|
            \nabla_1 f(\boldsymbol{y}_1^{\ast},\boldsymbol{y}_2^{0})
            -\Bigg\{\frac{1}{J}\sum_{j=1}^J \nabla_1 f_j(\boldsymbol{y}_1^{\ast},\boldsymbol{y}_2^{j})
            +\nabla_1 h(\boldsymbol{y}_1^{\ast},\boldsymbol{y}_2^{0})\Bigg\}
            \Bigg\|
            \]
            which, according to smoothness of \(f_j\), are further bounded by
            \begin{equation}\label{eq:gradient_bound_3}
            \|\nabla_1 f(\boldsymbol{y}_1^{\ast},\boldsymbol{y}_2^{0})\|
            \;\le\; \frac{1}{J}\sum_{j=1}^J L_j \|\boldsymbol{y}_2^{j}-\boldsymbol{y}_2^{0}\|.
            \end{equation}
            Similarly,
            \begin{equation}\label{eq:gradient_bound_4}
            \|\nabla_1 f(\boldsymbol{x}_1^{\ast}, \boldsymbol{x}_2^0)\| \leq \frac{1}{J}\sum_{j=1}^J L_j \|\boldsymbol{x}_2^j-\boldsymbol{x}_2^0\|.
            \end{equation}
            Combining the Equations \ref{eq:gradient_bound_1}--\ref{eq:gradient_bound_4} gives the stated result.

          \end{proof}

  \bibliography{references}  % Ensure your .bib file name matches

\end{document}